\documentclass[11pt,a4paper]{article}

\usepackage[utf8]{inputenc}
\usepackage{smallstyle}

\hypersetup{
  pdftitle={Balanced Allocations in Batches: Simplified and Generalized},
  pdfauthor={Dimitrios Los, Thomas Sauerwald}
  }

\begin{document}

\title{Balanced Allocations in Batches:  Simplified and Generalized}

\author[]{Dimitrios Los\thanks{\texttt{dimitrios.los@cl.cam.ac.uk}} }
\author[]{Thomas Sauerwald\thanks{\texttt{thomas.sauerwald@cl.cam.ac.uk}}}
\affil[]{Department of Computer Science \& Technology, University of Cambridge}

\maketitle

\begin{abstract}
We consider the allocation of $m$ balls (jobs) into $n$ bins (servers). In the \TwoChoice process, for each of $m$ sequentially arriving balls, two randomly chosen bins are sampled and the ball is placed in the least loaded bin. It is well-known that the maximum load is $m/n+\log_2 \log n + \Oh(1)$ with high probability.  

Berenbrink, Czumaj, Englert, Friedetzky and Nagel~\cite{BCE12} introduced a parallel version of this process, where $m$ balls arrive in consecutive batches of size $b=n$ each. Balls within the same batch are allocated in parallel, using the load information of the bins at the beginning of the batch. They proved that the gap of this process is $\Oh(\log n)$ with high probability. 

In this work, we present a new analysis of this setting, which is based on exponential potential functions. This allows us to both simplify and generalize the analysis of~\cite{BCE12} in different ways:
\begin{enumerate}
    \item Our analysis covers a broad class of processes. This includes not only \TwoChoice, but also processes with fewer bin samples like $(1+\beta)$, processes which can only receive one bit of information from each bin sample and graphical allocation, where bins correspond to vertices in a graph.
    \item Balls may be of different weights, as long as their weights are independent samples from a distribution satisfying a technical condition on its moment generating function. %
    \item For arbitrary batch sizes $b \geq n$, we prove a gap of $\Oh(b/n \cdot \log n)$. For any $b \in [n , n^3]$, we improve this to $\Oh(b/n + \log n)$ and show that it is tight for a family of processes. This implies the unexpected result that for e.g. $(1+\beta)$ with constant $\beta \in (0, 1]$, the gap is $\Theta(\log n)$ for all $b \in [n,n \log n]$.
\end{enumerate}
We also conduct experiments which support our theoretical results, and even hint at a superiority of less powerful processes like $(1+\beta)$ for large batch sizes.

\end{abstract}

\clearpage

\clearpage
\tableofcontents
~
\clearpage

\section{Introduction}

\textbf{Motivation.} We study the classical problem of allocating $m$ balls (jobs) into $n$ bins (servers). This  framework also known as balls-into-bins or balanced allocations~\cite{ABK99}~is a popular abstraction for various resource allocation and storage problems such as load balancing, scheduling or hashing (see surveys~\cite{MRS01,W17}). Following a long line of previous works, we consider randomized allocation schemes where for each ball we take a certain number of bin samples and then allocate the ball into one of these samples.

For the simplest allocation scheme, called \OneChoice, each 
of the $m$ balls is placed in a random bin chosen independently and uniformly. It is well-known that the maximum load is $\Theta( \log n / \log \log n)$ \Whp\footnote{In general, with high probability refers to probability of at least $1 - n^{-c}$ for some constant $c > 0$.} for $m=n$, and $m/n + \Theta( \sqrt{ (m/n) \log n})$ \Whp~for $m \gg n$. While this allocation scheme can be of course executed completely \emph{in parallel}, it results in a significantly large gap between the maximum load and average load if $m$ gets large. 

Azar, Broder, Karlin and Upfal~\cite{ABK99} (and implicitly Karp, Luby and Meyer auf der Heide~\cite{KLM96}) proved that if the $m$ balls are allocated \emph{sequentially}, and each ball is placed in the lesser loaded of $d \geq 2$ randomly chosen bins, then the maximum load drops to $\log_d \log n + \Oh(1)$ \Whp, if $m=n$. This dramatic improvement from $d=1$ (\OneChoice) to $d=2$ (\TwoChoice) is known as ``power of two choices'', and similar effects have been observed in other problems including routing, hashing and randomized rounding~\cite{MRS01}. V\"ocking~\cite{V03} proved that further improvements on the gap bound (which are more significant for larger $d$) are possible if one employs an asymmetric tie-breaking rule.

Berenbrink, Czumaj, Steger and V\"ocking~\cite{BCSV06} extended the analysis of~\cite{ABK99} (and~\cite{V03}) to the so-called \emph{heavily loaded case}, where $m \geq n$ can be arbitrarily large. In particular, for \TwoChoice an upper bound on the gap (the difference between the maximum and average load) of $ \log_2 \log n + \Oh(1)$ \Whp~was shown. A simpler but slightly weaker analysis was later given by Talwar and Wieder~\cite{TW14}.

The above studies for \DChoice, as well as many other works in balanced allocations, usually make the following assumptions: 
\begin{enumerate}
\item[1.)] All $m$ balls have to be allocated sequentially, with the load information being  updated immediately.
\item[2.)] All $m$ balls are of the same weight.
\item[3.)] All $m$ balls need to take $d \geq 2$ independent and uniform bin samples.
\end{enumerate}

While these assumptions are crucial in many of the mathematical analyses, they may be difficult to satisfy in practical applications. For example, in a highly parallel environment, the load information of a bin may not include the most recent allocations. Further, processing times of jobs (size of data items) may not be identical but rather follow a heterogeneous distribution, which leads to the so-called \emph{weighted balls-into-bins} setting. Finally, \TwoChoice (and \DChoice) assume that balls are able to sample two (or $d$) bins which are chosen independently and uniformly at random, in every round. It is natural to consider scenarios where the $d$ samples are correlated (e.g., through a network structure), or for some balls only one sample is available.

\textbf{Related Work.} In order to relax assumption $1.)$, Berenbrink, Czumaj, Englert, Friedetzky and Nagel~\cite{BCE12} proposed a model where balls are allocated in consecutive batches of size $b$. Each ball in a batch is allocated using \TwoChoice, but based on the load values of the bins prior to the batch. This means, that the decisions among the balls within the same batch do not depend on each other in any way and can therefore be executed in parallel. In~\cite{BCE12}, it was shown that for $b=n$ the gap is $\Oh(\log n)$ \Whp

A related setting with ball removals was studied in Berenbrink, Friedetzky, Kling, Mallmann-Trenn, Nagel and Wastell~\cite{BFKMNW18}, where at each batch $\lambda \cdot n$ arriving balls are allocated in parallel, and every non-empty bin removes one ball. The authors prove an upper bound of $\Oh(\frac{\log n}{1-\lambda})$ on the gap, when balls are allocated using \TwoChoice. As mentioned in~\cite[Lemma~3.5]{BFKMNW18}, their analysis can be modified to re-derive the main result in~\cite{BCE12} for the batched setting, using a simpler proof.

In~\cite{M00}, Mitzenmacher studied a similar model to the batched setting called \emph{bulletin board model with periodic updates}. However, this model assumes stochastic arrivals and removals of balls, and the paper also does not provide any rigorous and quantitative bounds on the gap. On a high level, this work, as well as a study by Dahlin~\cite{D00} raise the general question on how useful old load information is, and both works suggest that using a ``less aggressive'' strategy than \DChoice may lead to better performance in practice.

Other parallel allocation schemes which are typically based on a small number of communication rounds between bins and balls were studied in~\cite{ACM98,LW11,LPY19}. For example, Lenzen and Wattenhofer~\cite{LW11} proved that, for $m=n$, a maximum load of $2$ is possible using only $\log^{*} n+\Oh(1)$ communication rounds. Recently,~\cite{LPY19} also extended this direction to the heavily loaded case $m \geq n$, and proved that a maximum load of $m/n+\Oh(1)$ is possible using $\Oh( \log \log (m/n)+\log^* n)$ rounds. While the gap bounds in~\cite{LW11,LPY19} are stronger than in our setting, they require more communication and coordination.

Assumption 2.) that balls are unweighted is made in the vast majority of theoretical works in balanced allocations. One exception is the work of Talwar and Wieder~\cite{TW07}, who analyzed a wide class of weight distributions satisfying some mild conditions on its second and fourth moment. They proved that the gap remains independent of $m$, even though heavier and heavier weights may be encountered if $m$ gets large. For arbitrary weight distributions, results of~\cite{BFHM08} demonstrate that this setting is considerably harder than the unweighted setting, as many couplings and majorization results no longer hold. 

Concerning assumption 3.) on how bins are sampled, many allocation schemes with fewer or correlated bin samples have been analyzed. One key example is the $(1+\beta)$-process introduced by Peres, Wieder and Talwar~\cite{PTW15}, where each ball is allocated using \OneChoice with probability $1-\beta$ and otherwise is allocated using \TwoChoice. The authors proved that for any $\beta \in (0,1]$, the gap is only $\Oh(\log n/\beta + \log(1/\beta)/\beta)$ for any $m \geq n$. Hence, only a ``small'' fraction of \TwoChoice rounds are enough to inherit the property of \TwoChoice that the gap is independent of $m$. This result also applies to weighted balls into bins for a large class of weight distributions.

A nice application of the $(1+\beta)$-process is in the analysis of the so-called \emph{graphical balanced allocations}~\cite{PTW15}. In this model, first studied in~\cite{KP06} for $m=n$, we are given a graph where each bin corresponds to a vertex. For each ball, we pick a uniform edge in $G$ and place the ball in the lesser loaded endpoint of the edge. A reduction to the $(1+\beta)$-process implies that, if $G$ is a regular expander graph, then for any $m \geq n$ the gap is $\Oh(\log n)$ \Whp Very recently, Bansal and Feldheim~\cite{BF21} presented a more elaborate protocol based on multi-commodity-flows that achieves a poly-logarithmic gap for any bounded-degree regular graph. A natural extension of the graphical process was also studied for hypergraphs, see, e.g., \cite{G08,GMP20}. Other applications of the $(1+\beta)$-process in parallel computing include population protocols~\cite{AGR21} and distributed data structures~\cite{ABK18,N21}.

Related to $(1+\beta)$ process is the two-\Thinning process~\cite{FL20,IK04,ELZ86} with some load threshold $\ell$ which works in a two-stage procedure: First, take a uniform bin sample $i$. Secondly, if the load of bin $i$ is at most $\ell$ then allocate a ball into $i$, otherwise place a ball into another bin sample $j$ (without comparing its load with $i$). This process has received some attention recently, and several variations were studied in~\cite{FL20} for $m=n$ and~\cite{FGL21,LS21,LSS21} for $m \geq n$. \cite{LS21} investigated a variant of \Thinning called \Quantile, which uses \emph{relative} instead of \emph{absolute} loads. This means the ball is allocated in the first sample if its load is among the $(1-\delta) \cdot n$ lightest, for some quantile $\delta \in \{1/n, 2/n, \ldots, 1\}$, and otherwise the ball is allocated into a second bin sample. For both \Thinning and \Quantile, extensions exist which use more than two bin samples, and correspondingly, stronger gap bounds can be shown~\cite{FGL21,LS21}.

An even stronger class of adaptive sampling schemes was analyzed by Czumaj and Stemann~\cite{CS01}, where unlike \Thinning or \Quantile, the ball is always placed in the least loaded bin among all samples. However, their results hold only for $m=n$.

Another model relaxing the \emph{uniform sampling} assumption was introduced in Wieder~\cite{W07}, where the minimum and maximum probability for \emph{sampling} a bin may deviate from the uniform distribution by some factors $\alpha,\beta$. Wieder~\cite{W07} proved some tight trade-off between $\alpha,\beta$ and $d$, so that \DChoice still achieves a small gap for any $m \geq n$. A related model with heterogeneous bins capacities was studied in Berenbrink, Brinkmann, Friedetzky and Nagel~\cite{BBFN14}. The authors proved that \DChoice achieves a gap bound of $\log_{d} \log n+\Oh(1)$, matching the result in the classical setting.

\begin{table}
    \centering
    \renewcommand{\arraystretch}{1.75}
    \footnotesize{
    \begin{tabular}{cccccc}
        \textbf{Process} & \textbf{Graphical} & \textbf{Batch Size} & \textbf{Weights} & \textbf{Gap Bound} & \textbf{Reference} \\
        \hline
        \rowcolor{Gray} \TwoChoice & -- & $b =n$ & -- & $\Oh(\log n)$ & \cite[Thm~1]{BCE12} \\ \hline
        \rowcolor{Greenish} $\mathcal{C}_1,\mathcal{C}_2$ & -- & $b \geq n$ & random & $\Oh( \frac{b}{n} \cdot \log n )$ & Thm~\ref{thm:weak_gap_bound} \\ \hline
        \rowcolor{Greenish} $\mathcal{C}_1,\mathcal{C}_2$ & -- & $b \in [n,n^3]$ & random & $\Oh( \frac{b}{n} + \log n)$ & Thm~\ref{thm:strong_gap_bound} \\ \hline
                \rowcolor{Greenish} $(1+\beta), \ \beta \leq 1 - \Omega(1)$ & -- & $b \geq 1$ & -- & $\Omega(\log n)$ & Prop~\ref{pro:log_lower} \\ \hline
        \rowcolor{Greenish}  & &  &  &   & \\
        \rowcolor{Greenish}\multirow{-2}{*}{\makecell{\TwoChoice, \\ $(1+\beta), \beta = \Omega(1)$}} & \multirow{-2}{*}{--} & \multirow{-2}{*}{$b \geq n \log n$} & \multirow{-2}{*}{--} & \multirow{-2}{*}{$\Omega(\frac{b}{n})$} & \multirow{-2}{*}{Prop~\ref{pro:bn_lower}} \\ \hline
        \rowcolor{Gray} \TwoChoice & $d$-reg., conduct.~$\Phi$ & -- & -- & $\Oh(\frac{\log n}{\Phi})$ & \cite[Thm~3.2]{PTW15} \\ \hline
        \rowcolor{Greenish} \TwoChoice & $d$-reg., conduct.~$\Phi$ & -- & random & $\Oh(d \cdot \frac{\log n}{\Phi})$ & Thm~\ref{thm:unbatched_graphical} \\ \hline
        \rowcolor{Greenish} \TwoChoice & $d$-reg., conduct.~$\Phi$ & $b \geq n$ & random & $\Oh( \frac{d^2}{\Phi} \cdot \frac{b}{n} \cdot \log n)$ & Thm~\ref{thm:graphical} \\ \hline
        \rowcolor{Greenish} \TwoChoice & expander, $d=\Oh(1)$ & $b \in [n,n^3]$  & random & $\Oh( \frac{b}{n} + \log n)$ & Thm~\ref{thm:graphical} \\ \hline
        \rowcolor{Gray} $(1+\beta), \ \beta \leq 1 - \Omega(1)$ & -- & -- & -- & $\Omega(\frac{\log n}{\beta})$ & \cite[Sec~4]{PTW15} \\ \hline
        \rowcolor{Gray} $(1+\beta)$ & -- & -- & random & $\Oh(\frac{\log n}{\beta}+ \frac{\log(1/\beta)}{\beta})$ & \cite[Cor~2.12]{PTW15} \\ \hline
        \rowcolor{Greenish} $(1+\beta)$ & -- & -- & random & $\Oh(\frac{\log n}{\beta})$ & Thm~\ref{thm:OnePlusBetaGap} \\ \hline
    \end{tabular}
    }
    \caption{Overview of the gap bounds in previous works (rows in \hlgray{~Gray~}) and the gap bounds derived in this work (rows in \hlgreenish{~Green~}).
    For simplicity, gap bounds assume that in the conditions $\mathcal{C}_1,\mathcal{C}_2$, all three parameters $\epsilon,\delta,C$ are constants.
    As we show in \cref{pro:verify}, this is satisfied by several processes including \TwoChoice, $(1+\beta)$ for constant $\beta \in (0, 1]$ and $\Quantile(\delta)$ for constant $\delta \in (0, 1)$. The upper bounds on the gap hold for all values of $m$ (see \cref{lem:smoothing_argument}), while some of the lower bounds may only hold for certain $m$.}
    \label{tab:our_results}
\end{table}

\textbf{Our Results.}
In this work we revisit the  \emph{batched model} from~\cite{BCE12}, which allocates balls in batches of size $b$, but here we allow any value of $b \geq n$. Additionally, we consider a wider range of allocation processes, including not only \TwoChoice, but also $(1+\beta)$ or \Quantile. This relaxes the requirement of \TwoChoice of always taking two uniform bin samples at each round and allocating into the less loaded of the two (given the available load information). 

Our results hold for any process satisfying two natural conditions: $(i)$ there is a suitable bias away from allocating into the heavily loaded bins and $(ii)$ no single bin experiences are a ``too large'' bias. The second condition may seem a bit counter-intuitive at first, but it is crucial in the batched setting to prevent a lightly loaded bin from receiving too many allocations within the same batch. The precise definition of these conditions is given in Section~\ref{sec:class_of_processes}. %

Furthermore, our results are valid for the same class of weight distributions as considered in~\cite{PTW15}, which includes, for example, the geometric and the exponential distributions.

Our first result is that for any batch size $b \geq n$, after allocating any number of balls $m \geq n$, a gap bound of $\Oh( b/n \cdot \log n)$ holds \Whp For $b=n$, this matches the result of~\cite{BCE12} for the \TwoChoice process in the unweighted setting. 
Unlike the analysis in~\cite{BCE12}, which relies on some sophisticated Markov chain tools from~\cite{BCSV06} to prove a ``short memory behaviour'', the derivation of this gap bound $\Oh(b/n \cdot \log n)$ is based on a hyperbolic cosine potential function (a version of two exponential potential functions), and thus we believe it to be more elementary and self-contained. On a high level, this analysis shares some of the ideas from~\cite{BFKMNW18,PTW15} which both uses similar versions of exponential functions, but it seems difficult to apply these existing approaches directly to the general setting with weighted balls and any $b \geq n$.

We then proceed to a tighter bound and prove that for any $n \log n \leq b \leq n^{3}$ and any number of balls $m \geq n$, the gap is $\Oh(b/n+ \log n)$ \Whp This bound is derived through an interplay between different potential functions, in particular, we relate three hyperbolic cosine potential functions with different smoothing parameters. 

Next we turn to proving asymptotically tight lower bounds. We prove that for any $b \geq n$, there are processes falling into our framework that produce for certain values of $m$ a gap of $\Omega(b/n+\log n)$. These lower bounds are proven in the unweighted setting where all balls have weight one. For the \TwoChoice process, the lower bounds are tight for $b \geq n \log n$.

Combining our upper with lower bounds reveals an interesting behavior: For any $b \in [n,\Oh(n \log n)]$, the gap is $\Theta(\log n)$ \Whp, whereas for $b \geq n \log n$, the gap is $\Theta(b/n)$ \Whp In particular, the asymptotic gap bound does not change as $b$ moves from $n$ to $n \log n$. 

We further demonstrate the flexibility of our techniques by deriving results for the \emph{graphical allocation model} from~\cite{PTW15}, where bins are arranged as a graph and at each round a pair of bins is sampled by picking a random edge from the graph. One open question in~\cite[Section~4]{PTW15} was to derive results for graphical allocation with weights. In this work, we make progress towards that question by proving gap bounds that hold not only for weighted balls but also in the batched setting. For example, if the graph is a bounded-degree expander, then we recover the gap bound of $\Oh(\log n)$ from~\cite{PTW15} even if balls are weighted and are allocated in batches up to a size of $n \log n$. 
Finally, another consequence of our approach is a tight $\Oh(\log n/\beta)$ upper bound for the $(1+\beta)$ process for any $\beta \leq 1/2$.

Our results are summarized in~\cref{tab:our_results}.

\textbf{Organization.}
In \cref{sec:notation}, we present some standard notation for balanced allocations and define the processes and models used. In \cref{sec:general_hyperbolic_pot}, we generalize (and strengthen) the analysis of the hyperbolic cosine potential of~\cite{PTW15}. In \cref{sec:weak_bound}, we apply this analysis to obtain an $\Oh(\frac{b}{n} \log n)$ gap bound for a family of processes in the batched model with weighted balls. In \cref{sec:refined}, we improve this upper bound on the gap to $\Oh(\frac{b}{n} + \log n)$, for any $n \leq b \leq n^3$. In \cref{sec:graphical}, we demonstrate applications of our analysis to graphical allocation and the $(1+\beta)$ process. In \cref{sec:lower_bounds}, we show that our upper bound from \cref{sec:refined} is asymptotically tight, by providing lower bounds for a large family of processes.  In \cref{sec:experiments} we present some experimental results. Finally, in \cref{sec:conclusion}, we summarize the main results and point to some open problems. %

\section{Notation}\label{sec:notation}

\subsection{Basic Notation and Specific Processes} \label{sec:basic_notation}

We consider the allocation of $m$ balls into $n$ bins, which are labeled $[n]:=\{1,2,\ldots,n\}$. For the moment, the $m$ balls are unweighted (or equivalently, all balls have weight $1$). For any round $t \geq 0$, $x^{t}$ is the $n$-dimensional \emph{load vector}, where $x_i^{t}$ is the number of balls allocated into bin $i$ in the first $t$ allocations. In particular, $x_i^{0}=0$ for every $i \in [n]$. Finally, the \emph{gap} is defined as
\[
 \Gap(t) = \max_{i \in [n]} x_i^{t} - \frac{t}{n}.
\]
It will be also convenient to keep the load vector $x$ sorted. To this end, let $\tilde{x}^t:=x^t-\frac{t}{n}$. Then, relabel the bins such that $y^{t}$ is a permutation of $\tilde{x}^t$ and $y_1^{t} \geq y_2^{t} \geq \cdots \geq y_n^{t}$. Note that $\sum_{i \in [n]} y_i^t=0$ and $\Gap(t)=y_1^t$. We will call a bin $i \in [n]$ \emph{overloaded}, if $y_i \geq 0$ and \emph{underloaded} otherwise. Further, we say that a vector $v=(v_1,v_2,\ldots,v_n)$ \emph{majorizes} $u=(u_1,u_2,\ldots,u_n)$ if for all $1 \leq k \leq n$, the prefix sums satisfy: $\sum_{i=1}^k v_i \geq \sum_{i=1}^k u_i$.

Following~\cite{PTW15}, many allocation processes can be described by a time-invariant \emph{probability vector} $p_i$, $1 \leq i \leq n$, such that at each step $t \geq 0$, $p_i$ is the probability for allocating a ball into the $i$-th most heavily loaded bin (or equivalently, incrementing $y_i^t$ by one).

By $\mathfrak{F}^t$ we denote the filtration of the process until step $t$, which in particular reveals the load vector $x^t$.

We continue with a formal description of the \TwoChoice process.
\begin{framed}
\vspace{-.45em} \noindent
\underline{\TwoChoice Process:} \\
\textsf{Iteration:} For each $t \geq 0$, sample two bins $i_1$ and $i_2$ with replacement, independently and uniformly at random. Let $i$ be one bin with $x_{i}^{t} = \min\{ x_{i_1}^t,x_{i_2}^t\}$, breaking ties randomly. Then update:  
    \begin{equation*}
     x_{i}^{t+1} = x_{i}^{t} + 1.
 \end{equation*}\vspace{-1.5em}
\end{framed}

It is immediate that the probability vector of \TwoChoice is
\begin{equation*}
    p_{i} = \frac{2(i-1)}{n^2}, \qquad \mbox{ for all $i \in [n]$.}
\end{equation*}

Following~\cite{PTW15}, we recall the definition of $(1+\beta)$ which is a process interpolating between \OneChoice and \TwoChoice:
\begin{framed}
\vspace{-.45em} \noindent
\underline{($1+\beta$)-Process:}\\
\textsf{Parameter:} A mixing factor $\beta \in (0,1]$.\\
\textsf{Iteration:} For each $t \geq 0$, sample two bins $i_1$ and $i_2$ with replacement, independently and uniformly at random. Let $i$ be one bin with $x_{i}^{t} = \min\{ x_{i_1}^t,x_{i_2}^t\}$, breaking ties randomly. Then update:  
    \begin{equation*}
    \begin{cases}
     x_{i}^{t+1} = x_{i}^{t} + 1 & \mbox{with probability $\beta$}, \\
      x_{i_1}^{t+1} = x_{i_1}^{t} + 1 & \mbox{otherwise}.
   \end{cases}
 \end{equation*}\vspace{-1.em}
\end{framed}

In other words at each step, $(1+\beta)$-process allocates the ball following the \TwoChoice rule with probability $\beta$, and otherwise allocates the ball following the \OneChoice rule. Therefore, the probability vector is given by~\cite{PTW15}:
\begin{equation*}
    p_{i} =
    (1-\beta) \cdot \frac{1}{n} + \beta \cdot \frac{2(i-1)}{n^2}, \qquad \mbox{ for all $i \in [n]$.}
\end{equation*}

The next process is another relaxation of \TwoChoice.
\begin{framed}
\vspace{-.45em} \noindent
\underline{$\Quantile(\delta)$ Process:}\\
\textsf{Parameter:} A quantile $\delta \in \{1/n, 2/n, \ldots, 1 \}$.\\
\textsf{Iteration:} For each $t \geq 0$, sample two bins $i_1$ and $i_2$ with replacement, independently and uniformly at random, and update:  
    \begin{equation*}
    \begin{cases}
     x_{i_2}^{t+1} = x_{i_2}^{t} + 1 & \mbox{if $i_1$ is among the $\delta \cdot n$ most loaded bins}, \\
     x_{i_1}^{t+1} = x_{i_1}^{t} + 1 & \mbox{otherwise}.
   \end{cases}
 \end{equation*}\vspace{-1.em}
\end{framed}
Note that the $\Quantile(\delta)$ processes can be also implemented as a two-phase procedure: First probe the bin $i_1$ and place the ball there if $i_1$ is not among the $\delta \cdot n$ heaviest bins. Otherwise, take a second sample $i_2$ and place the ball there. Since we only need to know whether a bin's rank is above or below a value, the response by a bin can be encoded as a single bit. The probability vector of $\Quantile(\delta)$ is given by:
\begin{equation*}
    p_{i} =
    \begin{cases}
     \frac{\delta}{n} & \mbox{ if $1 \leq i \leq \delta \cdot n$}, \\
     \frac{1+\delta}{n} & \mbox{ if $\delta \cdot n < i \leq n$}.
    \end{cases}
\end{equation*}
Another, equivalent description of $\Quantile(\delta)$ is that we perform \TwoChoice, but only get to know whether a sampled bin's rank is below or above $\delta \cdot n$ and breaking ties randomly.

An example of the probability vectors of the three processes above can be found in \cref{fig:prob_vectors}.

Finally, we will also consider a graph-based version of balanced allocation, called \emph{graphical balanced allocation}~\cite{PTW15}. This process involves running \TwoChoice on a graph, where only bin pairs can be sampled which are connected by an edge. 
\begin{framed}
\vspace{-.45em} \noindent
\underline{$\Graphical(G)$}\\
\textsf{Parameter:} An undirected, connected, regular graph $G$.\\
\textsf{Iteration:} For each $t \geq 0$, sample an edge $e=\{i_1,i_2\} \in E$ uniformly at random. Let $i$ be one bin with $x_{i}^{t} = \min\{ x_{i_1}^t,x_{i_2}^t\}$, breaking ties randomly. Then update:  
    \begin{equation*}
     x_{i}^{t+1} = x_{i}^{t} + 1
 \end{equation*}\vspace{-1.5em}
\end{framed}
Note that unlike the other processes, the probability vector of \Graphical will generally not be time-invariant.

\begin{figure}[t]
\begin{center}
\begin{tikzpicture}[scale=0.7]
\begin{axis}[domain=0:10,
  samples=40,xlabel=$i$,ylabel=$p_i$,
  grid=both,xmin=0.5,xmax=10.5,xtick={1,2,3,4,5,6,7,8,9,10},ytick={0,0.05,0.1,0.15,0.2},ymin=0,ymax=0.22,
 legend pos=north west]
\addplot[domain=1:10,opacity=0.6,color=green,draw=green, mark=*] plot coordinates
    {
    (1,0.01)
    (2,0.03)
    (3,0.05)
    (4,0.07)
    (5,0.09)
    (6,0.11)
    (7,0.13)
    (8,0.15)
    (9,0.17)
    (10,0.19)
    };
     \addlegendentry{\TwoChoice};
     \addplot[domain=1:10,opacity=0.6,color=red,draw=red, mark=*] plot coordinates
    {
    (1,0.064)
    (2,0.072)
    (3,0.08)
    (4,0.088)
    (5,0.096)
    (6,0.104)
    (7,0.112)
    (8,0.12)
    (9,0.128)
    (10,0.136)
    };
     \addlegendentry{$(1+\beta)$, $\beta=0.4$};
     
    \addplot[domain=1:10,opacity=0.6,color=blue,draw=blue, mark=*] plot coordinates
    {
    (1,0.06)
    (2,0.06)
    (3,0.06)
    (4,0.06)
    (5,0.06)
    (6,0.06)
    (7,0.16)
    (8,0.16)
    (9,0.16)
    (10,0.16)
    };
     \addlegendentry{$\Quantile(0.6)$}; 
     
\end{axis}

\end{tikzpicture} \quad \quad \quad \begin{tikzpicture}[scale=0.7]

\begin{axis}[domain=0:10,
  samples=40,
  grid=both,xmin=0.5,xmax=10.5,ymin=0,ymax=1.03,xtick={1,2,3,4,5,6,7,8,9,10},xlabel=$i$,ylabel=$p_i$,legend pos=north west]
\addplot[domain=1:10,opacity=0.6,color=green,draw=green, mark=*] plot coordinates
    {
    (1,0.01)
    (2,0.04)
    (3,0.09)
    (4,0.16)
    (5,0.25)
    (6,0.36)
    (7,0.49)
    (8,0.64)
    (9,0.81)
    (10,1)
    };
     \addlegendentry{\TwoChoice};
     \addplot[domain=1:10,opacity=0.6,color=red,draw=red, mark=*] plot coordinates
    {
    (1,0.064)
    (2,0.136)
    (3,0.216)
    (4,0.304)
    (5,0.4)
    (6,0.504)
    (7,0.616)
    (8,0.736)
    (9,0.864)
    (10,1)
    };
     \addlegendentry{$(1+\beta)$, $\beta=0.4$};
     
    \addplot[domain=1:10,opacity=0.6,color=blue,draw=blue, mark=*] plot coordinates
    {
    (1,0.06)
    (2,0.12)
    (3,0.18)
    (4,0.24)
    (5,0.3)
    (6,0.36)
    (7,0.52)
    (8,0.68)
    (9,0.84)
    (10,1)
    };
     \addlegendentry{$\Quantile(0.6)$}; 
     
\end{axis}

\end{tikzpicture}
\end{center}
\caption{Illustration of the probability vector $(p_1,p_2,\ldots,p_{10})$ and cumulative probability distribution of \TwoChoice, $(1+\beta)$ with $\beta=0.4$ and $\Quantile(0.6)$, which is sandwiched between the other two processes.}\label{fig:prob_vectors}
\end{figure}
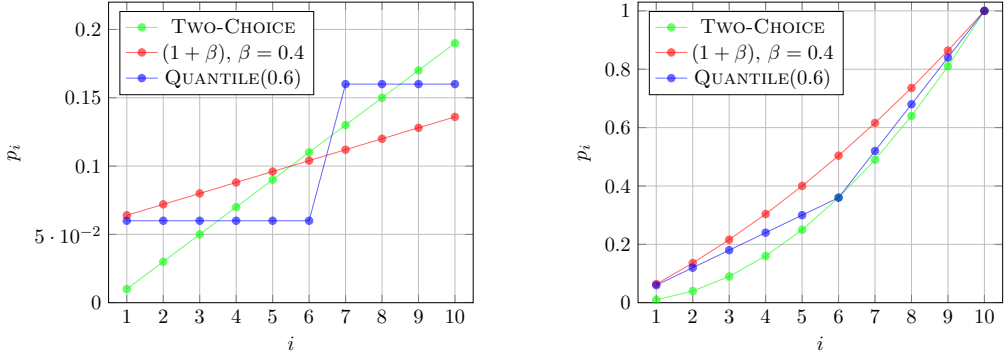

\subsection{Classes of Processes}\label{sec:class_of_processes}

We will now formulate general conditions based on the probability vector $p$ to which our analysis will apply:
\begin{itemize}\itemsep0pt
  \item \textbf{Condition $\mathcal{D}_0$}: $(p_i)_{i \in [n]}$ is a non-decreasing probability vector in $1 \leq i \leq n$.
  \item \textbf{Condition $\mathcal{D}_1$}: There exist constant $\delta \in (0, 1)$ and (not necessarily constant) $\eps \in (0, 1)$,
    \[
    p_{\delta n} \leq \frac{1 - \epsilon}{n}. 
    \]
     \item \textbf{Condition $\mathcal{D}_2$}: For some constant $C > 1$, $\max_{i \in [n]} p_i \leq \frac{C}{n}$. 
\end{itemize}

Note that condition $\mathcal{D}_1$ only provides an upper bound for allocating into the heavier bins. However, due to $\mathcal{D}_0$, this also implies a lower bound on the probability for allocating into the lighter bins (see \cref{obs:conditions}). Finally, we remark any \DChoice process satisfies $\mathcal{D}_2$ for $C = d$.

For the application to graphical allocation, we will relax these conditions slightly and drop the assumption that $(p_i)_{i \in [n]}$ is non-increasing in exchange for a stronger version of condition $\mathcal{D}_1$ that involves both prefix and suffix sums of $p$.
\begin{itemize}\itemsep0pt
  \item \textbf{Condition $\mathcal{C}_1$}: There exist constant $\delta \in (0, 1)$ and (not necessarily constant) $\eps \in (0, 1)$, such that for any $1 \leq k \leq \delta \cdot n$,
    \[
    \sum_{i=1}^{k} p_{i} \leq (1 - \epsilon) \cdot \frac{k}{n},
    \]
    and similarly for any $\delta \cdot n +1 \leq k \leq n$,
    \[
     \sum_{i=k}^{n} p_i \geq \left(1 + \epsilon \cdot \frac{\delta}{1-\delta} \right) \cdot \frac{n-k+1}{n}.
    \]
 \item \textbf{Condition $\mathcal{C}_2$}: For some constant $C > 1$, $\max_{i \in [n]} p_i \leq \frac{C}{n}$. 
\end{itemize}

\begin{restatable}{obs}{ObsConditions}\label{obs:conditions}
Conditions $\mathcal{D}_0$ and $\mathcal{D}_1$ with $\delta$ and $\epsilon$ imply condition $\mathcal{C}_1$ with the same $\delta$ and $\epsilon$.
\end{restatable}

\begin{proof}
Since $p_{\delta n} \leq \frac{1-\epsilon}{n}$ and $p_i$ is non-decreasing, it follows that $p_j \leq \frac{1-\epsilon}{n}$ for all $1 \leq j \leq \delta n$, and thus the prefix sum condition of $\mathcal{C}_1$ holds with equality. We can also conclude that
\[
 \sum_{i=1}^{\delta n} p_i \leq (1-\epsilon) \cdot \delta,
\]
and hence
\[
 \sum_{i=\delta n+1}^{n} p_i \geq 1 - (1 - \epsilon) \cdot \delta.
\]
Since $p_i$ is non-decreasing, it follows for any $\delta n +1 \leq k \leq n$,
\begin{align*}
 \sum_{i=k}^{n} p_i \geq \frac{n-k+1}{(1-\delta) n} \cdot \left(1 - (1 - \epsilon) \cdot \delta \right) = \frac{n-k+1}{n} \cdot \left(1 + \epsilon \cdot \frac{\delta}{1-\delta} \right).
\end{align*}
\end{proof}

Using this observation, it is easy to verify that \TwoChoice, $(1+\beta)$ and \Quantile satisfy the two conditions $\mathcal{C}_1$ and $\mathcal{C}_2$.
\begin{restatable}{pro}{ProVerify}\label{pro:verify}
For any $\beta \in (0,1]$, the $(1+\beta)$-process satisfies condition $\mathcal{C}_1$ with $\delta=\frac{1}{4}$ and $\epsilon=\frac{\beta}{2}$ and condition $\mathcal{C}_2$ with $C= 2$.
Further, for any constant $\delta \in (0,1)$, the $\Quantile(\delta)$ process satisfies condition $\mathcal{C}_1$ with $\delta$ and $\epsilon=1-\delta$, and condition $\mathcal{C}_2$ with $C = 2$.
\end{restatable}

\begin{proof}
For any $i \in [n]$,
\[
 p_i = (1-\beta) \cdot \frac{1}{n} + \beta \cdot \frac{2(i-1)}{n^2}.
\]
This shows that $p_i$ is increasing in $i \in [n]$ (condition $\mathcal{D}_0$), and thus also $\max_{i \in [n]} p_i \leq \frac{2}{n}$ (condition $\mathcal{C}_2$). Further, for $\delta=1/4$,
\[
 p_{\delta n} \leq (1-\beta) \cdot \frac{1}{n} + \beta \cdot \frac{1}{2n} = \left(1 - \frac{\beta}{2}\right) \cdot \frac{1}{n},
\]
proving that $\mathcal{D}_1$ holds with $\epsilon = \beta/2$. By \cref{obs:conditions},  $\mathcal{C}_1$ holds with the same $\epsilon$ and $\delta$.

For $\Quantile(\delta)$, it is obvious that condition $\mathcal{D}_0$ holds, as well as condition $\mathcal{C}_2$ with $C=2$. Further, for any $i \leq \delta \cdot n$, we have $p_i \leq \frac{\delta}{n}$, which means $\mathcal{D}_1$ holds with $\epsilon = 1 - \delta$.
\end{proof}

Note that for $\beta=1$, the $(1+\beta)$-process equals \TwoChoice, so the above statement also applies to \TwoChoice. Finally, since \TwoChoice satisfies $\mathcal{C}_1$, by majorisation also \DChoice for any $d > 2$ satisfies $\mathcal{C}_1$ with the same $\delta$ and $\epsilon$. Further,  \DChoice satisfies $\mathcal{C}_2$ with $C=d$ and thus:

\begin{pro}\label{pro:d_choice}
For any $d \geq 2$, \DChoice satisfies condition $\mathcal{C}_1$ with $\delta=\frac{1}{4}$, $\epsilon=\frac{1}{2}$ and condition $\mathcal{C}_2$ with $C=d$.
\end{pro}

\subsection{Batched Model and Weights}\label{sec:batched_model}

We will now extend the definitions of \cref{sec:basic_notation} and \cref{sec:class_of_processes} to \emph{weighted balls} into bins. To this end, let
 $w^t \geq 0$ be the weight of the $t$-th ball to be allocated ($t \geq 1$). By $W^{t}$ we denote the the total weights of all balls allocated after the first $t \geq 0$ allocations, so $W^t := \sum_{i=1}^n x_i^{t} = \sum_{s=1}^t w^s$. The normalized loads are $\tilde{x}_i^{t} := x_i^t - \frac{W^t}{n}$, and with $y_i^t$ being again the decreasingly sorted, normalized load vector, we have $\Gap(t)=y_1^t$. 

The weight of each ball will be drawn  independently from a fixed distribution $W$ over $[0,\infty)$. Following~\cite{PTW15}, we assume that the distribution $W$ satisfies:
\begin{itemize}
  \item $\ex{W} = 1$.
  \item $\ex{e^{\lambda W} } < \infty $ for some constant $\lambda > 0$.
\end{itemize}
It is clear that when $\ex{W} = \Theta(1)$, by scaling $W$, we can always achieve $\ex{W}=1$.
Specific examples of distributions satisfying above conditions (after scaling) are the geometric, exponential, binomial and Poisson distributions.

Similar to the arguments in~\cite{PTW15}, the above two assumptions can be used to prove that:
\begin{restatable}{lem}{LemBoundedWeightMoment}\label{lem:bounded_weight_moment}
There exists $S := S(\lambda) \geq \max(1, 1/\lambda)$, such that for any  $\alpha \in (0, \min(\lambda/2, 1))$ and any $\kappa \in [-1,1]$,
\[
\Ex{e^{\alpha \cdot \kappa \cdot W}} \leq 1 + \alpha \cdot \kappa + S \alpha^2 \cdot \kappa^2.
\]
\end{restatable}

\begin{proof}
This proof closely follows the argument in~\cite[Lemma 2.1]{PTW15}. Let $M(z) = \Ex{ e^{z W}}$, then using Taylor's Theorem (mean value form remainder),  for any $z \in [-\alpha,\alpha]$ there exists $\xi \in [-\alpha, \alpha]$ such that
\[
M(z) = M(0) + M'(0) \cdot z + M''(\xi) \cdot \frac{1}{2} \cdot z^2 = 1 + z + M''(\xi) \cdot \frac{1}{2} \cdot z^2.
\]
By the assumptions on $\alpha$ and $\lambda$,
\begin{align*}
M''(\xi) &= \ex{W^2 e^{\xi W}} \\ &\stackrel{(a)}{\leq} \sqrt{\ex{W^4 } \cdot \ex{e^{2\xi W} }} \\ &\stackrel{(b)}{\leq} \frac{1}{2} \cdot \left(
 \ex{W^4} + \ex{e^{2\xi W} }
\right) \\
& \stackrel{(c)}{\leq} \frac{1}{2} \cdot 
\Big(
 \Big( \frac{8}{\lambda} \cdot \log\Big( \frac{8}{\lambda} \Big) \Big)^{4} + \ex{ e^{\lambda W}} + \ex{ e^{\lambda W}}
\Big).
\end{align*}
where $(a)$ uses the Cauchy-Schwartz inequality $| \Ex{ X \cdot Y}| \leq \sqrt{\Ex{X^2} \Ex{Y^2} }$ for random variables $X$ and $Y$, $(b)$ uses a mean inequality, and $(c)$ uses \cref{lem:s_bound}. Now defining 
\[
S:=2 \cdot \max\left\lbrace  \left(\frac{8}{\lambda} \cdot \log\Big( \frac{8}{\lambda} \Big) \right)^{4} , 2 \cdot \ex{ e^{\lambda W}}, 1/2 \right\rbrace,
\] and choosing $z:=\kappa \cdot \alpha$,
the lemma follows.
\end{proof}

We will now describe the allocation of weighted balls into bins using a batch size of $ b \geq n$. For the sake of concreteness, let us first describe the batched model if the allocation is done using \TwoChoice. For a given batch size consisting of $b$ consecutive balls, each ball of the batch performs the following. First, it samples two bins $i_1$ and $i_2$ and compares the load the two bins had at the beginning of the batch (let us denote the bin which has less load by $i$). Secondly, a weight is sampled from the distribution $W$. Then a weighted ball is added to bin $i$. Recall that since the load information is only updated at the beginning of the batch, all allocations of the $b$ balls within the same batch can be performed in parallel.

In the following, we will use a more general framework, where the process of sampling (one or more) bins and then deciding where to allocate the ball to is described by a probability vector $p$ over the $n$ bins (\cref{sec:basic_notation} and \cref{sec:class_of_processes}). Also for the analysis, it will be convenient to focus on the normalized and sorted load vector $y$, which is why the definition below is based on $y$ rather than the actual load vector $x$.

\begin{framed}
\vspace{-.45em} \noindent
\underline{Batched Allocation with Weights}\\
\textsf{Parameters:} Batch size $b \geq n$, probability vector $p$, weight distribution $W$.%
\\
\textsf{Iteration:} For each $t = 0 \cdot b, 1 \cdot b, 2 \cdot b, \ldots$:
\begin{enumerate}\itemsep0pt
    \item Sample $b$ bins $i_1,i_2,\ldots,i_b$ from $[n]$ following $p$.
    \item Sample $b$ weights $w^{t+1},w^{t+2},\ldots,w^{t+b}$ from $W$.
    \item Update for each $i \in [n]$, 
    \[
    z_{i}^{t+b}=y_{i}^{t} + \sum_{j=1}^b w^{t+j} \cdot \mathbf{1}_{i_j=i} - \frac{1}{n} \cdot \sum_{j=1}^b w^{t+j}.
    \]
    \item Let $y^{t+b}$ be the vector $z^{t+b}$, sorted decreasingly.
\end{enumerate}
\vspace{-0.4cm}
\end{framed}

We also look at the version of the processes that performs random tie-breaking between bins of the same load. For $b = 1$, this makes no observable difference to the process, but for multiple steps, this effectively averages out the probability over (possibly) multiple bins that have the same load. This would, for instance, correspond to \TwoChoice, randomly deciding between the two bins if they have the same load. In particular, if $p$ is the original probability vector, then the one with random tie-breaking is $\tilde{p}(y^t)$ (for $t$ being the beginning of the batch), where
\begin{equation} \label{eq:averaging_pi}
\tilde{p}_i(y^t) := \frac{1}{|\{ j \in [n] : y_j^t = y_i^t \}|} \cdot \sum_{j \in [n] : y_j^t = y_i^t} p_j, \quad \text{for }i \in [n].
\end{equation}

\begin{framed}
\vspace{-.45em} \noindent
\underline{Batched Allocation with Weights and Random Tie-Breaking}\\
\textsf{Parameters:} Batch size $b \geq n$, probability vector $p$, weight distribution $W$.%
\\
\textsf{Iteration:} For each $t = 0 \cdot b, 1 \cdot b, 2 \cdot b, \ldots$:
\begin{enumerate}\itemsep0pt
    \item Let $\tilde{p} := \tilde{p}(y^t)$ be the probability vector accounting for random tie-breaking.
    \item Sample $b$ bins $i_1,i_2,\ldots,i_b$ from $[n]$ following $\tilde{p}$.
    \item Sample $b$ weights $w^{t+1},w^{t+2},\ldots,w^{t+b}$ from $W$.
    \item Update for each $i \in [n]$, 
    \[
    z_{i}^{t+b}=y_{i}^{t} + \sum_{j=1}^b w^{t+j} \cdot \mathbf{1}_{i_j=i} - \frac{1}{n} \cdot \sum_{j=1}^b w^{t+j}.
    \]
    \item Let $y^{t+b}$ be the vector $z^{t+b}$, sorted decreasingly.
\end{enumerate}
\vspace{-0.4cm}
\end{framed}

Following~\cite{BCE12}, our goal will be to bound the gap at the end of a batch, i.e., $m$ will be a multiple of $b$.

Next we prove the following simple lemma, that high probability gap bounds at the end of batches imply high probability gap bounds at all steps in between, for batch sizes $b = \poly(n)$. Thus, \cref{thm:weak_gap_bound} and \cref{thm:strong_gap_bound} only prove gap bounds at the end of batches.

\begin{lem}[Smoothing argument] \label{lem:smoothing_argument}
Consider any allocation process in the weighted batched setting with $b \geq n$ and a weight distribution satisfying \cref{lem:bounded_weight_moment} for some constants $\lambda > 0$ and $S := S(\lambda) \geq 1$. If for some $m$ being a multiple of $b$, some (not necessarily constant) $c > 0$ and constant $\kappa > 0$,
\[
\Pro{y_1^m \leq c} \geq 1 - n^{-\kappa},
\]
then for any $t \in [m - b, m]$,
\[
\Pro{y_1^t \leq c + \frac{2 \ln(S)}{\lambda} \cdot \frac{b}{n}} \geq 1 - 2n^{-\kappa}.
\]
Similarly, if for some (not necessarily constant) $c > 0$,
\[
\Pro{-y_n^m \leq c} \geq 1 - n^{-\kappa},
\]
then there exists a constant $\kappa := \kappa(S) > 0$, such that for any $t \in [m, m + b]$,
\[
\Pro{-y_n^t \leq c + \frac{2 \ln(S)}{\lambda} \cdot \frac{b}{n}} \geq 1 - 2n^{-\kappa}.
\]
\end{lem}
\begin{proof}
Applying \cref{lem:chernoff} for the $k := b \geq n$ weights $(w^{m - b + j})_{j=1}^b$,
\[
\Pro{\sum_{j = 1}^b w^{m - b + j} \leq \frac{2 \ln(S)}{\lambda} \cdot b } \geq 1 - e^{-\Omega(b)}.
\]
Hence, \Whp~the mean load does not increase by more than $\frac{2 \ln(S)}{\lambda} \cdot \frac{b}{n}$. Therefore by the union bound, for any $t \in [m - b, m]$, any bin load can increase by at most that amount, 
\[
\Pro{y_1^t \leq c + \frac{2 \ln(S)}{\lambda} \cdot \frac{b}{n}} \geq 1 - e^{-\Omega(b)} - n^{-\kappa} \geq 1 - 2n^{-\kappa},
\]
Similarly, by the union bound, for any $t \in [m, m + b]$, any bin load can decrease by at most $\frac{2 \ln(S)}{\lambda} \cdot \frac{b}{n}$,
\[
\Pro{-y_n^t \leq c + \frac{2 \ln(S)}{\lambda} \cdot \frac{b}{n}} \geq 1 - e^{-\Omega(b)} - n^{-\kappa} \geq 1 - 2n^{-\kappa}.
\]
which concludes the claim.
\end{proof}

\section{Analysis of the Hyperbolic Cosine Potential} \label{sec:general_hyperbolic_pot}

In this section we generalize~\cite[Theorem 2.10]{PTW15}. This generalization allows us to apply it to multi-step changes in \cref{sec:weak_bound}, handle general quantile conditions (arbitrary constant $\delta > 0$ instead of $\delta = 1/3$) and obtain tighter bounds of $\Oh(n)$ on the expectation of the potential, which we make use of in \cref{sec:refined}. Further, using this generalization, we obtain bounds on graphical allocation with weights and batches and a tighter upper bound for the $(1+\beta)$ process for very small $\beta$ (\cref{sec:graphical}). The \textit{hyperbolic cosine potential} is defined as
\begin{align}
\Gamma^t := \Phi^t + \Psi^t := \sum_{i = 1}^n e^{\alpha y_i^t} + \sum_{i = 1}^n e^{-\alpha y_i^t}, \label{eq:hyperbolic}
\end{align}
for $\alpha > 0$. We also decompose $\Gamma^t$ across bins as follows, and define for any bin $i \in [n]$:
\[
 \Gamma_i^t := \Phi_i^t + \Psi_i^t = e^{\alpha y_i^t} + e^{-\alpha y_i^t}.
\]
Further, we use the following shorthands to denote the changes in the potentials $\Delta\Phi_i^t := \Phi_i^{t+1} - \Phi_i^t$, $\Delta\Psi_i^t := \Psi_i^{t+1} - \Psi_i^{t}$ and $\Delta\Gamma_i^t := \Gamma_i^{t+1} - \Gamma_i^{t}$. 

The next result holds for any probability vector $p$ satisfying condition $\mathcal{C}_1$ and any load vector $x$. As, we show in \cref{cor:main_ptw}, this implies upper bounds on the expected change of the $\Gamma$ potential, under certain conditions.
\begin{thm} \label{thm:main_ptw}
Consider any probability vector $p$ satisfying condition $\mathcal{C}_1$ for constant $\delta \in (0, 1)$ and $\eps > 0$, and any load vector $x$ with $\Phi :=\Phi(x)$, $\Psi :=\Psi(x)$ and $\Gamma :=\Gamma(x)$. Further for some $K > 0$ define,
\[
\Delta\overline{\Phi} := \sum_{i=1}^n \Delta\overline{\Phi}_i = \sum_{i = 1}^n\Phi_i \cdot \Big(\Big(p_i - \frac{1}{n}\Big) \cdot \alpha + K \cdot \frac{\alpha^2}{n}\Big),
\]
and
\[
\Delta\overline{\Psi} := \sum_{i=1}^n \Delta\overline{\Psi}_i = \sum_{i = 1}^n \Psi_i \cdot \Big(\Big(\frac{1}{n} - p_i\Big) \cdot \alpha + K \cdot \frac{\alpha^2}{n}\Big).
\]
Then, there exists a constant $c := c(\delta) > 0$, such that for any $0 < \alpha < \min(1, \frac{\eps\delta}{8K})$,
\[
\Delta\overline{\Gamma} := \Delta\overline{\Phi} + \Delta\overline{\Psi} \leq -\frac{\eps\delta}{8} \cdot \frac{\alpha}{n} \cdot \Gamma + c \cdot \eps \cdot \alpha.
\]
\end{thm}

Before presenting the proof, we begin with an outline of the key observations in the proof. Let $\Delta\overline{\Gamma}_i := \Delta\overline{\Phi}_i + \Delta\overline{\Psi}_i$. %

\begin{enumerate}
    \item It suffices to analyze the potential for a process with probability vector,
\begin{align}
q_i := \begin{cases}
 \frac{1 - \eps}{n} & \text{if } i \leq \delta n, \\
 \frac{1 + \tilde{\eps}}{n} & \text{otherwise},
\end{cases}
\end{align}
where $\tilde{\eps} := \eps \cdot \frac{\delta}{1 - \delta}$, as this maximizes the terms $\Delta\overline{\Phi}$ and $\Delta\overline{\Psi}$.

\item For any bin $i \in [n]$, there is one dominant term in $\Gamma_i$: for overloaded bins it is $\Phi_i$ (and $\Psi_i \leq 1$) and for underloaded bins it is $\Psi_i$ (and $\Phi_i \leq 1$). The change of the smaller term is absorbed by the change of the dominant and the additive term, i.e., $c \cdot \eps \cdot \alpha$. 

\item It suffices to show that \begin{align*}
\sum_{i = 1}^n & \left( \Phi_i \cdot \Big( p_i - \frac{1}{n}\Big) \cdot \alpha + \Psi_i \cdot \Big(  \frac{1}{n} - p_i \Big) \cdot \alpha\right) \leq - \frac{\eps \delta}{4} \cdot \frac{\alpha}{n} \cdot \Gamma + c \cdot \eps \cdot \alpha,
\end{align*}
as half of the decrease term, i.e., $-\frac{\eps \delta}{8} \cdot \frac{\alpha}{n} \cdot \Gamma$ will counteract the increase term $K \cdot \alpha^2 \cdot \Gamma$ for sufficiently small $\alpha$. So, the main focus is on the coefficients of $\alpha$.

\item Any overloaded bin $i \in [n]$ with $i \leq \delta n$, satisfies $p_i = \frac{1-\eps}{n}$ and so $\Delta\overline{\Phi}_i \leq -\Phi_i \cdot \frac{\alpha\eps}{n} + \Oh(\alpha^2)$. We call these the set $\mathcal{G}_+$ of \textit{good overloaded bins}. The rest of the overloaded bins are the \textit{bad overloaded bins} $\mathcal{B}_+$ and these still satisfy $\Delta\overline{\Phi}_i \leq + \Phi_i \cdot \frac{\alpha\tilde{\eps}}{n} + \Oh(\alpha^2)$. 

Similarly, \textit{good underloaded bins} $\mathcal{G}_-$ with $i > \delta n$, satisfy $\Delta\overline{\Psi}_i \leq -\Psi_i \cdot \frac{\alpha\tilde{\eps}}{n} + \Oh(\alpha^2)$ and \textit{bad underloaded bins} $\mathcal{B}_-$ satisfy $\Delta\overline{\Psi}_i \leq +\Psi_i \cdot \frac{\alpha\eps}{n} + \Oh(\alpha^2)$.

\item We can either have $\mathcal{B}_+ \neq \emptyset$ or $\mathcal{B}_- \neq \emptyset$ (\cref{fig:general_case}). 

\begin{figure}[H]
    \centering
    \includegraphics[scale=0.6]{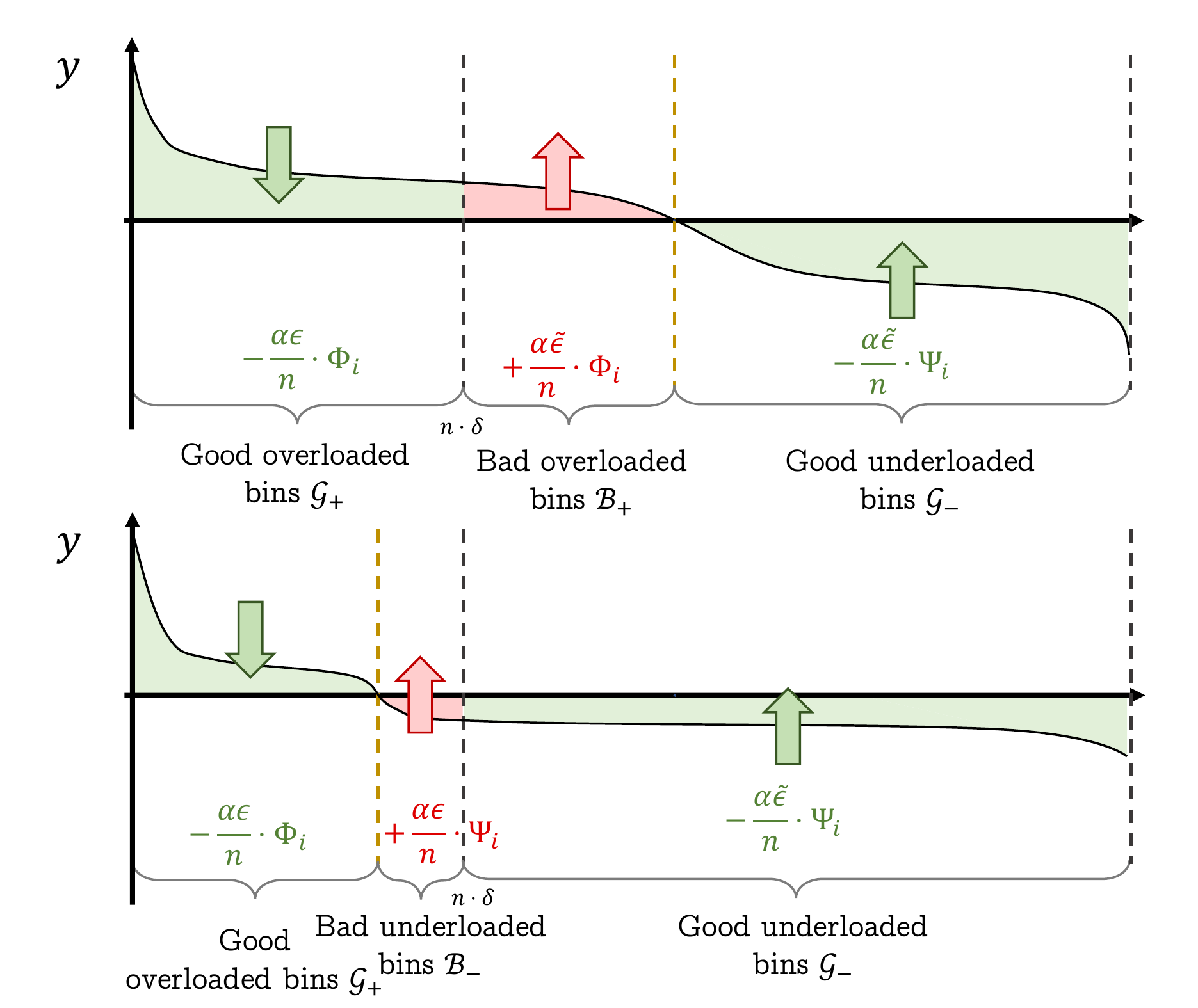}
    \caption{The two cases of bad bins in a configuration and their dominating terms in $\Delta\overline{\Gamma}$ for each of the set of bins.}
    \label{fig:general_case}
\end{figure}

The handling of one case is symmetric to the other due to the symmetric nature of $\Delta \overline{\Phi}$ and $\Delta \overline{\Psi}$ (with $\delta$ being replaced by $1-\delta$). So, from here on we only consider cases with $\mathcal{B}_+ \neq \emptyset$ (and $\mathcal{B}_- = \emptyset$).

\item \textbf{Case A:} When the number of bad overloaded bins is small (i.e., $|\mathcal{B}_+| \leq \frac{n}{2} \cdot (1 - \delta)$), the positive contribution of the bins in $\mathcal{B}_+$ is counteracted by the negative contribution of the bins in $\mathcal{G}_+$ (\cref{fig:case_a}). This is shown by analyzing the worst-case, where all bad bins are equal to $y_{\delta n}$. All underloaded bins are good and so on aggregate we get a decrease.

\begin{figure}[H]
    \centering
    \includegraphics[scale=0.6]{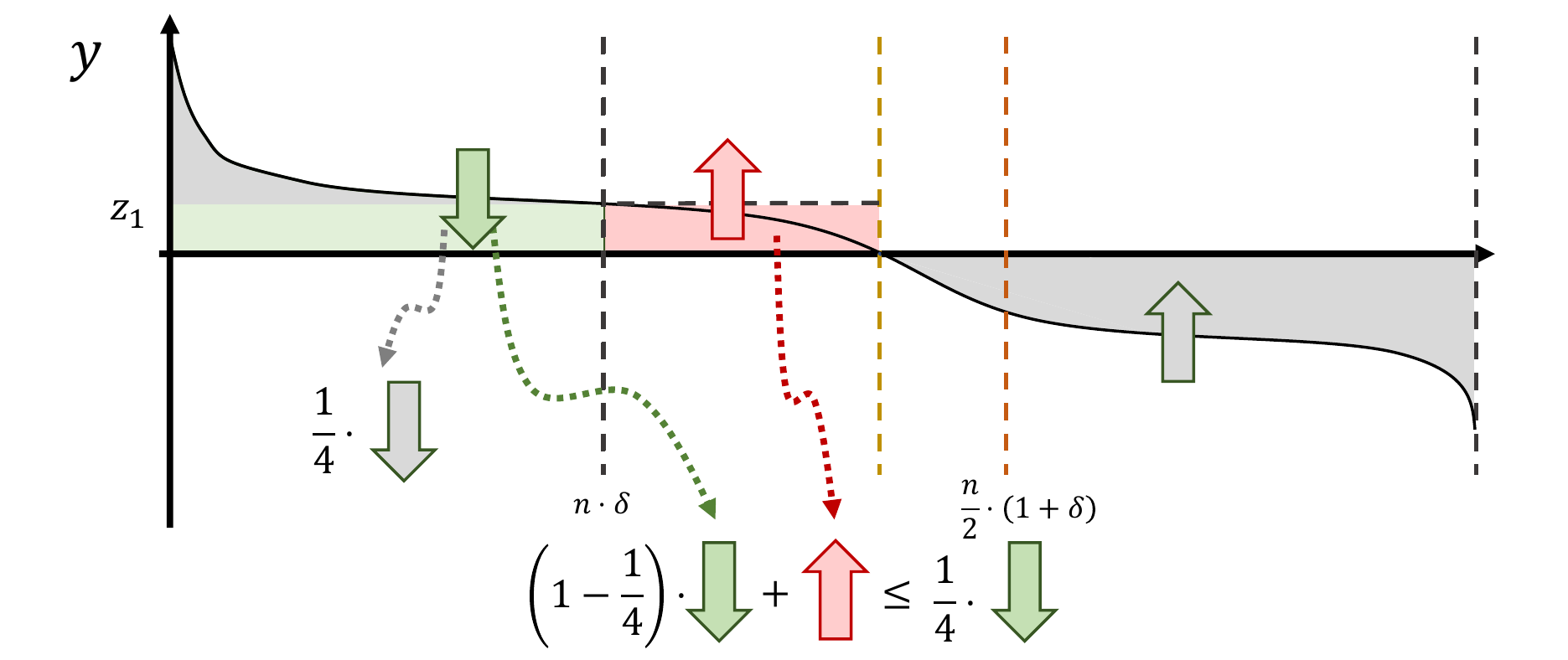}
    \caption{Case A: The positive dominant term in the contribution of bins in $\mathcal{B}_+$ is counteracted by a fraction of the negative contribution term of the good bins $\mathcal{G}_+$.}
    \label{fig:case_a}
\end{figure}

\item \textbf{Case B:} Consider the case when $|\mathcal{B}_+| > \frac{n}{2} \cdot (1 - \delta)$. The positive contribution of the first $\frac{n}{2} \cdot (1 - \delta)$ of the bins $\mathcal{B}_+$, call them $\mathcal{B}_1$, is counteracted by the negative contribution of the bins in $\mathcal{G}_+$ as in Case A. The positive contribution of the remaining bad bins $\mathcal{B}_2$ is counteracted by a fraction of the negative contribution of the bins in $\mathcal{G}_-$. This is because the number of holes in the bins of $\mathcal{G}_-$ are significantly more than the number of bins in $\mathcal{B}_2$. Hence, again on aggregate we get a decrease (\cref{fig:case_b_2}).

\begin{figure}[H]
    \centering
    \includegraphics[scale=0.6]{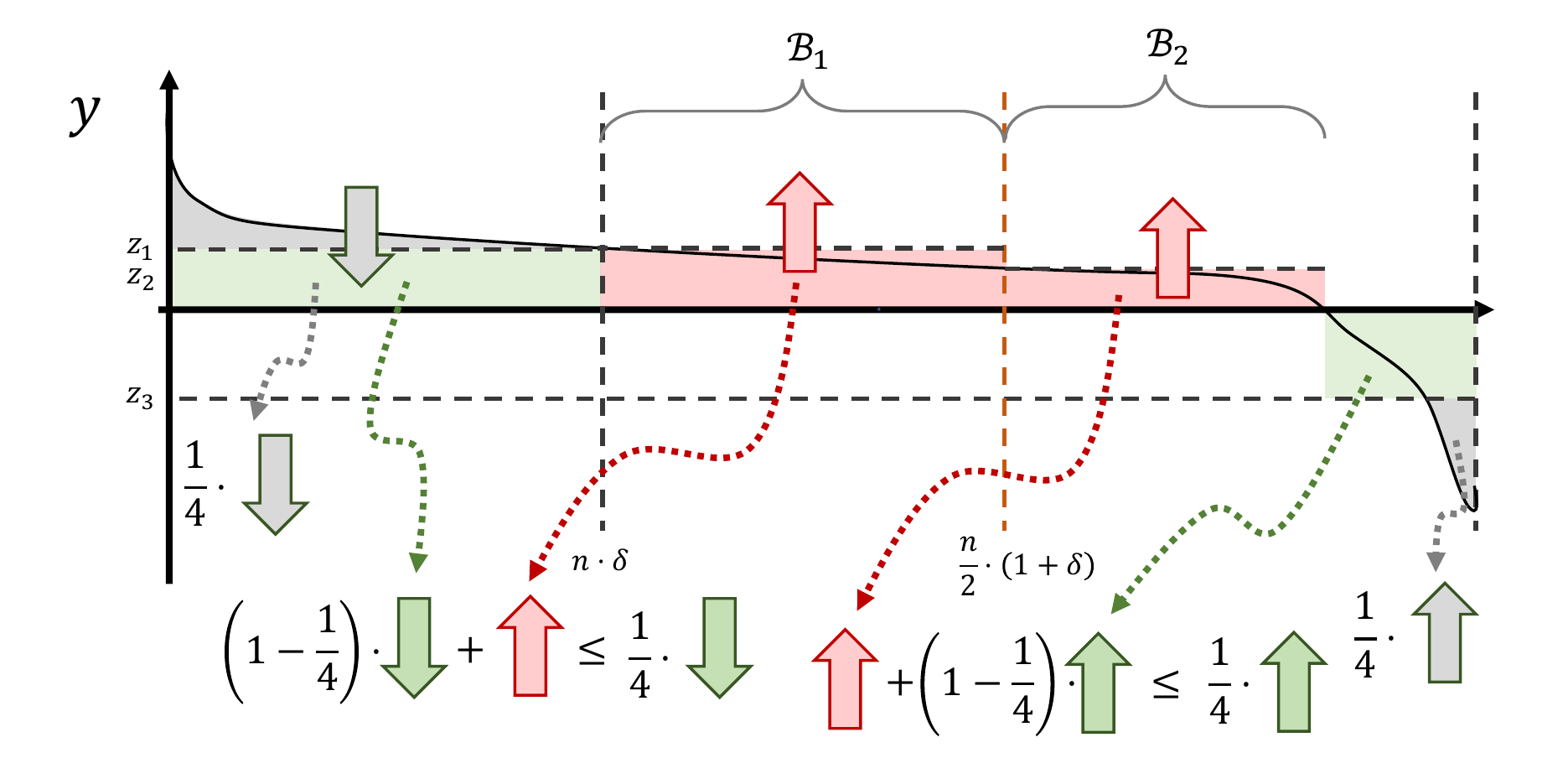}
    \caption{Case B: The dominant change of the bins in $\mathcal{B}_1$ is counteracted by a fraction of the decrease of the bins in $\mathcal{G}_+$ as in Case A. The dominant change of the bins in $\mathcal{B}_2$ is counteracted by a fraction of the decrease of the bins in $\mathcal{G}_-$, when $z_2$ is sufficiently large.}
    \label{fig:case_b_2}
\end{figure}
\end{enumerate}

\begin{proof}
Fix a labeling of the bins so that they are sorted non-increasingly according to their load in $x$. Let $p$ be the probability vector satisfying condition $\mathcal{C}_1$ for some $\epsilon \in (0,1)$ and $\delta \in (0,1)$. Then define another probability vector,
\begin{align}
q_i := \begin{cases}
 \frac{1 - \eps}{n} & \text{if } i \leq \delta n, \\
 \frac{1 + \tilde{\eps}}{n} & \text{otherwise},
\end{cases} \label{eq:probvector_maximized}
\end{align}
where $\tilde{\eps} := \eps \cdot \frac{\delta}{1 - \delta}$. Thanks to the definition of $\tilde{\epsilon}$, it is clear that this is a probability vector. Further, for any $1 \leq k \leq \delta n$,
\[
 \sum_{i=1}^k p_i \leq \sum_{i=1}^k q_i, 
\]
and any $ \delta \cdot n + 1 \leq k \leq n$,
\[
 \sum_{i=k}^{n} p_i \geq \sum_{i=k}^{n} q_i.
\]
This implies that $p$ is majorized by $q$.
Since $\Phi_i$ and $\Psi_i$ are non-increasing in $i \in [n]$, using \cref{lem:quasilem2}, the terms 
\[
\Delta\overline{\Phi} = \sum_{i = 1}^n \Phi_i \cdot \Big(\Big(p_i - \frac{1}{n}\Big) \cdot \alpha + K \cdot \frac{\alpha^2}{n}\Big),
\]
and
\[
\Delta\overline{\Psi} = \sum_{i = 1}^n \Psi_i \cdot \Big(\Big(\frac{1}{n} - p_i\Big) \cdot \alpha + K \cdot \frac{\alpha^2}{n}\Big)
\]
are larger for $q$ than for $p$.
Hence, from now on, we will be working with $p_i=q_i$ for all $i \in [n]$.

Recall, that we partition overloaded bins $i$ with $y_i \geq 0$ into \textit{good overloaded bins} $\mathcal{G}_+$ with $p_i = \frac{1 - \eps}{n}$ and into \textit{bad overloaded bins} $\mathcal{B}_+$ with $p_i = \frac{1 + \tilde{\eps}}{n}$. These are called good bins, because any bin $i \in \mathcal{G}_+$ satisfies $\Delta\overline{\Phi}_i \leq \Phi_i \cdot \left( - \frac{\alpha\eps}{n} + K \cdot \frac{\alpha^2}{n} \right)$ and since $\Psi_i \leq 1$ for overloaded bins, this will imply the drop condition for $\Gamma_i$. %

\textbf{Case A [$1 \leq |\mathcal{B}_+| \leq \frac{n}{2} \cdot (1 - \delta)$]:} Intuitively, in this case the contribution of the bad bins is counteracted by the contribution of the good overloaded bins (\cref{fig:case_a}). 
To formalize this, let $z_1 := y_{\delta n}$ (by assumption on $|\mathcal{B}_{+}|$, we know that $z_1 \geq 0$). Then for any $i \in \mathcal{G}_+$, $y_i \geq z_1$ and for any $i \in \mathcal{B}_+$, $y_i \leq z_1$. So,
\begin{align}
 \sum_{i \in \mathcal{B}_+} \Delta\overline{\Phi}_i
 & \leq \sum_{i \in \mathcal{B}_+} \Phi_i \cdot \Big( \frac{\alpha\tilde{\eps}}{n} + K \cdot \frac{\alpha^2}{n}\Big) 
 \notag \\ &\leq \sum_{i \in \mathcal{B}_+} e^{\alpha z_1} \cdot \frac{\alpha\tilde{\eps}}{n} + \sum_{i \in \mathcal{B}_+} \Phi_i \cdot K \cdot \frac{\alpha^2}{n} \notag \\
 & \leq \frac{n}{2} \cdot (1 - \delta) \cdot e^{\alpha z_1} \cdot \frac{\alpha\tilde{\eps}}{n} + \sum_{i \in \mathcal{B}_+} \Phi_i \cdot K \cdot \frac{\alpha^2}{n} \notag \\
 & =  e^{\alpha z_1} \cdot \frac{\alpha\eps\delta}{2} + \sum_{i \in \mathcal{B}_+} \Phi_i \cdot K \cdot \frac{\alpha^2}{n}, \label{eq:b_plus_change}
\end{align}
where we have used in the last inequality that  $|\mathcal{B}_+| \geq 1$ implies $|\mathcal{G}_+|=\delta \cdot n$. For bins in $\mathcal{G}_+$,
\begin{align}
\sum_{i \in \mathcal{G}_+} \Delta\overline{\Phi}_i  & \leq \sum_{i \in \mathcal{G}_+} \Phi_i \cdot \Big( - \frac{\alpha\eps}{n} + K \cdot \frac{\alpha^2}{n} \Big) \notag
 \\ &= -\sum_{i \in \mathcal{G}_+} \Phi_i \cdot \frac{\alpha\eps}{4n} -\sum_{i \in \mathcal{G}_+} \Phi_i \cdot \frac{3\alpha\eps}{4n} + \sum_{i \in \mathcal{G}_+} \Phi_i \cdot K \cdot \frac{\alpha^2}{n} \notag \\
 & \leq -\sum_{i \in \mathcal{G}_+} \Phi_i \cdot \frac{\alpha\eps}{4n} -\sum_{i \in \mathcal{G}_+} e^{\alpha z_1} \cdot \frac{3\alpha\eps}{4n} + \sum_{i \in \mathcal{G}_+} \Phi_i \cdot K \cdot \frac{\alpha^2}{n} \notag \\
 & = -\sum_{i \in \mathcal{G}_+} \Phi_i \cdot \frac{\alpha\eps}{4n} - e^{\alpha z_1} \cdot \frac{3\alpha\eps\delta}{4} + \sum_{i \in \mathcal{G}_+} \Phi_i \cdot K \cdot \frac{\alpha^2}{n}. \label{eq:g_plus_change}
\end{align}
Hence, combining \cref{eq:b_plus_change} and \cref{eq:g_plus_change}, the contribution for $\Phi$ of overloaded bins is given by 
\begin{align*}
\sum_{i : y_i \geq 0} \Delta\overline{\Phi}_i & \leq -\sum_{i \in \mathcal{G}_+} \Phi_i \cdot \frac{\alpha\eps}{4n} - e^{\alpha z_1} \cdot \frac{\alpha\eps\delta}{4} + \sum_{i : y_i \geq 0} \Phi_i \cdot K \cdot \frac{\alpha^2}{n}  \\
 & \leq -\sum_{i \in \mathcal{G}_+} \Phi_i \cdot \frac{\alpha\eps}{4n} - \sum_{i \in \mathcal{B}_+} \Phi_i \cdot \frac{\alpha\eps\delta}{2(1-\delta)n} + \sum_{i : y_i \geq 0} \Phi_i \cdot K \cdot \frac{\alpha^2}{n}  \\
 & \leq \sum_{i : y_i \geq 0} \Phi_i \cdot \left( - \frac{\alpha\eps \delta}{4n} +  K \cdot \frac{\alpha^2}{n} \right).
\end{align*}
So, using that $\Psi_i \leq 1$ for overloaded bins,
\begin{align}
\sum_{i : y_i \geq 0} \Delta\overline{\Gamma}_i & = \sum_{i : y_i \geq 0} \Delta\overline{\Phi}_i + \sum_{i : y_i \geq 0} \Delta\overline{\Psi}_i 
  \notag \\
 & \leq \sum_{i : y_i \geq 0} \Phi_i \cdot \Big( - \frac{\alpha \eps \delta}{4n} + K \cdot \frac{\alpha^2}{n}\Big) + \sum_{i : y_i \geq 0} \Psi_i \cdot \Big( \frac{\alpha \tilde{\eps}}{n} + K \cdot \frac{\alpha^2}{n}\Big) \notag \\
 & \leq \sum_{i : y_i \geq 0} \Gamma_i \cdot \Big( - \frac{\alpha \eps \delta}{4n} + K \cdot \frac{\alpha^2}{n}\Big) + \sum_{i : y_i \geq 0} \frac{2\alpha}{n} \cdot \max\Big(\tilde{\eps}, \frac{\eps\delta}{4}\Big). \label{eq:case_a_overloaded}
\end{align}
Since in this case all underloaded bins are good, i.e., for any $i \in [n]$ with $y_i <0$, we have $p_i = \frac{1+\tilde{\epsilon}}{n}$, we have
\begin{align}
\sum_{i : y_i < 0} \Delta\overline{\Psi}_i \leq \sum_{i : y_i < 0} \Psi_i \cdot \Big( - \frac{\alpha \tilde{\eps}}{n} + K \cdot \frac{\alpha^2}{n}\Big). \label{eq:case_a_underloaded_prep}
\end{align}
Combining the contribution across all underloaded bins,
\begin{align}
\sum_{i : y_i < 0} \Delta\overline{\Gamma}_i
 & = \sum_{i : y_i < 0}  \Delta\overline{\Phi}_i + \sum_{i : y_i < 0} \Delta\overline{\Psi}_i
  \notag \\
 & \leq \sum_{i : y_i < 0}  \Phi_i \cdot \Big( \frac{\alpha \tilde{\eps}}{n} + K \cdot \frac{\alpha^2}{n}\Big) + \sum_{i : y_i < 0} \Psi_i \cdot \Big( - \frac{\alpha \tilde{\eps}}{n} + K \cdot \frac{\alpha^2}{n}\Big) \notag \\
 & \leq \sum_{i : y_i < 0} \Gamma_i \cdot \Big( - \frac{\alpha \tilde{\eps}}{n} + K \cdot \frac{\alpha^2}{n}\Big) + \sum_{i : y_i < 0} \frac{2\alpha \tilde{\eps}}{n}, \label{eq:case_a_underloaded}
\end{align}
where in the first inequality we used \cref{eq:case_a_underloaded_prep} and the precondition of the theorem, while in the last inequality we used that $\Phi_i \leq 1$ for underloaded bins.

Combining \cref{eq:case_a_overloaded} and \cref{eq:case_a_underloaded},
\begin{align}
\Delta\overline{\Gamma} & = \sum_{i : y_i \geq 0} \Delta\overline{\Gamma}_i + \sum_{i : y_i < 0} \Delta\overline{\Gamma}_i \notag \\
 & \leq \sum_{i : y_i \geq 0} \Gamma_i \cdot \Big( - \frac{\alpha \eps \delta}{4n} + K \cdot \frac{\alpha^2}{n}\Big) + \sum_{i : y_i \geq 0} \frac{2\alpha}{n} \cdot \max\Big(\tilde{\eps}, \frac{\eps\delta}{4}\Big) \notag \\
 & \quad \quad + \sum_{i : y_i < 0} \Gamma_i \cdot \Big( - \frac{\alpha \tilde{\eps}}{n} + K \cdot \frac{\alpha^2}{n}\Big) + \sum_{i : y_i < 0} \frac{2\alpha \tilde{\eps}}{n} \notag \\
 & \leq \sum_{i = 1}^n \Gamma_i \cdot \Big( - \frac{\alpha \eps \delta}{4n} + K \cdot \frac{\alpha^2}{n}\Big) + \sum_{i = 1}^n \frac{2\alpha}{n} \cdot \max\Big(\tilde{\eps}, \frac{\eps\delta}{4}\Big) \notag \\
 & \leq - \Gamma \cdot \frac{\alpha \eps \delta}{8n} + 2\alpha \cdot \max\Big(\tilde{\eps}, \frac{\eps\delta}{4}\Big), \label{eq:case_a_gamma}
\end{align}
using in the last line that $\alpha \leq \frac{\eps\delta}{8K}$.

\textbf{Case B [$|\mathcal{B}_+| > \frac{n}{2} \cdot (1 - \delta)$]:} We partition $\mathcal{B}_+$ into $\mathcal{B}_1 := \mathcal{B}_+ \cap \{i \in [n]: i \leq \frac{n}{2} \cdot (1+\delta) \}$ and $\mathcal{B}_2 := \mathcal{B}_+ \setminus \mathcal{B}_1$. The positive contribution $\Delta\overline{\Phi}_i$ for bins $i \in \mathcal{B}_1$ will be counteracted by that of the bins in $\mathcal{G}_+$ as in Case A. For that of bins in $\mathcal{B}_2$ we consider two cases based on $z_2 := y_{\frac{n}{2} \cdot (1 + \delta)} > 0$, the load of the first bin in $\mathcal{B}_2$. Similarly to \cref{eq:b_plus_change},
\begin{align}
 \sum_{i \in \mathcal{B}_2} \Delta\overline{\Phi}_i
 & \leq  e^{\alpha z_2} \cdot \frac{\alpha\eps\delta}{2} + \sum_{i \in \mathcal{B}_2} \Phi_i \cdot K \cdot \frac{\alpha^2}{n}.\label{eq:b2_contribution}
\end{align}

\textbf{Case B.1 [$z_2 \leq \frac{1}{\alpha}  \cdot \frac{1 - \delta}{2 \delta} \cdot \ln(8/3)$]:} In this case, we will show that the contribution of the bad bins can be absorbed by the additive term. In particular, the contribution of the bins in $\mathcal{B}_2$ is 
\begin{align} 
\sum_{i \in \mathcal{B}_2} \Delta\overline{\Gamma}_i
 & \leq \sum_{i \in \mathcal{B}_2} 2 \cdot e^{\alpha z_2} \cdot \Big( \frac{\alpha \tilde{\eps}}{n} + K \cdot \frac{\alpha^2}{n} \Big) \notag \\
 & \leq \sum_{i \in \mathcal{B}_2} 4 \cdot e^{ \frac{1 - \delta}{2 \delta} \cdot \ln(8/3)} \cdot \frac{\alpha \tilde{\eps}}{n} \notag \\ & < 4 \cdot e^{ \frac{1 - \delta}{2 \delta} \cdot \ln(8/3)} \cdot \alpha \tilde{\eps}.\notag 
\end{align}
Hence, counteracting the positive contribution of the bins in $\mathcal{B}_1$ using that of the bins in $\mathcal{G}_+$ as in Case A (since $|\mathcal{B}_1| \leq \frac{n}{2} \cdot (1-\delta)$) as in \cref{eq:case_a_gamma}, we have
\[
\Delta\overline{\Gamma}  \leq -\Gamma \cdot \frac{\alpha\eps \delta}{8n} + \max\Big(\frac{\eps\delta}{4}, \tilde{\eps}, 2 \cdot e^{ \frac{1 - \delta}{2 \delta} \cdot \ln(8/3)} \cdot \tilde{\eps} \Big) \cdot 2\alpha.
\]

\textbf{Case B.2 [$z_2 > \frac{1}{\alpha}  \cdot \frac{1 - \delta}{2 \delta} \cdot \ln(8/3)$]:} In this case, it means that there are substantially more holes in the underloaded bins than balls in the overloaded bins of $\mathcal{B}_1$. Hence, as we will prove below, the negative contribution $\Delta\overline{\Psi}$ for bins in $\mathcal{G}_-$ will counteract the positive contribution of $\Delta\overline{\Phi}$ for $\mathcal{B}_1$ (\cref{fig:case_b_2}),

\begin{align}
\sum_{i \in \mathcal{G}_-} \Delta\overline{\Psi}_i  & \leq \sum_{i \in \mathcal{G}_-} \Psi_i \cdot \Big( - \frac{\alpha \tilde{\eps}}{n} + K \cdot \frac{\alpha^2}{n}\Big) \notag \\
 & \leq \sum_{i \in \mathcal{G}_-} \Psi_i \cdot \Big( - \frac{\alpha \tilde{\eps}}{4n} + K \cdot \frac{\alpha^2}{n}\Big) - \sum_{i \in \mathcal{G}_-} \Psi_i \cdot \frac{3\alpha \eps \delta}{4 \cdot (1 - \delta) \cdot n}. \label{eq:case_b2_psi_i}
\end{align}
The term $\sum_{i \in \mathcal{G}_-} \Psi_i$ is minimized when all underloaded bins are equal to the same load $-z_3 < 0$, i.e. $\sum_{i \in \mathcal{G}_-} \Psi_i \geq |\mathcal{G}_-| \cdot e^{\alpha z_3}$. Note that $z_3 \geq \frac{z_2 \cdot (|\mathcal{B}_1| + |\mathcal{G}_+|)}{|\mathcal{G}_-|} \geq \frac{z_2 \cdot \frac{n}{2} \cdot (1+\delta)}{|\mathcal{G}_-|}$ and that the function $f(z) = z \cdot e^{k/z}$ is decreasing for $0 \leq z \leq k$ (\cref{lem:decreasing_fn}). Hence, for $k = z_2 \cdot \frac{n}{2} \cdot (1+ \delta)$, the maximum size $|\mathcal{G}_-| = \frac{n}{2} \cdot (1-\delta) \leq k$, minimizes the term $\sum_{i \in \mathcal{G}_-} \Psi_i$. 
We lower bound $z_3$ as follows,
\[
z_3 \geq \frac{z_2 \cdot \frac{n}{2} \cdot (1 + \delta)}{\frac{n}{2} \cdot (1 - \delta)} = z_2 + z_2 \cdot \frac{2\delta}{1 - \delta} \geq z_2 + \frac{1}{\alpha} \cdot \ln(8/3),
\]
using the lower bound on $z_2$. Hence, 
\begin{align*}
\sum_{i \in \mathcal{G}_-} \Psi_i &\geq \frac{n}{2} \cdot (1-\delta) \cdot e^{\alpha \cdot (z_2 + \frac{1}{\alpha} \cdot \ln(8/3)) } \cdot \frac{3\alpha \eps \delta}{4 \cdot (1 - \delta) \cdot n} \\ &= e^{\alpha z_2} \cdot e^{\alpha \cdot \frac{1}{\alpha} \cdot \ln(8/3)} \cdot \frac{3\alpha \eps \delta}{8} \\ &= e^{\alpha z_2} \cdot (\alpha \eps \delta).
\end{align*}
Applying this to \cref{eq:case_b2_psi_i},
\begin{align}
\sum_{i \in \mathcal{G}_-} \Delta\overline{\Psi}_i
 & \leq \sum_{i \in \mathcal{G}_-} \Psi_i \cdot \Big( - \frac{\alpha \tilde{\eps}}{4n} + K \cdot \frac{\alpha^2}{n}\Big) - e^{\alpha z_2} \cdot (\alpha \eps \delta). \label{eq:good_underloaded}
\end{align}
Aggregating \cref{eq:b2_contribution} and \cref{eq:good_underloaded}, the contribution of underloaded bins to $\Delta\overline{\Phi}$ is
\begin{align*}
\sum_{i \in \mathcal{G}_-} \Delta\overline{\Psi}_i + \sum_{i \in \mathcal{B}_2} \Delta\overline{\Phi}_i & \leq \sum_{i \in \mathcal{G}_-} \Psi_i \cdot \Big( - \frac{\alpha \tilde{\eps}}{4n} + K \cdot \frac{\alpha^2}{n}\Big) - e^{\alpha z_2} \cdot (\alpha \eps \delta) \\ &\qquad + e^{\alpha z_2} \cdot \frac{\alpha\eps\delta}{2} + \sum_{i \in \mathcal{B}_2} \Phi_i \cdot K \cdot \frac{\alpha^2}{n} \\
 & = \sum_{i \in \mathcal{G}_-} \Psi_i \cdot \Big( - \frac{\alpha \tilde{\eps}}{4n} + K \cdot \frac{\alpha^2}{n}\Big) - e^{\alpha z_2} \cdot \frac{\alpha\eps\delta}{2} + \sum_{i \in \mathcal{B}_2} \Phi_i \cdot K \cdot \frac{\alpha^2}{n} \\
 & \leq \sum_{i \in \mathcal{G}_-} \Psi_i \cdot \Big( - \frac{\alpha \tilde{\eps}}{4n} + K \cdot \frac{\alpha^2}{n}\Big) - \sum_{i \in \mathcal{B}_2} \Phi_i \cdot \frac{\alpha\eps\delta}{(1-\delta) \cdot n} + \sum_{i \in \mathcal{B}_2} \Phi_i \cdot K \cdot \frac{\alpha^2}{n}.
\end{align*}
Hence, 
\begin{align} 
\lefteqn{ \sum_{i \in \mathcal{G}_-} \Delta\overline{\Gamma}_i  + \sum_{i \in \mathcal{B}_2} \Delta\overline{\Gamma}_i} \notag \\
 & \leq - \sum_{i \in \mathcal{G}_-} \Psi_i \cdot \frac{\alpha \tilde{\eps}}{4n} + \sum_{i \in \mathcal{G}_-} \Phi_i \cdot \frac{\alpha \tilde{\eps}}{n}- \sum_{i \in \mathcal{B}_2} \Phi_i \cdot \frac{\alpha\eps\delta}{(1-\delta) \cdot n} + \sum_{i \in \mathcal{B}_2} \Psi_i \cdot \frac{\alpha \eps}{n} + \sum_{i \in \mathcal{G}_- \cup \mathcal{B}_2} \Gamma_i \cdot K \cdot \frac{\alpha^2}{n} \notag \\
 & \leq \sum_{i \in \mathcal{G}_- \cup \mathcal{B}_2} \Gamma_i \cdot \Big( - \frac{\alpha\eps\delta}{4n} + K \cdot \frac{\alpha^2}{n} \Big) + \sum_{i \in \mathcal{G}_- \cup \mathcal{B}_2} \frac{2\alpha}{n} \cdot \max\Big(\frac{\eps\delta}{4}, \tilde{\eps} \Big). \label{eq:case_b2_gamma_i}
\end{align}
Aggregating similarly, to Case A, for $\mathcal{G}_+$ and $\mathcal{B}_1$,
\begin{align}
\lefteqn{ \sum_{i \in \mathcal{G}_+} \Delta\overline{\Gamma}_i + \sum_{i \in \mathcal{B}_1} \Delta\overline{\Gamma}_i} \notag \\ &\leq \sum_{i \in \mathcal{G}_+ \cup \mathcal{B}_1} \Gamma_i \cdot \Big( - \frac{\alpha\eps\delta}{4n} + K \cdot \frac{\alpha^2}{n} \Big) + \sum_{i \in \mathcal{G}_+ \cup \mathcal{B}_1} \frac{2\alpha}{n} \cdot \max\Big(\tilde{\eps}, \frac{\eps\delta}{4}\Big). \label{eq:case_b2_gamma_ii}
\end{align}
Hence, combining \cref{eq:case_b2_gamma_i} and \cref{eq:case_b2_gamma_ii}, 
\begin{align*}
\Delta\overline{\Gamma} 
 \leq -\Gamma \cdot \frac{\alpha\eps \delta}{8n} + \max\Big(\frac{\eps\delta}{4}, \tilde{\eps} \Big) \cdot 2\alpha.
\end{align*}

\textbf{Case C, D:} These are symmetric to Case A and Case B, but interchanging $\Phi$ with $\Psi$,  $\delta$ with $1-\delta$ and negating the normalized load vector and flipping the load vector.

Combining the four cases, we get that
\[
\Delta\overline{\Gamma} 
 \leq -\Gamma \cdot \frac{\alpha\eps \delta}{8n} + c \cdot \eps \cdot \alpha,
\]
where $c := 2 \cdot \max\Big(\frac{\delta}{4}, \frac{\delta}{1-\delta}, 2 \cdot e^{ \frac{1 - \delta}{2 \delta} \cdot \ln(8/3)} \cdot \frac{\delta}{1-\delta}, 2 \cdot e^{ \frac{\delta}{2 (1 - \delta)} \cdot \ln(8/3)}\Big)$, using that $\tilde{\eps} := \eps \cdot \frac{\delta}{1 - \delta} $.
\end{proof}

By scaling the quantities $\Delta\overline{\Phi}$ and $\Delta\overline{\Psi}$ in \cref{thm:main_ptw} by some $\kappa > 0$, we obtain:

\begin{restatable}{cor}{CorMainPtw} \label{cor:main_ptw}
Consider any allocation process with probability vector $p$ satisfying conditions $\mathcal{C}_1$ for constant $\delta \in (0, 1)$ and $\eps > 0$. Further assume that it satisfies for some $K > 0$ and some $\kappa > 0$, for any $t \geq 0$,
\[
\sum_{i = 1}^n \Ex{\Delta\Phi_i^{t+1} \mid \mathfrak{F}^t} \leq \sum_{i = 1}^n \Phi_i^t \cdot \Big(\Big(p_i - \frac{1}{n}\Big) \cdot \kappa \cdot \alpha + K \cdot \kappa \cdot \frac{\alpha^2}{n}\Big),
\]
and
\[
\sum_{i = 1}^n \Ex{\Delta\Psi_i^{t+1} \mid \mathfrak{F}^t} \leq  \sum_{i = 1}^n \Psi_i^t \cdot \Big(\Big(\frac{1}{n} - p_i\Big) \cdot \kappa \cdot \alpha + K \cdot \kappa \cdot \frac{\alpha^2}{n}\Big).
\]
Then, there exists a constant $c := c(\delta) > 0$, such that for $0 < \alpha < \min(1, \frac{\eps\delta}{8K})$
\[
\Ex{\Delta\Gamma^{t+1} \mid \mathfrak{F}^t} \leq -\frac{\eps\delta}{8} \cdot \kappa \cdot \frac{\alpha}{n} \cdot \Gamma^t + c \cdot \kappa \cdot \eps \cdot \alpha,
\]
and
\[
\Ex{\Gamma^t} \leq \frac{8c}{\delta} \cdot n.
\]
\end{restatable}

\begin{proof}
Applying \cref{thm:main_ptw} for the current load vector $x^t$ and for \[
\Delta\overline{\Phi} := \sum_{i = 1}^n \Phi_i^t \cdot \Big(\Big(p_i - \frac{1}{n}\Big) \cdot \alpha + K  \cdot \frac{\alpha^2}{n}\Big) \quad \text{ and } \quad \Delta\overline{\Psi} := \sum_{i = 1}^n \Psi_i^t \cdot \Big(\Big(\frac{1}{n} - p_i\Big) \cdot \alpha + K \cdot \frac{\alpha^2}{n}\Big),
\]
we get
\begin{align} \label{eq:application_res}
\Delta\overline{\Phi} + \Delta\overline{\Psi} %
\leq -\frac{\eps\delta}{8} \cdot \frac{\alpha}{n} \cdot \Gamma^t + c \cdot \eps \cdot \alpha.
\end{align}
By the assumptions, 
\begin{align} \label{eq:assumption}
\Ex{\Delta\Gamma^{t+1} \mid \mathfrak{F}^t} = \Ex{\Delta\Phi^{t+1} \mid \mathfrak{F}^t} + \Ex{\Delta\Psi^{t+1} \mid \mathfrak{F}^t} \leq \kappa \cdot (\Delta\overline{\Phi} + \Delta\overline{\Psi}).
\end{align}
Hence, combining \cref{eq:application_res} and  \cref{eq:assumption}, we get
\begin{align*}
\Ex{\Delta\Gamma^{t+1} \mid \mathfrak{F}^t} 
& \leq -\frac{\eps\delta}{8} \cdot \kappa \cdot \frac{\alpha}{n} \cdot \Gamma^t + c \cdot \kappa \cdot \eps \cdot \alpha.
\end{align*}
Now, we will show by induction that for any $t \geq 0$, $\Ex{\Gamma^t} \leq \frac{8c}{\delta} \cdot n$. Assume true for $t$, then  
\begin{align*}
\Ex{\Gamma^{t+1}} 
 & = \Ex{\Ex{\Gamma^{t+1} \mid \mathfrak{F}^t}} \\ &\leq \Ex{\Gamma^t \cdot \Big( 1 -\frac{\eps\delta}{8} \cdot \kappa \cdot \frac{\alpha}{n}\Big)} + c \cdot \kappa \cdot \eps \cdot \alpha \\
 & \leq \frac{8c}{\delta} \cdot n \cdot \Big( 1 -\frac{\eps\delta}{8} \cdot \kappa \cdot \frac{\alpha}{n}\Big) + c \cdot \kappa \cdot \eps \cdot \alpha \\ &= \frac{8c}{\delta} \cdot n - c \cdot \kappa \cdot \eps \cdot \alpha + c \cdot \kappa \cdot \eps \cdot \alpha = \frac{8c}{\delta} \cdot n. \qedhere
\end{align*}
\end{proof}

\section{A Simple Upper Bound} \label{sec:weak_bound}

In this section we derive an upper bound of $\Oh( b/n \cdot \log n)$ for the weighted batched setting. This upper bound is tight for $b=\Theta(n)$, as shown in \cref{sec:lower_bounds}. We will make use of the hyperbolic cosine potential as defined in \cref{eq:hyperbolic}. This will also serve as the base case for the tighter analysis in \cref{sec:refined}.

The main goal is to derive the preconditions of \cref{cor:main_ptw} and apply it for $\kappa := b$ over the batches (not individual time steps).

\begin{lem} \label{lem:batching_pot_changes}
Consider the weighted batched setting with batch size $b \geq n$, for a process with probability vector $p$ satisfying condition $\mathcal{C}_2$ for some $C > 1$ and the weight distribution satisfying \cref{lem:bounded_weight_moment} for some $S \geq 1$. Then for any $0 <\alpha \leq \frac{n}{2CSb}$, for any $t \geq 0$ being a multiple of $b$,
\begin{align}
\sum_{i = 1}^n \ex{\Phi_i^{t+b} \mid \mathfrak{F}^t} \leq \sum_{i = 1}^n \Phi_i^t \cdot \left(1 + \Big(p_i -\frac{1}{n}\Big) \cdot \alpha \cdot b + 5\cdot C^2 \cdot S^2 \cdot \frac{b}{n} \cdot \frac{\alpha^2}{n} \cdot b \right),\label{eq:one}
\end{align}
and 
\begin{align}
\sum_{i = 1}^n \ex{\Psi_i^{t+b} \mid \mathfrak{F}^t} \leq \sum_{i = 1}^n \Psi_i^t \cdot \left(1 + \Big(\frac{1}{n} - p_i\Big) \cdot \alpha \cdot b + 5\cdot C^2 \cdot S^2 \cdot \frac{b}{n} \cdot \frac{\alpha^2}{n} \cdot b \right). \label{eq:two}
\end{align}
\end{lem}
\begin{proof}
Consider an arbitrary bin $i \in [n]$. Let $Z \in \{0,1 \}^b$ be the indicator vector, where $Z_j$ indicates whether the $j$-th ball was allocated to bin $i$. The expected change for the overload potential $\Phi_i^t$, is given by
\begin{align}
& \ex{\Phi_i^{t+b} \mid \mathfrak{F}^t}  = \Phi_i^t \cdot \sum_{z \in \{0,1 \}^b} \Pro{Z = z} \cdot \Ex{\left. e^{\alpha \sum_{j = 1}^n (z_j w^{t+j} - \frac{w^{t+j}}{n})} \, \right\vert \, \mathfrak{F}^t, Z = z} \notag .
 \end{align}
 In the following, let us upper bound the factor of $\Phi_i^t$:
 \begin{align}
 \lefteqn{\sum_{z \in \{0,1 \}^b} \Pro{Z = z} \cdot \Ex{\left. e^{\alpha \sum_{j = 1}^n (z_j w^{t+j} - \frac{w^{t+j}}{n})} \, \right\vert \, \mathfrak{F}^t, Z = z} } \notag \\
 &\stackrel{(a)}{=} \!\!\! \sum_{z \in \{0,1 \}^b} \prod_{j = 1}^b (p_i)^{z_j}  (1 - p_i)^{1 - z_j}  (\ex{e^{\alpha W (1 - \frac{1}{n})}})^{z_j}  (\ex{e^{-\alpha \frac{W}{n}}})^{1- z_j} \notag \\
 &\stackrel{(b)}{\leq}\!\!\! \sum_{z \in \{0,1 \}^b} \prod_{j = 1}^b \left(p_i \cdot \Big(1 + \alpha \cdot \Big(1 - \frac{1}{n}\Big) + S\alpha^2 \Big)\Big)^{z_j}  \cdot \Big((1 - p_i) \cdot \Big(1 - \frac{\alpha}{n} + S\cdot \frac{\alpha^2}{n^2}\Big) \Big)^{1 - z_j} \right) \notag \\
 & \stackrel{(c)}{=} \left( p_i \cdot \Big(1 + \alpha \cdot \Big(1 - \frac{1}{n}\Big) + S\alpha^2 \Big) + (1 - p_i) \cdot \Big(1 - \frac{\alpha}{n} + S\cdot \frac{\alpha^2}{n^2}\Big) \right)^b \notag \\
 & = \left( 1 + \Big(p_i - \frac{1}{n}\Big) \cdot \alpha + p_i \cdot S\alpha^2 + (1-p_i) \cdot \frac{S\alpha^2}{n^2} \right)^b \notag \\
 & \leq \left( 1 + \Big(p_i - \frac{1}{n}\Big) \cdot \alpha + 2 \cdot p_i \cdot S\alpha^2 \right)^b, \label{eq:phi_batched_i}
\end{align}
using in $(a)$ that the weights are independent given $\mathfrak{F}^t$, in $(b)$ the \cref{lem:bounded_weight_moment} and in $(c)$ the binomial theorem. 
Let us define 
\[
y:=\Big(p_i - \frac{1}{n}\Big) \cdot \alpha + 2 \cdot p_i \cdot S\alpha^2.
\]
We first claim that $y \cdot b \leq 1$, which holds indeed since
\begin{align*}
    y \cdot b &= \Big(p_i -\frac{1}{n}\Big) \cdot \alpha \cdot b + 2 \cdot p_i \cdot S\alpha^2 \cdot b \\ &\leq \frac{C}{n} \cdot  \alpha \cdot b + 2 \cdot \frac{C}{n} \cdot S \alpha^2 \cdot b \leq 2CS \cdot \alpha \cdot \frac{b}{n} \leq 1,
\end{align*}
where in the last inequality we used $\alpha \leq \frac{n}{2CSb}$.

Then,
\begin{align*}
\ex{\Phi_i^{t+b} \mid \mathfrak{F}^t} 
 & \stackrel{(a)}{\leq} \Phi_i^t \cdot e^{y \cdot b}\\
 & \stackrel{(b)}{\leq} \Phi_i^t \cdot \left( 1 + y \cdot b + y^2 \cdot b^2 \right) \\
 & = \Phi_i^t \cdot \left(1 + \Big(p_i -\frac{1}{n}\Big) \cdot \alpha \cdot b + 2 \cdot p_i \cdot S \alpha^2 \cdot b + \Big(\Big(p_i -\frac{1}{n}\Big) \cdot \alpha \cdot b + 2 \cdot p_i \cdot S \alpha^2 \cdot b \Big)^2 \right),
 \end{align*}
 using in $(a)$ that $1 + y \leq e^y$ for any $y$, and in $(b)$ that $e^y \leq 1 + y + y^2$ for $y \leq 1.75$. Since $p_i \leq \frac{C}{n}$ for all $i \in [n]$, we conclude
 \begin{align*}
 \ex{\Phi_i^{t+b} \mid \mathfrak{F}^t} 
 & \leq \Phi_i^t \cdot \left(1 + \Big(p_i -\frac{1}{n}\Big) \cdot \alpha \cdot b + 2 \cdot \frac{CS}{n}\cdot \alpha^2 \cdot b + \left(\frac{2CS}{n} \cdot \alpha \cdot b \right)^2 \right) \\
 & \leq \Phi_i^t \cdot \left(1 + \Big(p_i -\frac{1}{n}\Big) \cdot \alpha \cdot b + 5\cdot C^2\cdot S^2 \cdot \frac{b}{n} \cdot \frac{\alpha^2}{n} \cdot b \right).
\end{align*}
Similarly, for the underloaded potential $\Psi^t$, for any bin $i \in [n]$,
\begin{align*}
\ex{\Psi_i^{t+b} \mid \mathfrak{F}^t} & = \Psi_i^t \cdot \sum_{z \in \{0,1 \}^b} \Pro{Z = z} \cdot \Ex{\left. e^{-\alpha \sum_{j = 1}^n (z_j w^{t+j} - \frac{w^{t+j}}{n})} \, \right\vert \, \mathfrak{F}^t, Z = z}. 
 \end{align*}
 As before, we will upper bound the factor of $\Psi_i^t$:
 \begin{align}
& \sum_{z \in \{0,1 \}^b} \Pro{Z = z} \cdot \Ex{\left. e^{-\alpha \sum_{j = 1}^n (z_j w^{t+j} - \frac{w^{t+j}}{n})} \, \right\vert \, \mathfrak{F}^t, Z = z} \\
 & \stackrel{(a)}{=} \!\!\! \sum_{z \in \{0,1 \}^b} \prod_{j = 1}^b (p_i)^{z_j}  (1 - p_i)^{1 - z_j}  (\ex{e^{-\alpha W \cdot (1 - \frac{1}{n})}})^{z_j}  (\ex{e^{\alpha \frac{W}{n}}})^{1- z_j} \notag \\
& \stackrel{(b)}{\leq}  \!\!\! \sum_{z \in \{0,1 \}^b} \prod_{j = 1}^b  p_i \cdot \Big(1 - \alpha \cdot \Big(1 - \frac{1}{n}\Big) + S\alpha^2 \Big)\Big)^{z_j} \cdot \Big((1 - p_i) \cdot \Big(1 + \frac{\alpha}{n} + S\cdot \frac{\alpha^2}{n^2}\Big) \Big)^{1 - z_j} \notag \\
 & \stackrel{(c)}{=} \left( p_i \cdot \Big(1 - \alpha \cdot \Big(1 - \frac{1}{n}\Big) + S\alpha^2 \Big) + (1 - p_i) \cdot \Big(1 + \frac{\alpha}{n} + S\cdot \frac{\alpha^2}{n^2}\Big) \right)^b \notag \\
 & =  \left( 1 + \Big(\frac{1}{n}- p_i\Big) \cdot \alpha + p_i \cdot S\alpha^2 + (1-p_i) \cdot \frac{\alpha^2 S}{n^2} \right)^b \notag \\
 & \leq \left( 1 + \Big( \frac{1}{n} - p_i\Big) \cdot \alpha + 2 \cdot p_i \cdot S\alpha^2 \right)^b \label{eq:psi_batched_i},
\end{align}
using in $(a)$ that the weights $W$ are independent given $\mathfrak{F}^t$, in $(b)$ \cref{lem:bounded_weight_moment} and in $(c)$ the binomial theorem. So,
\begin{align*}
\lefteqn{ \ex{\Psi_i^{t+b} \mid \mathfrak{F}^t} } \\
 & \stackrel{(a)}{\leq} \Psi_i^t \cdot e^{(\frac{1}{n} - p_i) \cdot \alpha \cdot b + 2 \cdot p_i \cdot S\alpha^2 \cdot b}\\
 & \stackrel{(b)}{\leq} \Psi_i^t \cdot \left(1 + \Big(\frac{1}{n} - p_i\Big) \cdot \alpha \cdot b + 2 \cdot p_i \cdot S \alpha^2 \cdot b + \Big(\Big(\frac{1}{n} - p_i\Big) \cdot \alpha \cdot b + 2 \cdot p_i \cdot S \alpha^2 \cdot b \Big)^2 \right) \\
 & \stackrel{(c)}{\leq} \Psi_i^t \cdot \left(1 + \Big(\frac{1}{n} - p_i\Big) \cdot \alpha \cdot b + 2 \cdot \frac{CS}{n}\cdot \alpha^2 \cdot b + \left(\frac{2CS}{n} \cdot \alpha \cdot b \right)^2 \right) \\
 & \leq \Psi_i^t \cdot \left(1 + \Big(\frac{1}{n} - p_i\Big) \cdot \alpha \cdot b + 5\cdot C^2\cdot S^2 \cdot \frac{b}{n} \cdot \frac{\alpha^2}{n} \cdot b \right),
\end{align*}
using in $(a)$ that $1 + y \leq e^y$ for any $y$, in $(b)$ that $e^y \leq 1 + y + y^2$ for $y \leq 1.75$ and that $(\frac{1}{n} - p_i) \cdot \alpha \cdot b + 2 \cdot p_i \cdot S\alpha^2 \cdot b \leq \frac{C}{n} \cdot  \alpha \cdot b + 2 \cdot \frac{C}{n} \cdot S \alpha^2 \cdot b \leq 2CS \cdot \alpha \cdot \frac{b}{n} \leq 1$, since $\alpha \leq \frac{n}{2CSb}$.
\end{proof}

We are now ready to apply \cref{cor:main_ptw} for $\kappa = b$.

\begin{thm} \label{thm:weak_gap_bound}
Consider any process $p$ satisfying conditions  $\mathcal{C}_1$ for constant $\delta \in (0, 1)$ and (not necessarily constant) $\eps \in (0,1)$ as well as condition $\mathcal{C}_2$ for some constant $C > 1$. Further, consider the batched setting with any $b \geq n$ and a weight distribution satisfying \cref{lem:bounded_weight_moment} with constant $S \geq 1$.
Then there exists a constant $k := k(\delta, C, S) > 0$, such that for any $m \geq 0$ being a multiple of $b$,
\[
\Pro{\max_{i \in [n]} |y_i^m| \leq k \cdot \frac{1}{\epsilon} \cdot \frac{b}{n} \cdot \log n } \geq 1 - n^{-2}.
\]
\end{thm}
\begin{rem}
The same upper bound as in \cref{thm:weak_gap_bound} holds also for processes with a time-dependent probability vector $p^t$, as long as for all $t$ being multiplies of $b$, the probability vector $p^t$ satisfies $\mathcal{C}_1$ and $\mathcal{C}_2$ for the same $\epsilon, \delta$ and $C$.
\end{rem}
\begin{rem}
The same upper bound as in \cref{thm:weak_gap_bound} holds also for processes with random tie-breaking and a probability vector $p$ satisfying the preconditions of \cref{lem:batching_pot_changes}. The reason for this is that $(i)$ averaging probabilities in \cref{eq:averaging_pi} can only reduce the maximum entry, i.e. $\max_{i \in [n]} \tilde{p}_i(x^t) \leq \max_{i \in [n]} p_i$, so it satisfies $\mathcal{C}_1$ and $(ii)$ moving probability between bins $i, j$ with $x_i = x_j$ (and thus $\Phi_i^t = \Phi_j^t$ and $\Psi_i^t = \Psi_j^t$), implies that the aggregate upper bounds in \eqref{eq:one} and \eqref{eq:two} remain the same.
\end{rem}
\begin{proof}[Proof of \cref{thm:weak_gap_bound}]
By \cref{lem:batching_pot_changes}, the preconditions of \cref{cor:main_ptw} are satisfied for $K := 5 \cdot C^2 \cdot S^2 \cdot \frac{b}{n}$, $\kappa := b$ and $\alpha := \frac{\eps\delta}{8K}$. Hence, there exists a constant $c := c(\delta) > 0$ such that for any step $m \geq 0$ which is a multiple of $b$,
\[
\Ex{\Gamma^m} \leq \frac{8c}{\delta} \cdot n.
\]
Hence, by Markov's inequality
\[
\Pro{\Gamma^m \leq \frac{8c}{\delta} \cdot n^3} \geq 1 - n^{-2}.
\]
To prove the claim, note that when $\{ \Gamma^m \leq \frac{8c}{\delta} \cdot n^3 \}$ holds, then also,
\[
\max_{i \in [n]} |y_i^m| \leq \frac{1}{\alpha} \cdot \left( \log \Big( \frac{8c}{\delta}\Big) + 3 \cdot \log n \right) \leq 4 \cdot \frac{8 \cdot 5 \cdot C^2 \cdot S^2}{\eps \delta} \cdot \frac{b}{n} \cdot \log n. \qedhere
\]
\end{proof}

\section{An Improved Upper Bound for Batch Sizes \texorpdfstring{$n \leq b \leq n^3$}{n <= b <= n\^3}} \label{sec:refined}

In this section, we will prove an improved (and tight) upper bound of $\Oh(b/n + \log n)$ on the gap for the weighted batched setting with batch size $n \leq b \leq n^3$. We will be assuming processes satisfying conditions  $\mathcal{C}_1$ for constant $\delta > 0$ and constant $\epsilon > 0$ and $\mathcal{C}_2$ with constant $C>1$.

\begin{restatable}{thm}{StrongGapBound}\label{thm:strong_gap_bound}
Consider any process $p$ satisfying conditions  $\mathcal{C}_1$ for constant $\delta \in (0, 1)$ and constant $\eps \in (0,1)$ as well as $\mathcal{C}_2$ for some constant $C > 1$. Further, consider the batched setting with any $n \leq b \leq n^3$ and a weight distribution satisfying \cref{lem:bounded_weight_moment} with constant $S \geq \max(1, 1/\lambda)$.
Then, there is a constant $\kappa := \kappa(\delta,\eps, C, S) > 0$, such that for any $m \geq 0$ being a multiple of $b$,
\[
\Pro{ y_1^m \leq \kappa \cdot \left(\frac{b}{n} + \log n \right) } \geq 1 - n^{-2}.
\]
\end{restatable}
\begin{rem}
The same gap bound holds also for processes with a time-dependent probability vector $p^t$, as long as for all $t$ being a multiple of $b$, the probability vector $p^t$ satisfies $\mathcal{C}_1$ and $\mathcal{C}_2$ for the same $\epsilon, \delta$ and $C$.
\end{rem}

There are two key steps in the proof:

\textbf{Step 1:} Similar to the analysis in~\cite[Theorem 5.3]{LS21}, we will be using two instances of the $\Gamma$ potential defined in \cref{sec:weak_bound} for $\alpha := \frac{\eps\delta}{40 \cdot C^2 \cdot S^2} \cdot \min\big(\frac{1}{\log n}, \frac{n}{b} \big)$. The second instance $\tilde{\Gamma}$ has a smaller smoothing factor $\tilde{\alpha} := \frac{\alpha}{8 \cdot 30}$, %
\[
\tilde{\Gamma}^t := \sum_{i = 1}^n \left( e^{\tilde{\alpha} y_i^t} + e^{-\tilde{\alpha} y_i^t} \right).
\]
So, in particular $\tilde{\Gamma}^t \leq \Gamma^t$ holds. Note that by varying $b \in [n, n \log n]$, both smoothing factors do not change, but this will not affect the upper bound, as we shall see below. %

We will show that \Whp~$\tilde{\Gamma} = \Oh(n)$ for $\log^3 n$ batches. 

\begin{restatable}{lem}{GammaLinearWhp} \label{lem:gamma_linear_whp}
Let $\tilde{c} := 2 \cdot \frac{8c}{\delta}$ where $c := c(\delta) > 0$ is the constant from \cref{cor:main_ptw}. Then, for any $t \geq 0$ being a multiple of $b$,
\[
\Pro{ \bigcap_{j \in [0, \log^3 n]} \left\lbrace \tilde{\Gamma}^{t + j \cdot b} \leq \tilde{c} \cdot n \right\rbrace } \geq 1 - n^{-3}.
\]
\end{restatable}

We prove this by conditioning on $\Gamma^t = \poly(n)$ which implies that $\Delta\tilde{\Gamma}^{t+1} = \Oh(\frac{n}{b} \cdot n^{1/4})$ (\cref{lem:gamma_1_poly_implies}~$(ii)$). This in turn allows us to apply a bounded difference inequality (\cref{lem:kutlin_3_3}) to prove concentration for $\tilde{\Gamma}$. %
The complete proof is given in \cref{sec:gamma_linear_whp}.

\textbf{Step 2:} We start by exploiting that conditioning on $\{ \tilde{\Gamma}^{t + j \cdot b} \leq \tilde{c} \cdot n \}$, the number of bins with load at least $k := \frac{1}{\tilde{\alpha}} \cdot \log(\tilde{c} / \delta) = \Theta(\max(b/n, \log n))$ is at most $\delta n$. We define the following potential function which only takes bins into account that are overloaded by at least $k$ balls:
\[
\Lambda^t := \sum_{i : y_i^t \geq k} \Lambda_i^t \cdot e^{\gamma \cdot (y_i^t - k)},
\]
where $\gamma := \min\big(\frac{\eps}{4CS}, \frac{n \log n}{b} \big)$. This means that when $\{ \tilde{\Gamma}^{t_0 + j \cdot b} \leq \tilde{c} \cdot n \}$ holds, the probability of allocating to one of these bins is $p_i \leq \frac{1-\eps}{n}$, because of condition $\mathcal{C}_2$. Hence, the potential drops in expectation (\cref{lem:lambda_drops}) and this means that \Whp~$\Lambda^m = \poly(n)$, implying an $\Oh(k + \gamma^{-1} \cdot \log n) = \Oh(b/n + \log n)$ gap.

\subsection{Step 1: \texorpdfstring{$\tilde{\Gamma}$}{tilde{Gamma}} is linear w.h.p.} \label{sec:gamma_linear_whp}

In this subsection, we will prove \cref{lem:gamma_linear_whp}. In \cref{sec:step_1_preliminaries}, we prove some properties of the $\Gamma$ and $\tilde{\Gamma}$ potential and in \cref{sec:gamma_linear_whp_complete} we combine these to show that \Whp~$\tilde{\Gamma}^t = \Oh(n)$ for $\log^3 n$ batches.

\subsubsection{Preliminaries} \label{sec:step_1_preliminaries}

For constant $\lambda > 0$ as defined in \cref{sec:batched_model}, we define the following event, for any round $t \geq 0$\[
\mathcal{H}^t := \left\lbrace  w^t \leq \frac{15}{\lambda} \cdot \log n \right\rbrace,
\]
which means that the weight of the ball sampled in round $t$ is $\Oh(\log n)$.

\begin{lem} \label{lem:many_h_i}
For any $b \leq n^{3}$ and for any $t \geq 0$,
\[
\Pro{\bigcap_{s \in [t, t + 2b \log^3 n]} \mathcal{H}^s} \geq 1 - n^{-10}
\]
\end{lem}
\begin{proof}
Since $w^t$ is sampled according to $W$ with $\ex{ e^{\lambda W}} < \infty$, by \cref{lem:chernoff}
\[
\Pro{w^t \geq \frac{15}{\lambda} \cdot \log n} \leq n^{-14}.
\]
By taking the union bound over the interval $[t, t + 2b \log^3 n]$ and since $b \leq n^3$ we get the conclusion.
\end{proof}

We will now show that when $\Gamma^t = \poly(n)$ and $\mathcal{H}^t$ holds, then $\Delta\tilde{\Gamma}^{t+1}$ is small.

\begin{lem} \label{lem:gamma_1_poly_implies}
Let  $\tilde{c} := \tilde{c}(\delta) > 0$ be the constant defined in \cref{lem:gamma_linear_whp}. For any $t \geq 0$, where $\Gamma^{t} \leq 2\tilde{c} \cdot n^{26}$ and $\mathcal{H}^t$ holds, then $(i)$ $\tilde{\Gamma}^t \leq n^{5/4}$ and $(ii)$
$
 | \tilde{\Gamma}^{t+1} - \tilde{\Gamma}^{t} | \leq \frac{n}{b} \cdot n^{1/4}
$. Further, let $\tilde{x}^t$ be the load vector obtained by moving the $t$-th ball of the load vector $x^t$ to some other bin, then $(iii)$ when $\mathcal{H}^t$ holds holds, $\Gamma^t(\tilde{x}^t) \leq 2 \cdot \Gamma^t(x^t)$.
\end{lem}
\begin{proof}
For any bin $i \in [n]$,
\begin{align*}
\Gamma^{t} \leq 2\tilde{c} \cdot n^{26} &\Rightarrow e^{\alpha\cdot y_i^{t}} + e^{-\alpha \cdot y_i^{t}} \leq \tilde{c} \cdot n^{26} \\ &\Rightarrow 
y_i^t \leq \frac{27}{\alpha} \log n \, \wedge \,
 -y_i^t \leq \frac{27}{\alpha} \log n,
\end{align*}
where in the second implication we used $\log (2\tilde{c}) + \frac{26}{\alpha} \log n \leq \frac{27}{\alpha} \log n$, for sufficiently large $n$.

This implies that
\begin{equation} \label{eq:gamma_i_bound}
\tilde{\Gamma}_i^t \leq e^{\tilde{\alpha} y_i^t} + e^{-\tilde{\alpha} y_i^t} \leq 2 \cdot e^{\tilde{\alpha} \cdot \frac{27}{\alpha} \log n} \leq 2 \cdot n^{1/8},
\end{equation}
using that $\tilde{\alpha} := \frac{\alpha}{8 \cdot 30}$. Hence, by aggregating, we get the first claim $\Gamma^t = \sum_{i = 1}^n \Gamma_i^t \leq 2 \cdot n \cdot n^{1/8} \leq n^{5/4}$.

We now proceed to the second statement. Consider the change for the bin $j \in [n]$ where the ball was allocated. 
Since $\tilde{\alpha} < \frac{1}{40\cdot S \cdot \log n}$ and $S > \frac{1}{\lambda}$, we have $\tilde{\alpha} \cdot \frac{15}{\lambda} \cdot \log n \leq 1$ and so by a Taylor estimate, $e^{\tilde{\alpha} \cdot \frac{15}{\lambda} \cdot \log n} \leq 1 + 2 \cdot \tilde{\alpha} \cdot \frac{15}{\lambda} \cdot \log n$. If $j \in [n]$ is an overloaded bin, then
\begin{align*}
|\Delta\tilde{\Gamma}_j^t| & \leq \tilde{\Gamma}_j^t \cdot e^{\tilde{\alpha} \cdot \frac{15}{\lambda} \cdot \log n}- \tilde{\Gamma}_j^t \leq \tilde{\Gamma}_j^t \cdot \Big( 1 + \tilde{\alpha} \cdot \frac{30}{\lambda} \cdot \log n \Big)- \tilde{\Gamma}_j^t \\
& = \tilde{\Gamma}_j^t \cdot \tilde{\alpha} \cdot \frac{30}{\lambda} \cdot \log n \leq \frac{n}{b} \cdot n^{1/8} \cdot \log n,
\end{align*}
using \cref{eq:gamma_i_bound} and $\tilde{\alpha} \leq \frac{\eps\delta}{40 \cdot C^2 \cdot S^2} \cdot \frac{n}{b}$.
Similarly, if $j$ is underloaded, then
\begin{align*}
|\Delta\tilde{\Gamma}_j^t| & \leq \tilde{\Gamma}_j^t - \tilde{\Gamma}_j^t \cdot e^{-\tilde{\alpha} \cdot \frac{15}{\lambda} \cdot \log n} \leq \tilde{\Gamma}_j^t - \tilde{\Gamma}_j^t \cdot \Big( 1 - \tilde{\alpha} \cdot \frac{30}{\lambda} \cdot \log n\Big) \\
 & = \tilde{\Gamma}_j^t \cdot \tilde{\alpha} \cdot \frac{30}{\lambda} \cdot \log n \leq \frac{n}{b} \cdot n^{1/8} \cdot \log n.
\end{align*}

The rest of the bins' contributions change due to the change in the average load. In particular, for any overloaded bin $i \in [n] \setminus \{ j \}$, %
\begin{align*}
|\Delta\tilde{\Gamma}_i^t| 
 & \leq \tilde{\Gamma}_i^t \cdot e^{ \tilde{\alpha} \cdot \frac{15}{\lambda} \cdot \frac{\log n}{n}}- \tilde{\Gamma}_i^t \leq \tilde{\Gamma}_i^t \cdot \Big( 1 + 2 \cdot \tilde{\alpha} \cdot \frac{15}{\lambda} \cdot \frac{\log n}{n}\Big)- \tilde{\Gamma}_i^t \\
 & = \tilde{\Gamma}_i^t \cdot \tilde{\alpha} \cdot \frac{30}{\lambda} \cdot \frac{\log n}{n} \leq \frac{1}{b} \cdot \log n \cdot n^{1/8}.
\end{align*}
Similarly, for an underloaded bin $i \in [n] \setminus \{ j \}$,
\begin{align*}
|\Delta\tilde{\Gamma}_i^t| 
 & \leq \tilde{\Gamma}_i^t - \tilde{\Gamma}_i^t \cdot e^{-\tilde{\alpha} \cdot \frac{15}{\lambda} \cdot \frac{\log n}{n}} \leq \tilde{\Gamma}_i^t - \tilde{\Gamma}_i^t \cdot \Big( 1 - 2 \cdot \tilde{\alpha} \cdot \frac{15}{\lambda} \cdot \frac{\log n}{n}\Big) \\
 & = \tilde{\Gamma}_i^t \cdot \tilde{\alpha} \cdot \frac{30}{\lambda} \cdot \frac{\log n}{n} \leq \frac{1}{b} \cdot \log n \cdot n^{1/8}.
\end{align*}
Hence, aggregating over all bins\[
|\Delta\tilde{\Gamma}^{t+1}| \leq |\Delta\Gamma_j^{t+1}| + \sum_{i \in [n] \setminus \{ j \}} |\Delta\Gamma_i^{t+1}| \leq 2 \cdot \frac{n}{b} \cdot n^{1/8} \cdot \log n + n \cdot \frac{1}{b} \cdot \log n \cdot n^{1/8} \leq \frac{n}{b} \cdot n^{1/4},
\]
for sufficiently large $n$.

For statement $(iii)$, let $i, j \in [n]$ be the differing bins between $x^t$ and $\tilde{x}^t$. Then since $\mathcal{H}^t$ holds, $w^t \leq \frac{15}{\lambda} \cdot \log n$, so 
\[
\Gamma_i(\tilde{x}^t) \leq e^{\alpha w^t} \cdot \Gamma_i^t(x^t) \leq 2 \cdot \Gamma_i^t(x^t),
\]
since $\alpha < \frac{1}{40 \cdot S \cdot \log n}$ and $S > 1/\lambda$. Similarly, for $j$,
\[
\Gamma_j(\tilde{x}^t) \leq e^{\alpha w^t} \cdot \Gamma_j^t(x^t) \leq 2 \cdot \Gamma_j^t(x^t),
\]
Hence, 
\[
\Gamma^t(\tilde{x}^t) = \sum_{k = 1}^n \Gamma_k^t(\tilde{x}^t) \leq \sum_{k = 1}^n 2 \cdot \Gamma_k^t(x^t) = 2 \cdot \Gamma^t(x^t).
\]
\end{proof}

Next, we will show that $\ex{\tilde{\Gamma}} = \Oh(n)$ and that when $\tilde{\Gamma}$ is sufficiently large, it drops in expectation over the next batch.

\begin{lem}
\label{lem:large_gamma_exponential_drop}
Let $\tilde{c} := 2 \cdot \frac{8c}{\delta}$ where $c := c(\delta) > 0$ is the constant from \cref{cor:main_ptw}. Then, for any step $t\geq 0$ being a multiple of $b$, \[
(i) \quad \ex{\tilde{\Gamma}^t} \leq \frac{\tilde{c}}{2} \cdot n,
\quad \text{ and } \quad (ii) \quad \ex{\Gamma^t} \leq \frac{\tilde{c}}{2} \cdot n.
\]
Further, there exists a constant $\tilde{c}_1 := \tilde{c}_1(\eps, \delta) > 0$ such that \[
(iii) \quad \Ex{\tilde{\Gamma}^{t+b} \,\, \Big\vert\,\, \mathfrak{F}^{t},\tilde{\Gamma}^t \geq \tilde{c} \cdot n} \leq 
\Big(1-\frac{\tilde{c}_1}{\log n}\Big) \cdot \tilde{\Gamma}^{t},
\]
and 
\[
(iv) \quad \Ex{\tilde{\Gamma}^{t+b} \,\,\Big\vert\,\, \mathfrak{F}^{t},\tilde{\Gamma}^t \leq \tilde{c} \cdot n} \leq 
\tilde{c} \cdot n - \frac{n}{\log^2 n}.
\]
\end{lem}
\begin{proof}
The first two statements follow immediately by \cref{lem:batching_pot_changes} and \cref{cor:main_ptw}, by setting $\tilde{c} := 16 c/\delta$, since $c := c(\delta) > 0$.

Also, using \cref{lem:batching_pot_changes} and \cref{cor:main_ptw} for $\tilde{\alpha}$, we get that for any $t \geq 0$,
\begin{equation} \label{eq:tilde_gamma_drop}
\Ex{\tilde{\Gamma}^{t+b} \mid \mathfrak{F}^t} \leq \tilde{\Gamma}^t \cdot \Big(1 - \frac{\eps\delta}{8} \cdot \frac{b}{n} \cdot \tilde{\alpha}\Big) + c \cdot b \cdot \eps \cdot \tilde{\alpha}.
\end{equation}

Let $\tilde{c}_3 := \frac{1}{2} \cdot \frac{\eps\delta}{8} \cdot \frac{b}{n} \cdot \tilde{\alpha} \geq \tilde{c}_1/\log n$, for some constant $\tilde{c}_1 > 0$ since $\tilde{\alpha} = \Theta(\min(n/b, 1/\log n))$ and $\eps$ is constant. When $\tilde{\Gamma}^{t} \geq \tilde{c} \cdot n$, then \cref{eq:tilde_gamma_drop} yields,
\begin{align*}
\Ex{\tilde{\Gamma}^{t+b} \,\Big\vert\, \mathfrak{F}^t, \tilde{\Gamma}^{t} \geq \tilde{c} \cdot n} & \leq \tilde{\Gamma}^{t} \cdot \Big(1 - 2 \cdot \tilde{c}_3 \Big) + c \cdot b \cdot \eps \cdot \tilde{\alpha} \\
 & \leq
\tilde{\Gamma}^{t} - \tilde{c}_3 \cdot \tilde{\Gamma}^{t} + \Big(c \cdot b \cdot \eps \cdot \tilde{\alpha}- \tilde{c}_3 \cdot \tilde{\Gamma}^{t}\Big)
 \\  & \leq
\tilde{\Gamma}^{t} - \tilde{c}_3 \cdot \tilde{\Gamma}^{t} + \Big(c \cdot b \cdot \eps \cdot \tilde{\alpha}- \frac{1}{2} \cdot \frac{\eps\delta}{8} \cdot \frac{b}{n} \cdot \tilde{\alpha} \cdot \frac{16 c}{\delta} \cdot n \Big)
 \\ &\leq
\Big(1-\frac{\tilde{c}_1}{\log n} \Big) \cdot \tilde{\Gamma}^{t}.
\end{align*}
Similarly, when $\Gamma^t < \tilde{c} \cdot n$, \cref{eq:tilde_gamma_drop} yields,
\begin{align*}
\Ex{\tilde{\Gamma}^{t+b} \,\Big\vert\, \mathfrak{F}^t, \tilde{\Gamma}^{t} < \tilde{c} \cdot n} & \leq \tilde{c} \cdot n \cdot \Big(1 - 2 \cdot \tilde{c}_3\Big) + c \cdot b \cdot \eps \cdot \tilde{\alpha} \\
 & = \tilde{c} \cdot n - \tilde{c} \cdot \tilde{c}_3 \cdot n + \Big(c \cdot b \cdot \eps \cdot \tilde{\alpha}- \tilde{c} \cdot \tilde{c}_3 \cdot n \Big)
 \\ &\leq
\tilde{c} \cdot n - \frac{\tilde{c} \cdot \tilde{c}_1}{\log n} \leq \tilde{c} \cdot n - \frac{n}{\log^2 n}. \qedhere
\end{align*}
\end{proof}

In the next lemma, we show that \Whp~$\Gamma$ is $\poly(n)$ for every step in an interval of length $2 b \log^3 n$.

\begin{lem} \label{lem:gamma_continuous}
Let $\tilde{c} := 2 \cdot \frac{8c}{\delta}$ be the constant defined in \cref{lem:large_gamma_exponential_drop}. For any $n \leq b \leq n^3$ and for any $t \geq 0$ being a multiple of $b$,
\[
\Pro{ \bigcap_{s \in [t, t + 2b \log^ 3 n]} \left\{ \Gamma^{s} \leq \tilde{c} \cdot n^{26} \right\} } \geq 1 - n^{-10}.
\]
\end{lem}
\begin{proof}
Using \cref{lem:large_gamma_exponential_drop}~$(i)$, Markov's inequality and the union bound, we have for any $t \geq 0$,
\begin{equation} \label{eq:base_union_bound}
\Pro{ \bigcap_{s \in [0, 2 \log^ 3 n]} \left\{ \Gamma^{t + s \cdot b} \leq \tilde{c} \cdot n^{12} \right\} } \geq 1 - \frac{2\log^3 n}{n^{11}}.
\end{equation}
Given that $\Gamma^{t + s \cdot b} \leq \tilde{c} \cdot n^{12}$, we will upper bound $\Gamma^{t + s \cdot b + r}$ for any $r \in [0, b)$. To this end, we will upper bound for each bin $i \in [n]$ the terms $\Phi_i^{t + s \cdot b + r}$ and $\Psi_i^{t + s \cdot b + r}$ separately. Proceeding using \cref{eq:phi_batched_i} in \cref{lem:batching_pot_changes},
\begin{align*}
\Ex{\Phi_i^{t + s \cdot b + r} \mid \Phi_i^{t + s \cdot b}}& \leq  \Phi_i^{t + s \cdot b} \cdot \Big( 1 + \Big(p_i - \frac{1}{n}\Big) \cdot \alpha + 2 \cdot p_i \cdot S\alpha^2 \Big)^r \\ & \leq  \Phi_i^{t + s \cdot b} \cdot \Big( 1 + \frac{C\alpha}{n} + 2 \cdot \frac{C}{n} \cdot S\alpha^2 \Big)^r \\ &\stackrel{(a)}{\leq} \Phi_i^{t + s \cdot b} \cdot \Big( 1 + \frac{2C\alpha}{n} \Big)^r \\ &\leq \Phi_i^{t + s \cdot b} \cdot e^{2\alpha C \cdot \frac{r}{n}} \leq \Phi_i^{t + s \cdot b} \cdot e^{2\alpha C \cdot \frac{b}{n}} \stackrel{(b)}{\leq} 2 \cdot \Phi_i^{t + s \cdot b},
\end{align*}
using in $(a)$ that $\alpha \leq \frac{\eps\delta}{40 \cdot C^2 \cdot S^2} \leq \frac{2}{S}$ and in $(b)$ that $\alpha \leq \frac{\eps\delta}{40 \cdot C^2 \cdot S^2} \cdot \frac{n}{b} \leq \frac{1}{4C} \cdot \frac{n}{b}$. Similarly, using \cref{eq:psi_batched_i} in \cref{lem:batching_pot_changes},
\begin{align*}
\Ex{\Psi_i^{t + s \cdot b + r} \mid \Psi_i^{t + s \cdot b}} &\leq \Psi_i^{t + s \cdot b} \cdot \Big( 1 + \Big( \frac{1}{n} - p_i\Big) \cdot \alpha + 2 \cdot p_i \cdot S\alpha^2 \Big)^r \\ &\leq \Psi_i^{t + s \cdot b} \cdot \Big( 1 + \frac{C \alpha}{n} + 2 \cdot \frac{C}{\alpha} \cdot S\alpha^2 \Big)^r \\ &\stackrel{(a)}{\leq} \Psi_i^{t + s \cdot b} \cdot \Big( 1 + \frac{2C\alpha}{n} \Big)^r \\ &\leq \Psi_i^{t + s \cdot b} \cdot e^{2\alpha C \cdot \frac{r}{n}} \leq \Psi_i^{t + s \cdot b} \cdot e^{2\alpha C \cdot \frac{b}{n}} \stackrel{(b)}{\leq} 2 \cdot \Psi_i^{t + s \cdot b},
\end{align*}
using in $(a)$ that $\alpha \leq \frac{\eps\delta}{40 \cdot C^2 \cdot S^2} \leq \frac{2}{S}$ and in $(b)$ that $\alpha \leq \frac{\eps\delta}{40 \cdot C^2 \cdot S^2} \cdot \frac{n}{b} \leq \frac{1}{4C} \cdot \frac{n}{b}$. Hence, aggregating over the bins,
\[
\Ex{\Gamma^{t + s \cdot b + r} \mid \Gamma^{t + s \cdot b}} \leq 2 \cdot \Gamma^{t + s \cdot b}.
\]
Applying Markov's inequality, for any $r \in [0, b)$,
\[
\Pro{\Gamma^{t + s \cdot b + r} \leq n^{14} \cdot \Gamma^{t + s \cdot b}} \geq 1 - 2 \cdot n^{-14}.
\]
Hence, by a union bound over the $2b \cdot \log^3 n \leq 2 \cdot n^3 \cdot \log^3 n$ possible rounds for $s \in [0, 2\log^3 n]$ and $r \in [0, b)$,%
\begin{align}
\Pro{ \bigcap_{r \in [0, b]}\bigcap_{s\in [0, 2\log^3 n]} \left\{ \Gamma^{t + s \cdot b + r} \leq n^{14} \cdot \Gamma^{t + s \cdot b} \right\} } \geq 1 - 2 \cdot n^{-14} \cdot 2 b \log^3 n \geq 1 - \frac{1}{2} \cdot n^{-10}. \label{eq:double_intersection_lb}
\end{align}
Finally, taking the union bound of \cref{eq:base_union_bound} and \cref{eq:double_intersection_lb}, we conclude
\begin{align*}
\lefteqn{\Pro{ \bigcap_{s \in [t, t + 2b \log^ 3 n]} \left\{ \Gamma^{s} \leq \tilde{c} \cdot n^{26} \right\} }}  \\
& \geq \Pro{ \bigcap_{r \in [0, b]}\bigcap_{s\in [0, 2\log^3 n]} \left\{ \Gamma^{t + s \cdot b + r} \leq n^{14} \cdot \Gamma^{t + s \cdot b} \right\} \cap \bigcap_{s \in [0, 2 \log^ 3 n]} \left\{ \Gamma^{t + s \cdot b} \leq \tilde{c} \cdot n^{12} \right\} }\\
 & \geq 1 - \frac{1}{2} \cdot n^{-10} - \frac{2\log^3 n}{n^{11}} \geq 1 - n^{-10}. \qedhere
\end{align*}
\end{proof}

We will now show that \Whp~there is a step where the exponential potential $\tilde{\Gamma}$ becomes $\Oh(n)$.

\begin{lem} \label{lem:gamma_1_poly_n_implies_gamma_2_linear_whp}
Let $\tilde{c} := 2 \cdot \frac{8c}{\delta}$ be the constant defined in \cref{lem:large_gamma_exponential_drop}. For any $t \geq 0$ being a multiple of $b$,
\[
\Pro{\bigcup_{s \in [0, b \log^3 n]} \left\lbrace \tilde{\Gamma}^{t + s \cdot b} \leq \tilde{c} \cdot n \right\rbrace} \geq 1 - 2 \cdot n^{-8}.
\]
\end{lem}

\begin{proof}
By \cref{lem:large_gamma_exponential_drop}~$(ii)$, using Markov's inequality at time $t$ being a multiple of $b$, we have
\begin{equation} \label{eq:basic_markov_tilde_gamma}
\Pro{\tilde{\Gamma}^{t} \leq \tilde{c} \cdot n^9} \geq 1 - n^{-8}.
\end{equation}
Assuming $\tilde{\Gamma}^{t} \leq \tilde{c} \cdot n^9$ and by \cref{lem:large_gamma_exponential_drop}~$(iii)$ if at some step $\tilde{\Gamma}^{r} > \tilde{c} \cdot n$, then \[
\Ex{\tilde{\Gamma}^{r+1} \,\mid\, \tilde{\Gamma}^{r}, \tilde{\Gamma}^{r} > \tilde{c} \cdot n} \leq 
\Big(1-\frac{\tilde{c}_1}{\log n}\Big) \cdot \tilde{\Gamma}^{r},
\]
where $\tilde{c}_1 > 0$ is some constant.
For any $r \in [0, \log^3 n]$, we define the ``killed'' potential function,
\[
\widehat{\Gamma}^{t + r\cdot b} := \tilde{\Gamma}^{t + r \cdot b} \cdot \mathbf{1}_{\bigcap_{ \tilde{r} \in [0, r)} \{ \tilde{\Gamma}^{t + \tilde{r} \cdot b} > \tilde{c} \cdot n\} }.
\]
This potential satisfies the drop inequality of \cref{lem:large_gamma_exponential_drop}  without any condition on the value of $\widehat{\Gamma}^{r}$, that is,
\[
\Ex{\widehat{\Gamma}^{t + (r+1) \cdot b} \,\mid\, \widehat{\Gamma}^{t + r \cdot b}} \leq \Big(1-\frac{\tilde{c}_1}{\log n}\Big) \cdot \widehat{\Gamma}^{t + r \cdot b}.
\]
Inductively applying this for $\log^3 n$ batches, and since $\tilde{c}_1 := \tilde{c}_1(\eps, \delta) > 0$ is a constant, %
\[
\Ex{\left. \widehat{\Gamma}^{t + (\log^3 n) \cdot b} \,\, \right\vert \,\, \tilde{\Gamma}^{t}} \leq \Big(1-\frac{\tilde{c}_1}{\log n}\Big)^{\log^3 n} \cdot \tilde{\Gamma}^{t} \leq e^{-\tilde{c}_1 \cdot \log^2 n} \cdot \tilde{c} \cdot n^{9} < n^{-7}. 
\]
So by Markov's inequality, \[
 \Pro{\widehat{\Gamma}^{t + (\log^3 n) \cdot b} \geq n  \, \mid\, \tilde{\Gamma}^{t} \leq \tilde{c} \cdot n^9} \leq n^{-{8}}
\]
By union bound with \cref{eq:basic_markov_tilde_gamma},
\begin{align*}
 \Pro{\widehat{\Gamma}^{t + (\log^3 n) \cdot b} \geq n} 
 & = \Pro{\widehat{\Gamma}^{t + (\log^3 n) \cdot b} \geq n\, \mid\, \tilde{\Gamma}^{t} \leq \tilde{c} \cdot n^9} \cdot \Pro{\tilde{\Gamma}^{t} \leq \tilde{c} \cdot n^9} + \Pro{\tilde{\Gamma}^{t} > \tilde{c} \cdot n^9} \\
 & < n^{-8} + n^{-8} = 2 \cdot n^{-8}.
\end{align*}
Due to the definition of $\tilde{\Gamma}$, at any step $t \geq 0$, deterministically $\tilde{\Gamma}^t \geq 2n$. So,
we conclude that~w.p.~at least $1 - 2 \cdot n^{-8}$, there must be at least one time step $r \in [0, \log^3 n]$, with $\widehat{\Gamma}^{t + r \cdot b} = 0$ and so $\tilde{\Gamma}^{t + s \cdot b} \leq \tilde{c} \cdot n$ for some $s \in [0, \log^3 n]$.
\end{proof}

\subsubsection{Completing the Proof of Lemma~\ref{lem:gamma_linear_whp}}\label{sec:gamma_linear_whp_complete}

We are now ready to prove \cref{lem:gamma_linear_whp}, using a method of bounded differences with a bad event \cref{lem:kutlin_3_3} (\cite[Theorem 3.3]{K02}).

\GammaLinearWhp*

\begin{proof}

Our starting point is to apply
 \cref{lem:gamma_1_poly_n_implies_gamma_2_linear_whp}, which proves that there is at least one time step $t + \rho \cdot b \in [t - b\log^3 n, t]$ with $\rho \in [-\log^3 n, 0]$ such that the potential $\tilde{\Gamma}$ is small,
 \begin{align} \label{eq:starting_point}
\Pro{\bigcup_{\rho \in [- \log^3 n, 0]} \left\lbrace \tilde{\Gamma}^{t + \rho \cdot b} \leq \tilde{c} \cdot n \right\rbrace } &\geq 1 - 2 \cdot n^{-8}.
 \end{align}
Note that if $t < b \cdot \log^3 n$, then deterministically $\tilde{\Gamma}^0 = 2n \leq \tilde{c} \cdot n$ (which corresponds to $\rho = -t/b$).

We are now going to apply the concentration inequality \cref{lem:kutlin_3_3} to each of the batches starting at $t + \rho \cdot b, \ldots, t + (\log^3 n) \cdot b$ and show that the potential remains $ \leq \tilde{c} \cdot n$ at the end of each batch. In particular, we will show that for any $\tilde{r} \in [\rho, \log^3 n]$, for $r = t + b \cdot \tilde{r}$,
\[
\Pro{\tilde{\Gamma}^{r+b} > \tilde{c} \cdot n \mid \mathfrak{F}^r, \tilde{\Gamma}^r \leq \tilde{c} \cdot n} \leq 3 \cdot n^{-4}.
\]

We will show this by applying \cref{lem:kutlin_3_3} for all steps of the batch $[r, r + b]$. We define the good event \[
\mathcal{G}_r := \mathcal{G}_r^{r+b} := \bigcap_{s \in [r, r + b]} \Big( \{ \Gamma^s \leq \tilde{c} \cdot n^{26} \} \cap \mathcal{H}^s \Big),
\] 
and $\mathcal{B}_r := (\mathcal{G}_r)^c$ the bad event. Using a union bound over \cref{lem:many_h_i} and \cref{lem:gamma_continuous},
\begin{align} \label{eq:bad_event_union_bound}
\Pro{ \bigcap_{s \in [t-b \log^3 n, t+ b \log^ 3 n]} \left( \{ \Gamma^{s} \leq \tilde{c} \cdot n^{26} \} \cap \mathcal{H}^s \right) } \geq 1 - 2n^{-10}.
\end{align}

Consider any $u \in [r, r + b]$. Further, we define the slightly weaker good event, $\tilde{\mathcal{G}}_r^u := \bigcap_{s \in [r, u]} \Big( \{ \Gamma^s \leq 2\tilde{c} \cdot n^{26} \} \cap \mathcal{H}^s \Big)$ and the ``killed'' potential, %
\[
\widehat{\Gamma}_r^u := \mathbf{1}_{\tilde{\mathcal{G}}_r^u} \cdot \tilde{\Gamma}^u.
\]
We will show that the sequence $\widehat{\Gamma}_r^r, \ldots , \widehat{\Gamma}_r^{r + b}$ is strongly difference-bounded by $(n^{5/4}, \frac{n}{b} \cdot n^{1/4}, 2 \cdot n^{-10})$ (\cref{def:strongly_dif_bounded}).

Let $\omega \in [n]^{b}$ be an allocation vector encoding the allocations made in $[r, r + b]$. Let $\omega'$ be an allocating vector resulting from $\omega$ by changing one arbitrary allocation. It follows that,
\begin{align*}
|\widehat{\Gamma}_r^{r+b}(\omega) - \widehat{\Gamma}_r^{r+b}(\omega')| &\leq \max_{\tilde{\omega}} \widehat{\Gamma}_r^{r+b}(\tilde{\omega}) - \min_{\tilde{\omega}} \widehat{\Gamma}_r^{r+b}(\tilde{\omega}) \\
&\leq \max_{\tilde{\omega} \in \tilde{\mathcal{G}}_{r}^{r+b}} \tilde{\Gamma}_r^{r+b}(\tilde{\omega}) -0 
\\ & \leq n^{5/4},
\end{align*}
where in the last inequality we used \cref{lem:gamma_1_poly_implies}~$(i)$ that for any $\tilde{\omega} \in \tilde{\mathcal{G}}_r^{r+b}$, we have $\widehat{\Gamma}_r^{r+b}(\tilde{\omega}) \leq \tilde{\Gamma}_r^{r+b}(\tilde{\omega}) \leq n^{5/4}$. 

We will now derive a refined bound by additionally assuming that $\omega \in \mathcal{G}_r$. Then, for any $u \in [r, r + b]$,
\[
\Gamma^{r+u}(\omega') \leq 2 \cdot \Gamma^{r+u}(\omega) \leq 2 \tilde{c} \cdot n^{26},
\]
where the first inequality is by \cref{lem:gamma_1_poly_implies}~$(iii)$. 
Hence $\omega' \in \tilde{\mathcal{G}}_r^{r+b}$, so $\mathbf{1}_{\tilde{\mathcal{G}}_r^{r+b}(\omega')} = 1$ and  $\widehat{\Gamma}_r^{r+b}(\omega') 
= \tilde{\Gamma}_r^{r+b}(\omega')$. Similarly, for $\omega \in \mathcal{G}_r \subseteq \tilde{\mathcal{G}}_r^{r+b}$, we have $\widehat{\Gamma}_r^{r+b}(\omega) 
= \tilde{\Gamma}_r^{r+b}(\omega)$ and by \cref{lem:gamma_1_poly_implies}~$(ii)$,
\[ 
|\widehat{\Gamma}_r^{r+b}(\omega) - \widehat{\Gamma}_r^{r+b}(\omega')| = |\tilde{\Gamma}^{r+b}(\omega) - \tilde{\Gamma}^{r+b}(\omega')| \leq \frac{n}{b} \cdot n^{1/4}.
\]

Within a single batch all allocations are independent, so we apply \cref{lem:kutlin_3_3}, choosing $\gamma_k := \frac{1}{b}$ and $N := b$, which states that for any $\lambda > 0$ and $\mu := \Ex{\widehat{\Gamma}_r^{r + b} > \mu + \lambda \mid \mathfrak{F}^r, \tilde{\Gamma}^r \leq \tilde{c} \cdot n}$,
\[
\Pro{\widehat{\Gamma}_r^{r + b} > \mu + \lambda \mid \mathfrak{F}^r, \tilde{\Gamma}^r \leq \tilde{c} \cdot n} \leq \exp\left( - \frac{\lambda^2}{2 \cdot \sum_{k = 1}^{b} (\frac{n}{b} \cdot n^{1/4} + n^{5/4} \cdot \frac{1}{b})^2 }\right) + 2 \cdot n^{-10} \cdot \sum_{k = 1}^b b. 
\]
By \cref{lem:large_gamma_exponential_drop}~$(iv)$, we have $\mu \leq \ex{\widehat{\Gamma}_r^{r+b} \mid \tilde{\Gamma}^r < \tilde{c} \cdot n}  \leq \ex{\tilde{\Gamma}_r^{r+b} \mid \tilde{\Gamma}^r < \tilde{c} \cdot n} \leq \tilde{c} \cdot n - n/\log^2 n$. Hence, for $\lambda := n / \log^2 n$, since $n \leq b \leq n^3$, we have
\begin{align*}
\Pro{\widehat{\Gamma}_r^{r+b} > \tilde{c} \cdot n \mid \mathfrak{F}^r, \tilde{\Gamma}^r \leq  \tilde{c} \cdot n}  & \leq \exp\left( - \frac{n^2/\log^4 n}{2 \cdot b \cdot (2 \cdot \frac{n}{b} \cdot n^{1/4})^2 }\right) + 2n^{-10} \cdot b^2 \\
& \leq \exp\left( - \frac{b}{8 \cdot \log^4 n \cdot n^{1/2}}\right) + 2n^{-10} \cdot n^6 \leq 3 \cdot n^{-4}. 
\end{align*}
Let $\mathcal{K}_{\rho}^{\tilde{r}} := \mathcal{G}_{\rho}^{t + \tilde{r} \cdot b} \cap \{ \tilde{\Gamma}^{t + \rho \cdot b} \leq \tilde{c} \cdot n \}$ for $\tilde{r} \in [\rho, \log^3 n]$. For any $\tilde{r} \geq \rho$, since $\mathcal{K}_\rho^{\tilde{r}+1} \subseteq \mathcal{K}_\rho^{\tilde{r}}$, we have%
\begin{align} \label{eq:k_killed}
\Pro{\mathbf{1}_{\mathcal{K}_\rho^{\tilde{r}+1}} \cdot \widehat{\Gamma}^{t + (\tilde{r}+1) \cdot b} >  \tilde{c} \cdot n \mid \mathfrak{F}^r, \mathbf{1}_{\mathcal{K}_\rho^{\tilde{r}}} \cdot \widehat{\Gamma}^{t + \tilde{r} \cdot b} \leq  \tilde{c} \cdot n} \leq 3 \cdot n^{-4}.
\end{align}
By union bound of \cref{eq:starting_point} and \cref{eq:bad_event_union_bound}, 
\begin{align} 
\Pro{\bigcup_{\rho \in [-\log^3 n]} \mathcal{K}_{\rho}^{\log^3 n}}
 & \geq \Pro{\mathcal{G}_{-\log^3 n}^{\log^3 n} \cap \bigcup_{\rho \in [- \log^3 n, 0]} \left\lbrace \tilde{\Gamma}^{t + \rho \cdot b} \leq \tilde{c} \cdot n \right\rbrace} \notag \\
 & \geq 1 - 2 \cdot n^{-8} - 2\cdot n^{-10} \geq 1 - 3 \cdot n^{-8}.\label{eq:exists_k_event}
\end{align}
Let $\mathcal{A} := \bigcap_{\tilde{r} \in [0, \log^3 n]} \left\lbrace \tilde{\Gamma}^{t + \tilde{r} \cdot b} \leq \tilde{c} \cdot n \right\rbrace$ and $\mathcal{A}_{\rho} := \bigcap_{\tilde{r} \in [\rho, \log^3 n]} \left\lbrace \widehat{\Gamma}^{t + \tilde{r} \cdot b} \cdot \mathbf{1}_{\mathcal{K}_\rho^{\tilde{r}}} \leq \tilde{c} \cdot n \right\rbrace$. Then, 
\begin{align*}
\Pro{\mathcal{A}_\rho \mid \tilde{\Gamma}^{t + \rho \cdot b} \leq \tilde{c} \cdot n } 
  & \geq \prod_{\tilde{r} \in [\rho, \log^3 n - 1]} \mathbf{Pr} \left[ \bigcap_{\tilde{s} \in [\rho+1, \tilde{r}+1]} \left\lbrace \mathbf{1}_{\mathcal{K}_\rho^{\tilde{s}}} \cdot \widehat{\Gamma}^{t + \tilde{s} \cdot b} \leq \tilde{c} \cdot n \right\rbrace \right. \\
  &  \qquad \qquad \qquad  \left. \bigg\vert \, \bigcap_{\tilde{s} \in [\rho+1, \tilde{r} - 1]} \left\lbrace \mathbf{1}_{\mathcal{K}_\rho^{\tilde{s}}} \cdot \widehat{\Gamma}^{t + \tilde{s} \cdot b} \leq \tilde{c} \cdot n \right\rbrace, \mathbf{1}_{\mathcal{K}_\rho^{\tilde{s}}} \cdot \widehat{\Gamma}^{t + \tilde{r} \cdot b} \leq \tilde{c} \cdot n\right] \\
 & \geq \prod_{\tilde{r} \in [\rho, \log^3 n - 1]} \Pro{\mathbf{1}_{\mathcal{K}_\rho^{\tilde{r}+1}} \cdot \widehat{\Gamma}^{\tilde{r}+b} > \tilde{c} \cdot n \mid \mathfrak{F}^{t + \tilde{r} \cdot b}, \mathbf{1}_{\mathcal{K}_\rho^{\tilde{r}}} \cdot \widehat{\Gamma}^{t + \tilde{r} \cdot b} \leq  \tilde{c} \cdot n} \\
 & \geq (1 - 3n^{-4})^{2\log^3 n} \geq 1 - 6 \cdot n^{-4} \cdot \log^3 n,
\end{align*}
where in the last inequality we have used \cref{eq:k_killed} and the fact $\rho \geq -\log^3 n$. So,
\begin{align}
\Pro{\mathcal{A}_\rho} 
 & = \Pro{\mathcal{A}_\rho \mid \tilde{\Gamma}^{t + \rho \cdot b} \leq \tilde{c} \cdot n} \cdot \Pro{\tilde{\Gamma}^{t + \rho \cdot b} \leq \tilde{c} \cdot n} + 1 \cdot \Pro{\neg \left\{ \tilde{\Gamma}^{t + \rho \cdot b} \leq \tilde{c} \cdot n \right\}} \notag \\
 & \geq 1 - 6 \cdot n^{-4} \cdot \log^3 n. \label{eq:event_a_rho} 
\end{align}
Note that for any $\rho \in [-\log^3 n, 0]$, we have that $\mathcal{A}_\rho \cap \mathcal{K}_\rho^{\log^3 n} \subseteq \mathcal{A}$. Hence we conclude by the union bound of \cref{eq:exists_k_event} and \cref{eq:event_a_rho}, that
\[
\Pro{\mathcal{A}} \geq \Pro{\bigcup_{\rho \in [-\log^3 n, 0]} \mathcal{K}_{\rho}^{\log^3 n} \cap \bigcap_{\rho \in [-\log^3 n, 0]} \mathcal{A}_\rho} \geq 1  - 3 \cdot n^{-8} - 6 \cdot n^{-4} \cdot \log^6 n \geq 1- n^{-3}.
\]

\end{proof}

\subsection{Step 2: Completing the Proof of Theorem~\ref{thm:strong_gap_bound}}\label{sec:step_two}

Recall the definition of the $\Lambda$ potential function,
\[
\Lambda^t := \sum_{i : y_i^t \geq k} \Lambda_i^t \cdot e^{\gamma \cdot (y_i^t - k)},
\]
where $\gamma := \min\big(\frac{\eps}{4CS}, \frac{n \log n}{b} \big)$ and $k := \frac{1}{\tilde{\alpha}} \cdot \log(\tilde{c} / \delta) = \Theta(\max(b/n, \log n))$.

We will now show that when $\tilde{\Gamma}^t = \Oh(n)$, the stronger potential function $\Lambda^t$ drops in expectation. This will allow us to prove that $\Lambda^m = \poly(n)$ and deduce that \Whp~$\Gap(m) = \Oh(b/n + \log n)$.

\begin{lem} \label{lem:lambda_drops}
Let $\tilde{c} := 2 \cdot \frac{8c}{\delta}$ be the constant defined in \cref{lem:large_gamma_exponential_drop}. For any $t \geq 0$ being a multiple of $b$,
\[
\Ex{\Lambda^{t+b} \mid \mathfrak{F}^t, \tilde{\Gamma}^t \leq \tilde{c} \cdot n} \leq \Lambda^t \cdot e^{-\frac{\gamma\eps}{2n} \cdot b} + n \cdot e^{\frac{C \gamma}{n} \cdot b}.
\]
\end{lem}
\begin{proof}
 When $\{ \tilde{\Gamma}^t \leq \tilde{c} \cdot n \}$ holds, the number of bins with load $y_i^t \geq k$ is at most 
\[
\tilde{c} \cdot n \cdot e^{- \tilde{\alpha} \cdot k} = \tilde{c} \cdot n \cdot e^{-\log(\tilde{c}/\delta)} = \delta \cdot n.
\]
For any bin $i$ with $y_i^t \geq k$, we get as in \cref{eq:phi_batched_i},%
\begin{align*}
\Ex{\Lambda_i^{t + b} \mid \mathfrak{F}^t} 
 & \leq \Lambda_i^t \cdot \left( 1 + \Big(p_i - \frac{1}{n}\Big) \cdot \gamma + 2 \cdot p_i \cdot S\gamma^2 \right)^b.
\end{align*}
For these bins $i \leq \delta n$ and so by \cref{lem:quasilem2}, the upper bound on $\Ex{\Lambda^{t+b} \mid \mathfrak{F}^t, \tilde{\Gamma}^t \leq \tilde{c} \cdot n}$ is maximized when $p_i = \frac{1 - \eps}{n}$, so
\begin{align*}
\sum_{i : y_i^t \geq k} \Ex{\Lambda_i^{t + b} \mid \mathfrak{F}^t, \tilde{\Gamma}^t \leq \tilde{c} \cdot n} 
 & \stackrel{(a)}{\leq} \Lambda_i^t \cdot \left( 1 - \frac{\gamma\eps}{n} + 2 \cdot C \cdot S \cdot \frac{\gamma^2}{n} \right)^b \\
 & \stackrel{(b)}{\leq} \Lambda_i^t \cdot \left( 1 - \frac{\gamma\eps}{2n}\right)^b \stackrel{(c)}{\leq} \Lambda_i^t \cdot e^{-\frac{\gamma\eps}{2n} \cdot b},
\end{align*}
using in $(a)$ that $p_i \leq C/n$, in $(b)$ that $\gamma\leq \frac{\eps}{4CS}$ and in $(c)$ that $1 + z \leq e^z$ for any $z$. For the rest of the bins with $i > \delta n$,
\begin{align*}
\Ex{\Lambda_i^{t + b} \mid \mathfrak{F}^t} 
 & \leq \Lambda_i^t \cdot \left( 1 + \Big(p_i - \frac{1}{n}\Big) \cdot \gamma + 2 \cdot p_i \cdot S\gamma^2 \right)^b \\
 & \stackrel{(a)}{\leq} \Lambda_i^t \cdot \left( 1 + \frac{C}{n} \cdot \gamma - \frac{1}{n} \cdot \gamma + 2 \cdot \frac{C}{n} \cdot S\gamma^2 \right)^b \\
 & \stackrel{(b)}{\leq} \Lambda_i^t \cdot \left( 1 + \frac{C \cdot \gamma}{n}\right)^b \stackrel{(c)}{\leq} \left( 1 + \frac{C \cdot \gamma}{n}\right)^b \stackrel{(d)}{\leq} e^{\frac{C\gamma}{n} \cdot b},
\end{align*}
using in $(a)$ that $p_i \leq C/n$, in $(b)$ that $\gamma \leq \frac{\eps}{4CS}$, in $(c)$ that $\Lambda_i^t \leq 1$ and in $(d)$ that $1 + z \leq e^z$ for any $z$.

Aggregating the contributions of all bins, 
\begin{align*}
\Ex{\Lambda^{t+b} \mid \mathfrak{F}^t, \tilde{\Gamma}^t \leq \tilde{c} \cdot n} 
 &\leq \sum_{ i : y_i^t \geq k} \Lambda_i^t \cdot e^{-\frac{\gamma\eps}{2n} \cdot b} + \sum_{ i : y_i^t < k} e^{\frac{C\gamma}{n} \cdot b} \\ 
 &\leq \Lambda^t \cdot e^{-\frac{\gamma\eps}{2n} \cdot b} + n \cdot e^{\frac{C\gamma}{n} \cdot b}.\qedhere
\end{align*}
\end{proof}

\StrongGapBound*
\begin{proof}
Consider first the case when $m \geq b \cdot \log^3 n$. Let $t_0 = m - b \cdot \log^3 n$. Let $\mathcal{E}^t := \{ \tilde{\Gamma}^{t} \leq \tilde{c} \cdot n\}$. Then using \cref{lem:gamma_linear_whp},
\begin{equation} \label{eq:eps_interval}
\Pro{ \bigcap_{j \in [0, \log^3 n]} \mathcal{E}^{t_0 + j \cdot b} } \geq 1 - n^{-3}.
\end{equation}
We define the killed potential $\tilde{\Lambda}$, with $\tilde{\Lambda}^{t_0} := \Lambda^{t_0}$ and for $j > 0$,
\[
\tilde{\Lambda}^{t_0 + j \cdot b} := \mathbf{1}_{\cap_{s \in [0, j)} \mathcal{E}^{t_0 + s \cdot b}} \cdot \Lambda^{t_0 + j \cdot b}.
\]
By \cref{lem:lambda_drops}, we have
\begin{align} \label{eq:lambda_drop}
\Ex{\tilde{\Lambda}^{t_0+(j+1)\cdot b} \mid \mathfrak{F}^{t_0+j \cdot b}} \leq \tilde{\Lambda}^{t_0 + j \cdot b} \cdot e^{-\frac{\gamma\eps}{2n} \cdot b} + n \cdot e^{\frac{C \gamma}{n} \cdot b}.
\end{align}
Assuming $\mathcal{E}^{t_0}$ holds, we have %
\[
y_1^{t_0} \leq \frac{1}{\tilde{\alpha}} \cdot ( \log \tilde{c} + \log n) \leq \frac{2}{\tilde{\alpha}} \cdot \log n,
\]
for sufficiently large $n$. Hence for some constant $\kappa_1 >0$,
\[
\tilde{\Lambda}^{t_0} \leq n \cdot e^{\gamma \cdot y_1^{t_0}} \leq e^{\kappa \log^2 n}.
\]

Applying \cref{lem:geometric_arithmetic} to \cref{eq:lambda_drop} with $a := e^{-\frac{\gamma\eps}{2n} \cdot b}$ and $b := n \cdot e^{\frac{C \gamma}{n} \cdot b}$ for $\log^3 n$ steps, %
\begin{align}
\Ex{\tilde{\Lambda}^{m} \mid \mathfrak{F}^{t_0}, \tilde{\Lambda}^{t_0} \leq e^{\kappa_1 \log^2 n}} 
& \leq e^{\kappa_1 \log^2 n} \cdot a^{\log^3 n} + \frac{b}{1 - a} \notag \\
&\stackrel{(a)}{\leq} 1 + 1.5 \cdot b \leq  2 \cdot n \cdot e^{\frac{C\gamma}{n} \cdot b} \stackrel{(b)}{\leq} 2 \cdot n^{1 + \kappa_2} \label{eq:poly_n_expectation}.
\end{align}
using in $(a)$ that $\frac{\gamma\eps}{2n} \cdot b = \Omega(1)$ and $a$ a constant $< 1$ and in $(b)$ that  $\frac{C \gamma}{n} \cdot b \leq \kappa_2 \cdot \log n$ for some constant $\kappa_2 > 0$, since $\gamma = \min\big(\frac{\eps}{4CS}, \frac{n \log n}{b} \big)$.

By Markov's inequality, we have
\begin{align*}
\Pro{\tilde{\Lambda}^{m} \leq 2 \cdot n^{4 + \kappa_2} \mid \mathfrak{F}^{t_0}, \tilde{\Lambda}^{t_0} \leq e^{\kappa_1 \log^2 n} } \geq 1 - n^{-3}.
\end{align*}
Hence, by \cref{eq:eps_interval},
\begin{align} \label{eq:tilde_lambda_poly_n}
\Pro{\tilde{\Lambda}^{m} \leq 2 \cdot n^{4 + \kappa_2}} = \Pro{\tilde{\Lambda}^{m} \leq 2 \cdot n^{4 + \kappa_2} \mid \mathcal{E}^{t_0}} \cdot \Pro{\mathcal{E}^{t_0}} \geq (1 - n^{-3}) \cdot (1 - n^{-3}).
\end{align}
Combining \cref{eq:eps_interval} and \cref{eq:tilde_lambda_poly_n}, we have
\begin{align*}
\Pro{\Lambda^{m} \leq 2 \cdot n^{4 + \kappa_2}}  &\geq \Pro{\left\lbrace\tilde{\Lambda}^{m} \leq 2 \cdot n^{4 + \kappa_2} \right\rbrace \cap \bigcap_{j \in [0, \log^3 n]} \mathcal{E}^{t_0 + j \cdot b}} 
\\ %
\\ &\geq (1 - n^{-3}) \cdot (1 - n^{-3}) - n^{-3} \geq 1 - n^{-2}.
\end{align*}
Finally, $\{ \Lambda^{m} \leq 2 \cdot n^{4 + \kappa_2} \}$ implies
\[
y_1^{m} \leq k + \frac{\log 2}{\gamma} + \frac{1}{\gamma} \cdot (4 + \kappa_2) \cdot \log n = \Oh( b/n + \log n),
\]
since $\gamma = \min\big(\frac{\eps}{4CS}, \frac{n \log n}{b} \big)$ and $\Theta(\max(b/n, \log n))$, so the claim follows.

For the case when $m < b \cdot \log^3 n$, note that $\tilde{\Lambda}^{t_0} = 2n$ deterministically, which is a stronger starting point in \cref{eq:poly_n_expectation} to prove that $\ex{\Lambda^m} \leq 2 \cdot n^{1 + \kappa_2}$, which in turn implies the gap bound.
\end{proof}

\section{Application to Graphical Allocations and (\texorpdfstring{$1+\beta$}{1+beta})-process}\label{sec:graphical}

In~\cite{PTW15}, the authors proved several bounds on the gap for the $(1+\beta)$ process (in the setting without batches) where balls are sampled from a weight distribution with constant $\lambda > 0$ as defined in~\cref{sec:batched_model}. 
In the second part of~\cite{PTW15}, the authors used a majorization argument to deduce gap bounds for graphical balanced allocation. However, due to the involved majorization argument not working for weights, all results for graphical allocation in~\cite{PTW15} assume balls are unweighted. This lack of results for weighted graphical allocations is summarized as Open Question 1 in~\cite{PTW15}. By leveraging the results in previous sections, we are able to fill this ``gap''. 

For a $d$-regular (and connected) graph, let us define the conductance as:
\[
 \Phi(G) := \min_{S \subseteq V \colon 1 \leq |S| \leq n/2} \frac{|E(S,V \setminus S)|}{|S| \cdot d}.
\]
We will call a family of graphs an \emph{expander}, if $\Phi$ is at least a constant bounded below from $0$ (as $n \rightarrow \infty$).

\begin{restatable}{lem}{LemExpansion}\label{lem:expansion}
Consider \Graphical on a $d$-regular graph with conductance $\Phi$ with batch size $b = 1$. Then for any $\mathfrak{F}^t$, the probability vector $p^t$ in round $t \geq 0$ satisfies for all $1 \leq k \leq n/2$,
\[
 \sum_{i=1}^k p_i^t \leq (1-\Phi) \cdot \frac{k}{n},
\]
and similarly, for any $n/2+1 \leq k \leq n$,
\[
 \sum_{i=k}^{n} p_i^t \geq (1+\Phi) \cdot \frac{n-k+1}{n}.
\]
Further, $\max_{i \in [n]} p_i^t \leq  \frac{d}{n}$. Thus, the vector $p^t$ satisfies condition $\mathcal{C}_1$ with $\delta=1/2$, $\epsilon=\Phi$ and condition $\mathcal{C}_2$ with $C=d$.
\end{restatable}
The proof of this lemma closely follows~\cite[Proof of Theorem~3.2]{PTW15}.

\begin{proof}
Fix any load vector $x^t$ in round $t$. Consider any $1 \leq k \leq n/2$. Let $S_k$ be the $k$ ``heaviest'' bins with the largest load. Hence in order to allocate a ball into $S_k$, both endpoints of the sampled edge must be in $S_k$, and hence
\begin{align*}
 \sum_{i=1}^k p_i^t &= \frac{2 \cdot |E(S_k,S_k)|}{2 \cdot |E|} \\
 &= \frac{|E(S_k,V) | -|E(S_k,V \setminus S_k)|}{2 \cdot |E|} \\
 &\leq \frac{|S_k| \cdot d - |S_k| \cdot \Phi \cdot d}{nd} = (1-\Phi) \cdot \frac{k}{n},
\end{align*}
where the inequality used the definition of conductance $\Phi$. Hence $p^t$ satisfies condition $\mathcal{C}_1$ with $\epsilon=\Phi$. 
Now, we will consider the suffix sums for $n/2+1 \leq k \leq n$. We start by upper bounding the prefix sum up to $k-1$,
\begin{align*}
\sum_{i=1}^{k-1} p_i^t &\leq \frac{|S_k| \cdot d - |V \setminus S_k| \cdot \Phi \cdot d}{nd} \\
&\leq \frac{(k-1) \cdot d - (n - k + 1) \cdot \Phi \cdot d}{nd} \\
&= \frac{(k-1) - (n - k + 1) \cdot \Phi}{n},
\end{align*}
where the inequality used our assumption that $G$ has expansion $\Phi > 0$.
Hence the suffix sum is
\[
\sum_{i = k}^n p_i^t = 1 - \sum_{i=1}^{k-1} p_i^t \geq 1 - \frac{(k-1) - (n - k + 1) \cdot \Phi}{n} = (1+\Phi) \cdot \frac{n - k + 1}{n}.
\]

Finally, we also know that for any bin $i \in [n]$,
$
p_i^{t} \leq \frac{d}{n},
$
since in the worst-case we allocate a ball into bin $i$ whenever an edge incident to $i$ is chosen. 
\end{proof}

The next result is for \Graphical in the unbatched setting.
\begin{thm}\label{thm:unbatched_graphical}
Consider \Graphical on a $d$-regular graph with conductance $\Phi > 0$. Further, consider the non-batched setting, i.e., $b = 1$ and assume that balls are sampled from a weight distribution with constant $\lambda > 0$. Then there is a constant $k := k(\lambda) > 0$ such that for any $m \geq 0$,
\[
\Pro{\max_{i \in [n]} |y_i^m| \leq k \cdot \frac{d}{\Phi} \cdot \log n } \geq 1-n^{-2}.
\]
\end{thm}
\begin{proof}
Using~\cite[Lemma 2.1]{PTW15}, we have for $\alpha = \frac{\Phi}{16 \cdot d \cdot S}$ and some constant $S := S(\lambda) > 1$,
\[
\Ex{\Delta\Phi_i^{t+1} \mid x^t} \leq \Phi_i^t \cdot \Big(\Big(p_i - \frac{1}{n}\Big) \cdot \alpha + dS \cdot \frac{\alpha^2}{n}\Big),
\]
and using~\cite[Lemma 2.3]{PTW15},
\[
\Ex{\Delta\Psi_i^{t+1} \mid x^t} \leq  \Psi_i^t \cdot \Big(\Big(\frac{1}{n} - p_i\Big) \cdot \alpha + dS \cdot \frac{\alpha^2}{n}\Big).
\]
Hence, applying \cref{cor:main_ptw} for $\eps := \Phi$, $\delta := 1/2$, we get for any $m \geq 0$, 
\[
\Ex{\Gamma^m} \leq \frac{8c}{\delta} \cdot n,
\]
for some constant $c := c(\delta) > 0$. Hence, by Markov's inequality
\[
\Pro{\Gamma^m \leq \frac{8c}{\delta} \cdot n^3} \geq 1 - n^{-2}.
\]
The event $\{ \Gamma^m \leq \frac{8c}{\delta} \cdot n^3 \}$ implies that
\[
\max_{i \in [n]} |y_i^m| \leq \log \Big(\frac{8c}{\delta} \Big) + 3 \cdot \frac{16 \cdot d \cdot S}{\Phi} \cdot \log n \leq k \cdot \frac{d}{\Phi} \cdot \log n,
\]
for some constant $k := k(\lambda) > 0$.
\end{proof}
The next result is the batched version of \cref{thm:unbatched_graphical}.
\begin{thm}\label{thm:graphical}
Consider \Graphical on a $d$-regular graph with conductance $\Phi > 0$. Further, consider the batched setting with $b \geq n$ and assume that balls are sampled from a weight distribution with constant $\lambda > 0$. Then there is a constant $k := k(\lambda) > 0$ such that it holds for any $m \geq 0$ being a multiple of $b$,
\[
\Pro{\max_{i \in [n]} |y_i^m| \leq k \cdot \frac{d^2}{\Phi} \cdot \frac{b}{n} \cdot \log n } \geq 1-n^{-2}.
\]
Further, if the conductance $\Phi$ is lower bounded by a constant $>0$ (i.e., $G$ is an expander), $d > 0$ is constant and $n \leq b \leq n^3$, then there is a constant $k := k(\lambda, d) >0$ such that for any $m \geq 0$ being a multiple of $b$,
\[
 \Pro{ y_1^m \leq k \cdot \left( \frac{b}{n} + \log n \right)} \geq 1-n^{-2}.
\]
\end{thm}
Note that our first gap bound for constant $d > 0$, generalizes~\cite[Theorem 3.2]{PTW15}, which is a gap bound of $\Oh(\frac{\log n}{\Phi})$ in the setting without batches and weights. Similarly, our second result extends the $\Oh(\log n)$ bound from~\cite{PTW15} for expanders, and proves that the same gap bound applies in the weighted batched setting with any $b=\Oh(n \log n)$. 

\begin{proof}
The first result follows directly from \cref{lem:expansion} and \cref{thm:weak_gap_bound}.
For the second result, $\epsilon=\Phi$ is a constant $>0$, and we can apply the refined gap bound from~\cref{thm:strong_gap_bound}.
\end{proof}

Next we improve the upper bound on the gap for $(1+\beta)$ for very small $\beta$. In~\cite[Corollary 2.12]{PTW15}, it was shown that this gap is $\Oh(\log n/\beta + \log(1/\beta)/\beta)$. For $1/\beta = n^{\omega(1)}$, the second term dominates. We improve this gap bound to $\Oh(\log n/\beta)$. This is tight up to multiplicative constants for $\beta \leq 1/2$, due to a lower bound of $\Omega( \log n/\beta)$ as shown in~\cite[Section~4]{PTW15}.
\begin{restatable}{thm}{OnePlusBetaGap}\label{thm:OnePlusBetaGap}
Consider the $(1+\beta)$ process for any $\beta \in (0,1]$. Then there exists a constant $k > 0$, such for any $m \geq 1$,
\[
\Pro{\Gap(m) \leq k \cdot \frac{\log n}{\beta} } \geq 1 - n^{-2}.
\]
\end{restatable}

\begin{proof}
Using~\cite[Lemma 2.1]{PTW15}, we have for $\alpha = \frac{\beta}{4}$ and some constant $S > 1$,
\[
\Ex{\Delta\Phi_i^{t+1} \mid x^t} \leq \Phi_i^t \cdot \Big(\Big(p_i - \frac{1}{n}\Big) \cdot \alpha + 2S \cdot \frac{\alpha^2}{n}\Big),
\]
and using~\cite[Lemma 2.3]{PTW15},
\[
\Ex{\Delta\Psi_i^{t+1} \mid x^t} \leq  \Psi_i^t \cdot \Big(\Big(\frac{1}{n} - p_i\Big) \cdot \alpha + 2S \cdot \frac{\alpha^2}{n}\Big).
\]
By \cref{pro:verify}, the $1+\beta$ process satisfies the $\mathcal{C}_1$ condition for $\eps = \frac{\beta}{4}$ and $\delta = \frac{1}{4}$. By \cref{cor:main_ptw}, there exists $c := c(\delta) > 0$ such that for any $m \geq 1$
\[
\Ex{\Gamma^m} \leq \frac{8c}{\delta} \cdot n.
\]
Hence, using Markov's inequality
\[
\Pro{\Gamma^m \leq n^4} \geq 1 - n^{-2}.
\]
Note that when $\{ \Gamma^m \leq n^4 \}$ holds, we have \[
\Gap(m) \leq \frac{4}{\alpha} \cdot \log n = 4 \cdot \frac{8\cdot (2S)}{\beta \delta} \cdot \log n = \Oh\Big(\frac{\log n}{\beta}\Big). \qedhere
\]
\end{proof}

\section{Lower Bounds}\label{sec:lower_bounds}
For the lower bounds we always assume that balls are unweighted (or equivalently, have unit weight). We recall the following result which assumes no batching, i.e., balls are allocated sequentially using perfect knowledge about the bin loads.
\begin{lem}[{\cite[Theorem 10.4]{LSS21}}] \label{lem:lower_bound_b_n}
Consider any allocation process in the unweighted setting with probability vector $q$ with $\min_{i \in [n]} q_i \geq C/n$ for some constant $C > 0$. Then there exists a constant $k > 0$, such that for $m = \Theta(n \log n)$,
\[
\Pro{\Gap(m) \geq k \cdot \log n} \geq 1 - n^{-2}.
\]
\end{lem}

We use the following majorization result from~\cite{LSS21} (see also~\cite[Section 3]{PTW15}).

\begin{lem}[Lemma 4.13 in~\cite{LSS21}]\label{lem:majorization}
Consider two allocation processes $Q$ and $P$. The allocation process $Q$ uses at each round a fixed allocation distribution $q$. The allocation process $P$ uses a time-dependent allocation distribution $p^{t}$, which may depend on $\mathfrak{F}^{t}$ but majorizes $q$ at each round $t \geq 0$. 
Let $y^{t}(Q)$ and $y^{t}(P)$ be the two normalized load vectors, sorted decreasingly. Then there is a coupling such that for all rounds $t \geq 0$, $y^{t}(P)$ majorizes $y^{t}(Q)$.
\end{lem}

Combining the two lemmas above, we can now prove a lower bound which holds for any batch size:
\begin{restatable}{pro}{LogLower}\label{pro:log_lower}
Consider any allocation process with probability vector $p$ with $\min_{i \in [n]} p_i \geq C/n$ for some constant $C > 0$, in the unweighted batched setting for any $b \geq 1$. Then there exists a constant $k > 0$ such that for $m = \Theta(n \log n)$, 
\[
\Pro{\Gap(m) \geq k\cdot \log n } \geq 1 - n^{-2}.
\]
\end{restatable}

Note that this statement applies to the $(1+\beta)$-process for constant $\beta \in (0,1)$ and $\Quantile(\delta)$ and constant $\delta \in (0,1)$, but it does not apply to \TwoChoice.

\begin{proof}
For the purpose of this lower bound derivation, we assume that the batched setting allocates all $m$ balls sequentially in rounds $t=1,2,\ldots,m$. As the load information does not get updated within each batch of size $b$, this means that the allocation made in each round is described by an allocation vector $p^t$, which depends on $t$ but also on the history of the process, i.e., $\mathfrak{F}^{t}$. 

Let $q$ be the vector $p$ sorted in non-decreasing order. Then, at each round $p^t$ majorizes $q$, since the outdated information implies that $p^t$ is a permutation of $q$.%

We apply Lemma~\ref{lem:majorization} with $p^t$ and $q$ as defined above. Hence for $t=m$, there is a coupling such that the load vector $y^{m}(P)$ majorizes $y^{m}(Q)$, in particular,
\[
 y_1^{m}(P) \geq y_1^{m}(Q),
\]
which is equivalent to $\Gap(P,m) \geq \Gap(Q,m)$. Hence the statement of the lemma follows by \cref{lem:lower_bound_b_n}.
\end{proof}

\begin{restatable}{pro}{LemBnLower}\label{pro:bn_lower}
Consider any allocation process with probability vector $p$ with  $\max_{i \in [n]} p_i \geq \frac{C}{n}$ for some $C > 1$, in the unweighted batched setting with $b \geq n \log n$.  Then, for $\gamma := \min(C - 1, 0.5)$, any bin $j=\argmax_{i \in [n]} p_i$ satisfies
\[
\Pro{ y_j^{b} \geq \frac{\gamma}{4} \cdot \frac{b}{n}} \geq 1 - n^{-\gamma^2/8}.
\]
\end{restatable}
\begin{proof}
For convenience, let us define $\gamma:=\min(C-1, 0.5)$, so  $\gamma \in (0,1/2)$.
Note that during the first batch consisting of $b \geq n \log n$ balls, the load vector is never updated and all balls are allocated using the same probability vector $p$. Hence
each ball will be allocated into some bin $i$ with probability $\frac{C}{n}$, independently. Let $X:=x_i^{b}=\sum_{j=1}^b X_j$, where the $X_j$'s are independent Bernoulli random variables with  $\Ex{X_j} \geq \frac{1+\gamma}{n}$. Hence $\Ex{X} \geq b \cdot \frac{1+\gamma}{n}$. Using the following Chernoff bound, which states that for any $\lambda > 0$,
\[
 \Pro{ X \leq (1-\lambda) \cdot \Ex{ X } } \leq \exp\left( -\lambda^2/ 2 \cdot \Ex{ X} \right).
\]
Picking $\lambda = \gamma/2$ implies
\[
 \Pro{ x_i^{b} \leq (1-\gamma/2) \cdot (1+\gamma) \cdot \frac{b}{n} } \leq \exp\left(-\frac{\gamma^2}{8} \cdot \frac{b}{n} \right) \leq n^{-\gamma^2/8},
\]
where the last inequality used our assumption that $b \geq n \log n$.
If $x_i^b \geq (1-\gamma/2) \cdot (1+\gamma) \cdot \frac{b}{n} $, then this implies for the normalised load,
\[
 y_j^b = x_i^b - \frac{b}{n} \geq \frac{\gamma}{2} \cdot \frac{b}{n} - \frac{\gamma^2}{2} \cdot \frac{b}{n} \geq \frac{\gamma}{4} \cdot \frac{b}{n},
\]
where the last inequality used $\gamma \leq 1/2$. %
\end{proof}

The lemma above can be applied to any process satisfying condition $\mathcal{C}_2$, so unlike \cref{pro:log_lower}, it applies to \TwoChoice.

Note that for \DChoice, $\max_{i \in [n]} p_i \approx \frac{d}{n}$. Hence this lower bound shows that, in sharp contrast to the classical setting without batches, that large values of $d$ lead to a worse performance. This is explained by the higher bias towards underloaded bins, which when given no (or outdated) information about the bins, will lead to a larger gap.

Let us remark that in the proof above, we assumed that the allocation process uses the same probability vector and bin labeling in all rounds of the same batch. In particular, this analysis does not apply to \TwoChoice with random tie-breaking. However, \TwoChoice with random tie-breaking will allocate all balls in the first batch following \OneChoice. Exploiting this, we can then prove that by the end of the batch, there is a unique bin which attains the minimum load if $b = \Omega(n \log n)$, which means for the second batch we can apply \cref{pro:bn_lower}, and conclude that a lower bound of $\Omega(\frac{b}{n})$ holds with constant probability $>0$.

First, we will make use of the following property of $n$ independent Poisson random variables:

\begin{lem} \label{lem:poisson_property}
Consider any $n \geq 2$ and  $\lambda \geq 16 \cdot \log n$. Let $X_1, \ldots , X_n$ be independent Poisson random variables with $X_i \sim \Pois(\lambda)$, and denote by for $Y_{(n)}, Y_{(n-1)}$ the smallest and second smallest of the $X_i$'s. Then there exist constants $\kappa_1, \kappa_2 > 0$ such that,
\[
\Pro{Y_{(n-1)} - Y_{(n)} \geq \kappa_1 \cdot \sqrt{\lambda/\log n}} \geq \kappa_2.
\]
\end{lem}
\begin{proof}
Let $X \sim \Pois(\lambda)$, where $\lambda := m/n \geq 16 \cdot \log n$. Let $k \geq 0$ be the minimal integer such that
\begin{align*}
 \Pro{ \Pois(\lambda) \leq k } \geq n^{-1}.
\end{align*}
By \cref{lem:chernoff_lower} for $\delta := \sqrt{4 \cdot \lambda^{-1} \cdot \log n}$, we have
\[
\Pro{X \leq \lambda - \sqrt{4 \cdot \lambda \cdot \log n}} \leq e^{-\lambda \cdot \delta^2/2} = e^{-2 \log n} = n^{-2}.
\]
Hence it follows that $k \geq \lambda - 2 \cdot \sqrt{ \lambda \cdot \log n}$.
Next note that
\begin{align}
 \frac{ \Pro{ \Pois(\lambda) = k+1 }}{ \Pro{ \Pois(\lambda) = k } }
 &= \frac{\lambda}{k+1}, \label{eq:ratio}
\end{align}
which, since $k \geq \frac{1}{2} \lambda$ (as $\lambda \geq 16 \log n $), also implies that
\begin{align*}
 \Pro{ \Pois(\lambda) \leq k } \leq 2 \cdot n^{-1}.
\end{align*}
Our next claim is that
\begin{align*}
 \Pro{ \Pois(\lambda) = k } \leq 2 \cdot n^{-1} \cdot 1/\sqrt{\lambda/\log n}.
\end{align*}
We will now derive this claim. 
We have
\begin{align*}
 2n^{-1} &\geq
 \Pro{ \Pois(\lambda) \leq k } \\
 &= \sum_{j=0}^{k} \Pro{ \Pois(\lambda) =j } \\
 &\stackrel{(a)}{=} \sum_{j=0}^{k} \Pro{ \Pois(\lambda) =k } \cdot \prod_{i=j}^{k-1} \frac{i}{\lambda} \\
  &\geq
   \sqrt{\lambda/\log n} \cdot \Pro{ \Pois(\lambda) =k } \cdot \left( \frac{k-\sqrt{\lambda}}{\lambda} \right)^{\sqrt{\lambda/\log n}} \\
  &\stackrel{(b)}{\geq} \sqrt{\lambda/\log n} \cdot \Pro{ \Pois(\lambda) =k } \cdot \left( \frac{\lambda-3 \sqrt{\lambda \log n}}{\lambda} \right)^{\sqrt{\lambda/\log n}} \\
   &= \sqrt{\lambda/\log n} \cdot \Pro{ \Pois(\lambda) =k } \cdot \left(1 - \frac{3}{\sqrt{ \lambda / \log n}} \right)^{\sqrt{\lambda/\log n}} \\
  &\stackrel{(c)}{\geq} \sqrt{\lambda/\log n} \cdot \Pro{ \Pois(\lambda) =k } \cdot c_1,
 \end{align*}
for some constant $c_1 > 0$, where in $(a)$ we used \cref{eq:ratio} and in $(b)$ we used that $k \geq \lambda - 2 \sqrt{\lambda \log n}$, and in $(c)$ that $\lambda \geq 16 \log n$.

Next we wish to upper bound
\begin{align*}
 \Pro{ \Pois(\lambda) \leq k + c_2 \cdot \sqrt{ \lambda / \log n} }, %
\end{align*}
for some constant $c_2 > 0$.
Note that
\begin{align}
\lefteqn{  \Pro{ \Pois(\lambda) \leq k + c_2 \cdot \sqrt{ \lambda/\log n}}} \notag \\ &\stackrel{(a)}{\leq}  \Pro{ \Pois(\lambda) \leq k }
 + \sum_{i=1}^{c_2 \sqrt{\lambda/\log n}} \frac{\lambda^i}{k \cdot (k+1) \cdot \ldots \cdot (k+i-1)} \cdot \Pro{ \Pois(\lambda) = k } \notag \\
 &\leq 2n^{-1} + c_2 \cdot \sqrt{\lambda/\log n} \cdot \frac{\lambda^{c_2 \sqrt{\lambda/\log n}}}{k^{c_2 \sqrt{\lambda/\log n}}} \cdot 2 \cdot n^{-1} \cdot 1/\sqrt{\lambda/\log n} \notag \\
 &= 2n^{-1} + 2c_2 \cdot \left(1 - \frac{1}{c \sqrt{\lambda \cdot \log n}} \right)^{-c_2 \sqrt{\lambda/\log n} } \cdot n^{-1} \notag \\
 &\leq c_3 \cdot n^{-1}, \label{eq:upper}
 \end{align}
 for another constant $c_3 > 0$, where $(a)$ is due to \cref{eq:ratio}.

We now use the principle of deferred decisions when exposing the $n$ independent Poisson variables with mean $\lambda$ denoted by $X_1,X_2,\ldots,X_n$ one by one. Let $\tau:= \min\{ j \colon X_j \leq k\}$.
With probability $1-(1-1/n)^n \geq 1-1/e$, we have $\tau < n$. Conditional on that, $X_{\tau+1},\ldots,X_{n}$ are still $n - \tau$ independent Poisson variables with mean $\lambda$. Due to \cref{eq:upper}, the probability that all of the following Poisson random variables are larger than $k + c_2 \cdot \sqrt{\lambda/\log n}$ is at least
\[
   \left( 1 - c_3 \cdot n^{-1} \right)^{\tau} 
  \geq  \left( 1 - c_3 \cdot n^{-1} \right)^{n} \geq c_4,
\]
where $c_4 > 0$ is another constant.

Hence with probability at least $(1 - 1/e) \cdot c_4$, we have a gap of at least $c_2 \cdot \sqrt{\lambda/\log n}$ between $Y_{(n-1)}$ and $Y_{(n)}$.
\end{proof}

We can now derive the lower bound for allocation processes with random tie-breaking.
\begin{lem}\label{pro:new_bn_lower}
Consider an allocation process with probability vector $p$ and random tie-breaking, such that $p_{n} \geq \frac{C}{n}$ for some constant $C \in (1,1.5]$ in the unweighted batched setting with $b \geq \frac{384}{(C-1)^2} n \log n$. Then, there exist constants $\kappa_1 := \kappa_1(C), \kappa_2 := \kappa_2(C) > 0$, such that
\[
\Pro{\Gap(2b) \geq \frac{C-1}{8} \cdot \frac{b}{n}} \geq \kappa_2.
\]
\end{lem}
\begin{proof}
Initially, all bins have load $0$, so the first $b$ balls will be allocated using \OneChoice. %
In order to use the Poisson Approximation Method~\cite[Theorem 5.6]{MU2017}, let $\tilde{X}_1,\tilde{X}_2, \ldots , \tilde{X}_n$ be $n$ independent Poisson distributed random variables with rate $\lambda = (b - 4 \cdot \sqrt{b})/n$. By \cref{lem:chernoff}, the sum $S_n := \sum_{i = 1}^n \tilde{X}_i$ is in the range $[b-8\sqrt{b}, b]$, with probability at least $1-o(1)$. By \cref{lem:poisson_property}, we have that with at least constant $\kappa_2 > 0$ probability, the difference between the smallest and second smallest bin is at least $\kappa_1 \cdot \sqrt{\lambda / \log n}$, for some constant $\kappa_1 > 0$.

Consider now the allocation of the remaining $b - S_n \leq 8 \sqrt{b}$ balls. The average load of a bin through these balls is $8 \sqrt{b} / n$. Using Markov's inequality, the smallest bin does not receive more than $16 \sqrt{b} / n$ additional balls with probability at least $1/2$. 

Since $\kappa_1 \cdot \sqrt{\lambda/\log n} \geq \kappa_1 \cdot \sqrt{0.5 \cdot b/n \cdot 1/\log n} \geq 16 \sqrt{b} / n$ we can conclude that there is still a unique minimally loaded bin after the allocation of all $b$ balls. Further, by using a Chernoff bound for \OneChoice, it follows that 
\[
 \Pro{ y_{n}^{b} \leq b/n - \sqrt{6 \cdot b/n \log n} } \geq 1-n^{-2}.
\]
Taking the union bound, we conclude that at the end of the first batch, the following holds:
\begin{align}
 \Pro{ y_{n}^{b} \in [- \sqrt{6 \cdot b/n \log n}, y_{n-1}^{b}-1 ] } \geq \kappa_1 \cdot \frac{1}{2} - o(1) - n^{-2}. \label{eq:lower_one}
\end{align}
Conditioning on $y_n^{b} \leq y_{n-1}^{b}-1$, we have $\tilde{p}_{n}(x^{b}) \geq p_{n} \geq \frac{C}{n}$. For simplicity, let us fix label $n$ to be the index of the bin with smallest load at time $b$. Applying \cref{pro:bn_lower} to the allocations made in the second batch to bin $n$, we conclude that there is a constant $\gamma > 0$ such that
\begin{align}
 \Pro{ x_{n}^{2b} - x_{n}^{b} \geq \left(1 + \frac{\gamma}{4} \right) \cdot \frac{b}{n} ~\Big|~y_{n}^{b} \in [- \sqrt{6 \cdot b/n \log n}, y_{n-1}^{b}-1 ] } \geq 1 - n^{-\gamma^2/8}. \label{eq:lower_two}
\end{align}
Both events in \cref{eq:lower_one} and \cref{eq:lower_two} hold with probability at least $\kappa_1 \cdot \frac{1}{3}$, and in this case,
\begin{align*}
    x_n^{2b} &= x_n^{b} + x_{n}^{2b} - x_{n}^{b} \\
    &\geq \frac{b}{n} - \sqrt{6 \cdot b/n \log n} + \left(1 + \frac{C-1}{4} \right) \cdot \frac{b}{n} \\
    &\geq \frac{2b}{n} - \sqrt{6 \cdot b/n \log n} + \frac{C-1}{4} \cdot \frac{b}{n} \\
    &\stackrel{(a)}{\geq} \frac{2b}{n} + \frac{C-1}{8} \cdot \frac{b}{n},
\end{align*}
where we have used in $(a)$ that if $b \geq \frac{384}{(C-1)^2} n \log n$ then,
\[
 \frac{C-1}{4} \cdot \frac{b}{n} \geq 2 \cdot \sqrt{6 \cdot b/n \log n} 
 \quad \Leftrightarrow 
 \quad b \geq \frac{384}{(C-1)^2} \cdot n \log n.
\]
Hence $\Gap(2b) \geq \frac{C-1}{8} \cdot \frac{b}{n}$.
\end{proof}

\section{Experiments} \label{sec:experiments}

In this section, we complement out analysis with some experiments (\cref{fig:batch_comparisons}, \cref{fig:hybrid},\cref{fig:exp_weighted} and \cref{fig:batched_with_random_tie_breaking}).

\begin{figure}[H]
    \centering
    \includegraphics[scale=0.8]{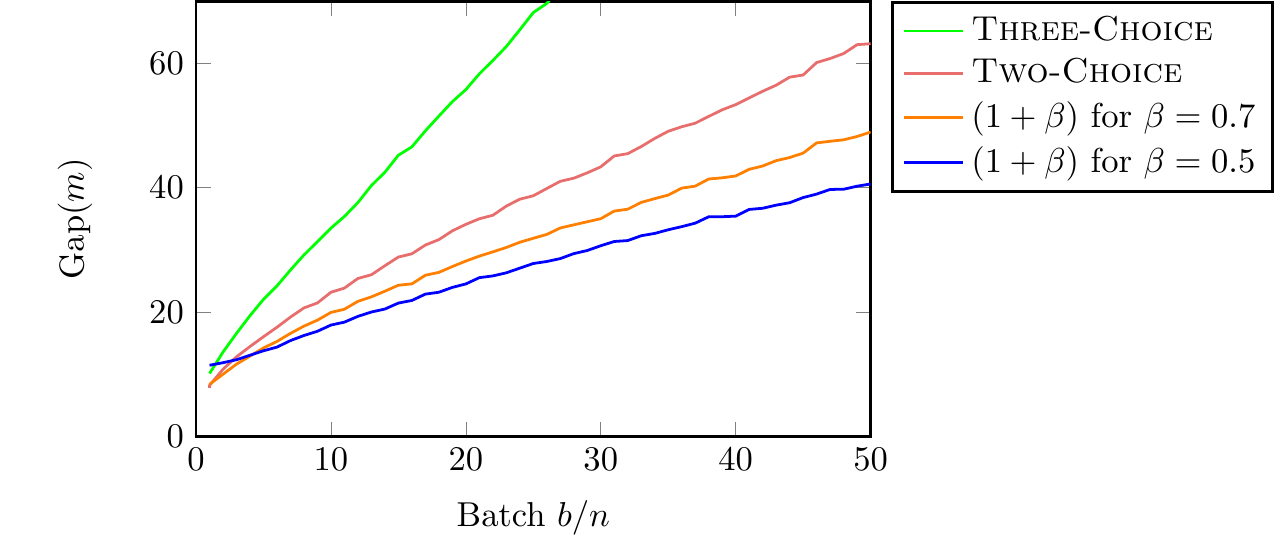}
    \caption{Comparison between $(1+\beta)$ for $\beta = 0.5$ and $\beta = 0.7$, \TwoChoice and \ThreeChoice without random tie-breaking, for the unweighted batched
    setting with $b \in \{ n, 2n , \ldots 50n \}$ for $n = 10^3$ and $m = n^2$ ($100$ runs). For small $b$, \TwoChoice
    and \ThreeChoice outperform the $(1+\beta)$ processes, which is caused by the smaller probabilities for the heavily loaded bins. Conversely, for larger $b$, the larger probabilities for the lightly loaded bins are responsible for creating larger gaps for \TwoChoice and \ThreeChoice, as suggested by  \cref{pro:bn_lower}. For \DChoice the largest term is approximately $d/n$ and for $(1+\beta)$ this is $(1+\beta)/n$, which corresponds to the observed performance: $3/n \geq 2/n \geq 1.7/n \geq 1.5/n$. Similar observations were made in the queuing setting in~\cite{M00}.}
    \label{fig:batch_comparisons}
\end{figure}

\begin{figure}[H]
    \centering
    \includegraphics[scale=0.8]{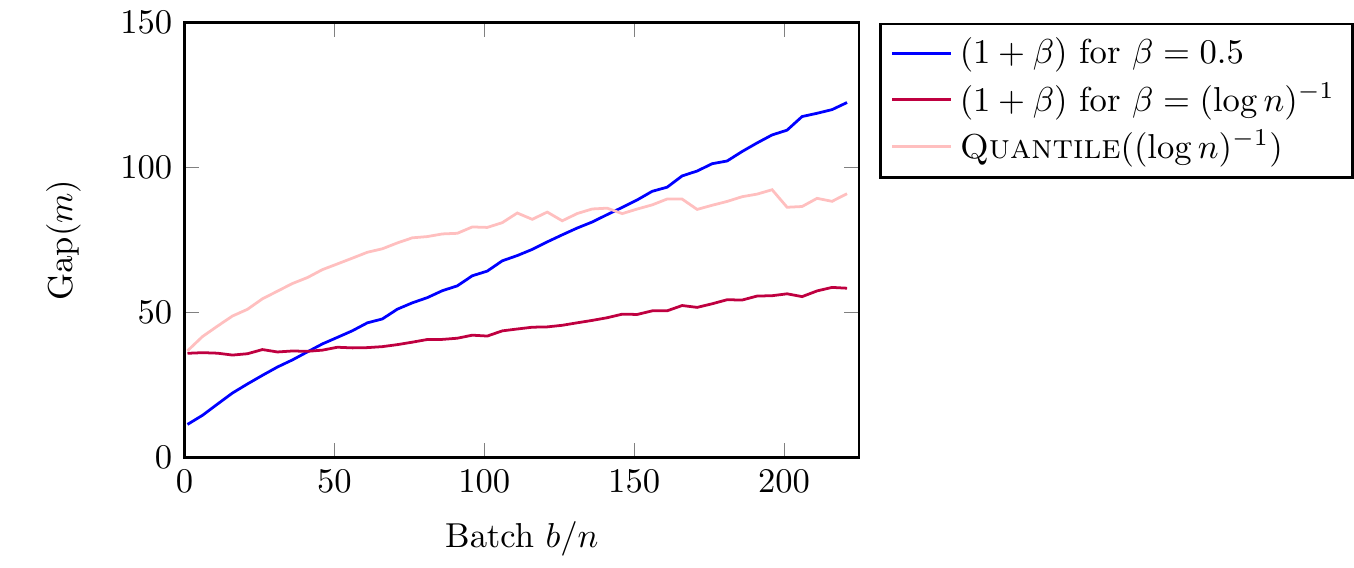}
    \caption{Empirical results for the unweighted batched setting showing that the $\textsc{Quantile}( (\log n)^{-1})$ process and $(1+\beta)$ with $\beta = (\log n)^{-1}$ achieve better gaps than $(1+\beta)$ for $\beta = 0.5$ for large values of $b \geq 150n$, $n = 10^3$ and $m = n^2$ ($100$ runs). This is probably due to the smaller maximum entry in the probability vector, which for the first two processes is $(1+(\log n)^{-1})/n$, while for the third process it is $1.5/n$.}
    \label{fig:hybrid}
\end{figure}

\begin{figure}[H]
    \centering
    \includegraphics[scale=0.8]{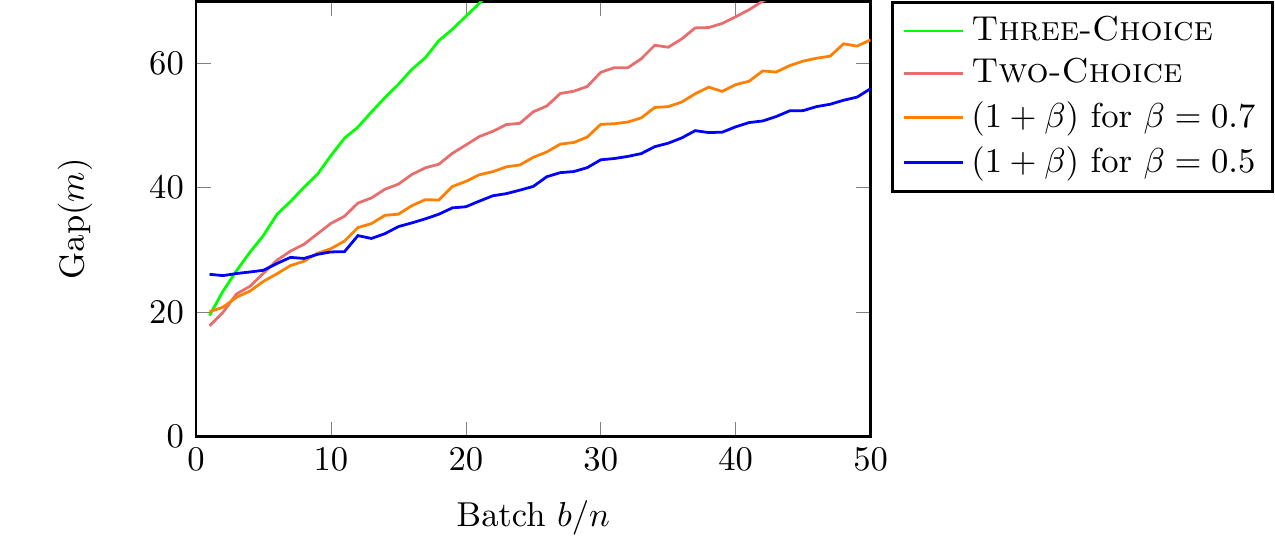}
    \caption{Empirical results for the weighted batched setting, where weights are sampled from an exponential distribution with mean $1$. Further, $n=1000$ and $m=n^2$ ($100$ runs). Overall, we seem to have a similar ordering among the four processes as in Figure~\ref{fig:batch_comparisons}, but for small values of $b$ the weights of the balls create larger gaps (in comparison to the unweighted setting). This makes sense, as some of the balls will be of weight $\Omega(\log n)$.}
    \label{fig:exp_weighted}
\end{figure}

\begin{figure}[H]
    \centering
    \includegraphics[scale=0.8]{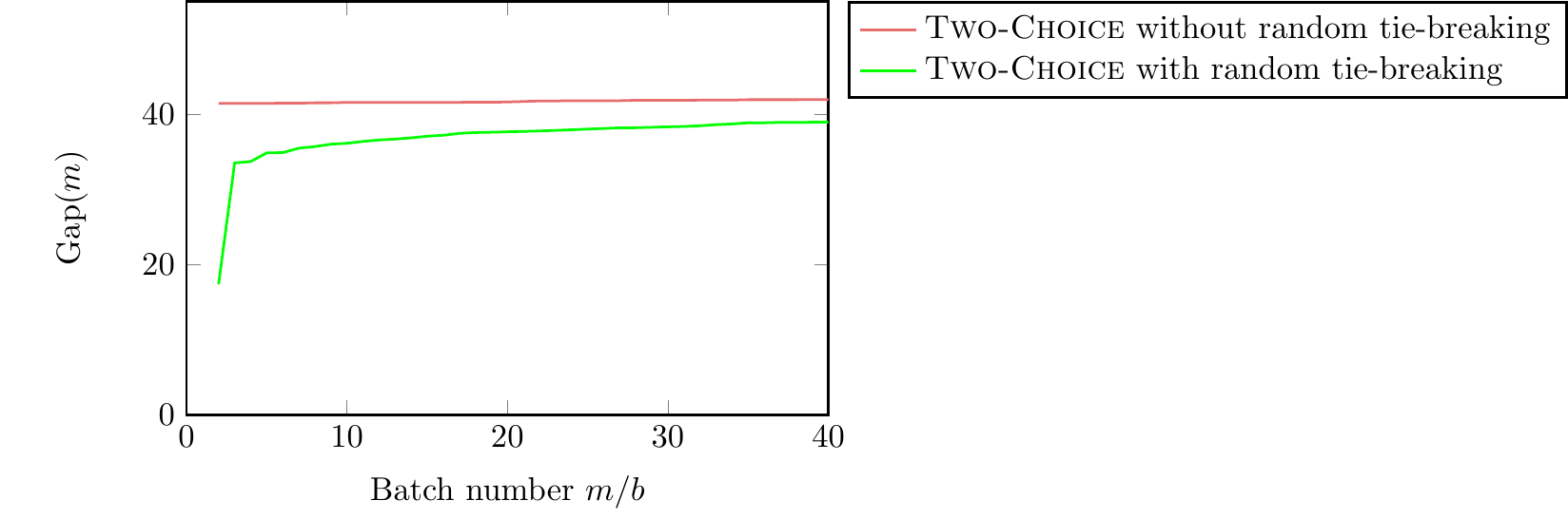}
    \caption{Empirical results for the unweighted batched setting, for the \TwoChoice process with and without random tie-breaking, for $b = 25\cdot n$, $n = 10^3$ and $100$ runs. When using random tie-breaking the maximum load is slightly smaller, especially for the first batch where it allocates using \OneChoice (so $\max_{i \in [n]} p_i = 1/n$ instead of $2/n$). However, from the second batch onwards these are asymptotically $\Omega(b/n \log n)$ for $b \geq n \log n$ (see \cref{pro:new_bn_lower}).}
    \label{fig:batched_with_random_tie_breaking}
\end{figure}

\section{Conclusions}\label{sec:conclusion}

In this work, we studied balanced allocations in a batched setting, following the model proposed in~\cite{BCE12}. As our main result, we proved that for any batch size $n \leq b \leq n^3$ and $m \geq n$, a gap bound of  $\Oh( b/n + \log n)$ holds with high probability. This analysis covered both weighted balls and a number of allocation processes satisfying two mild conditions on their probability vector, thereby demonstrating that many of the sequential allocation processes perform well in the batched setting, and can thus be ``parallelized''. We also proved lower bounds which match our upper bound up to multiplicative constants for a family of processes. 

Our results also imply a slight improvement on the gap for the $(1+\beta)$-process with very small $\beta$. Further, we proved the first gap bounds for graphical allocation with weights, thereby addressing Open Question 1 in~\cite{PTW15}.

A natural open problem is to investigate other batch sizes, e.g., $b < n$ or $b=\Omega(n^3)$, or consider a dynamic setting where the batch sizes may vary over time. Our new bounds for graphical allocation crucially depend on how much the maximum probability deviates from $1/n$, and thus on the maximum degree of the graph. Improving this dependence may lead to stronger bounds for dense graphs.

The experimental results exhibit an interesting trade-off between the probability vector and the achieved gap; having small probabilities for the heavily loaded bins is not significant for large $b$, but more important is to avoid large probabilities for the lightly loaded bins, which is achieved by processes like $\Quantile((\log n)^{-1})$ or $(1+\beta)$ for $\beta=(\log n)^{-1}$.
In other words, following a more powerful process with more choices like \DChoice leads to worse performance than a more ``indifferent'' allocation scheme like $(1+\beta)$ with $\beta=(\log n)^{-1}$.

Recall that our lower bound in \cref{pro:bn_lower} 
sheds some light onto this phenomenon, but further bounds are needed (in particular, refined upper bounds based on $\max_{i \in [n]} p_i$) so that we can rigorously compare the performance of these processes.

\bibliographystyle{plain}
\bibliography{bibliography}
\clearpage

\appendix

\section{Tools}

\subsection{Concentration inequalities}

The first lemma is a standard Chernoff bound for sum of independent random variables whose moment generating function is bounded.
\begin{lem}\label{lem:chernoff}
Assume $Z_1,Z_2,\ldots,Z_k$ are independent samples from a distribution $W$, for which there is a constant $\lambda > 0$ such that $\Ex{W}=1$ and $\ex{ e^{\lambda W} } \leq S$. Then for $Z:=\sum_{i=1}^k Z_i$, it holds for that
\[
 \Pro{ Z \geq 2 \ln(S)/\lambda \cdot k } \leq \exp\left(- \ln(S) \cdot k\right).
\]
Furthermore, for the special case $k=1$, we have for any $c > 0$,
\[
 \Pro{ Z_1 \geq 1/\lambda \cdot (c \cdot \ln(n) + \ln(S))} \leq n^{-c}.
\]
\end{lem}
\begin{proof}
Let $t \in (0,\lambda]$ to be specified later. Then,
\begin{align*}
 \Pro{ Z \geq 2 \ln(S)/\lambda \cdot k } &= \Pro{ e^{t Z} \geq e^{t \cdot 2 \ln(S) /\lambda \cdot k} } \\
 &\leq  \Ex{ e^{tZ } } \cdot \exp\left(-t \cdot 2 \log(S)/ \lambda \cdot k \right) \\
 &= \left( \Ex{ e^{ t Z_1} } \right)^k \cdot \exp\left(-t \cdot 2 \log(S) /\lambda \cdot k \right)  \\
  &\leq \left( \ex{ e^{ \lambda Z_1} } \right)^{k \cdot t/\lambda} \cdot \exp\left(-t \cdot 2 \log(S) /\lambda \cdot k \right)  \\
  &\leq S^{k \cdot t/\lambda} \cdot \exp\left(-t \cdot 2 \ln(S)/\lambda \cdot k \right) \\
  &= \exp\left( k \cdot (\ln(S) \cdot t/\lambda  - t \cdot 2 \ln(S)/ \lambda ) \right),
\end{align*}
where the second inequality is due to Jensen's inequality. Choosing $t=\lambda$ yields the claim.

For the second statement, for any $c > 0$,
\begin{align*}
    \Pro{ Z_1 \geq 1/\lambda \cdot ( c \cdot \ln(n) + \ln(S))} &\leq 
    \Pro{ e^{\lambda Z_1} \geq 
    e^{- c \cdot \ln(n)- \ln(S) }} \\
    &\leq \Ex{ e^{\lambda W}} \cdot e^{ - c \cdot \ln(n) - \ln(S)} \\
    &\leq S \cdot n^{-c} \cdot \frac{1}{S} = n^{-c}. \qedhere
\end{align*}
\end{proof}

Next we state a Chernoff bound for Poisson random variables. 

\begin{lem}[Theorem 5.4 from~\cite{MU2017}] \label{lem:chernoff_lower}
Let $X \sim \Pois(\lambda)$, then for any $0 < \delta < 1$,
\[
\Pro{X \leq (1-\delta) \cdot \lambda} \leq e^{-\lambda \delta^2/2},
\]
and 
\[
\Pro{X \geq (1+\delta) \cdot \lambda} \leq e^{-\lambda \delta^2/3}.
\]
\end{lem}

Following~\cite{K02}, we will now give the definition for \textit{strongly difference-bounded} and then give the statement for a bounded differences inequality with bad events.

\begin{defi}[Strongly difference-bounded -- Definition 1.6 in~\cite{K02}] \label{def:strongly_dif_bounded}
Let $\Omega_1, \ldots \Omega_N$ be probability spaces. Let $\Omega = \prod_{k = 1}^N \Omega_k$ and let $X$ be a random variable on $\Omega$. We say that $X$ is \textit{strongly difference-bounded by $(\eta_1, \eta_2, \xi)$} if the following holds: there is a ``bad'' subset $\mathcal{B} \subseteq \Omega$, where $\xi = \Pro{\omega \in \mathcal{B}}$. If $\omega, \omega' \in \Omega$ differ only in the $k$-th coordinate, and $\omega \notin B$, then
\[
|X(\omega) - X(\omega')| \leq \eta_2.
\]
Furthermore, for any $\omega$ and $\omega'$ differing only in the $k$-th coordinate, \[
|X(\omega) - X(\omega')| \leq \eta_1.
\]
\end{defi}

\begin{thm}[Theorem 3.3 in~\cite{K02}] \label{lem:kutlin_3_3}
Let $\Omega_1, \ldots, \Omega_N$ be probability spaces. Let $\Omega = \prod_{k = 1}^N \Omega_k$, and let $X$ be a random variable on $\Omega$ which is strongly difference-bounded by $(\eta_1, \eta_2, \xi)$. Let $\mu = \ex{X}$. Then for any $\lambda > 0$ and any $\gamma_1, \ldots , \gamma_N > 0$,
\[
\Pro{X \geq \mu + \lambda} \leq \exp\Big( - \frac{- \lambda^2}{2 \cdot \sum_{k \in [N]} (\eta_2 + \eta_1 \gamma_k)^2} \Big) + \xi \cdot \sum_{k \in [N]} \frac{1}{\gamma_k}.
\]
\end{thm}

\subsection{Auxiliary Probabilistic Claims}

We give a proof for the well-known fact that when $\Ex{e^{\lambda W}} < \infty$ then $\ex{W^4}$ is also bounded. %

\begin{lem}\label{lem:s_bound}
Consider a random variable $W$ with $\Ex{e^{\lambda W}} < \infty$ for some $\lambda > 0$. then 
\[
\Ex{W^4} < \left(\Big(\frac{8}{\lambda}\Big) \cdot \log\Big(\frac{8}{\lambda}\Big) \right)^4 +  \Ex{e^{\lambda W}}.
\]
\end{lem}
\begin{proof}
Let $\kappa := (8/\lambda) \cdot \log(8/\lambda)$. Consider $x \geq \max(0, \kappa) =: \kappa^*$. Then
\[
e^{\lambda x/4} = e^{\lambda x/8} \cdot e^{\lambda x/8} \geq e^{\log(8/\lambda)} \cdot e^{\lambda x/8} \geq \frac{8}{\lambda} \cdot \frac{\lambda x}{8} = x,
\]
using that $e^z \geq z$ for any $z$. Hence,
\[
e^{\lambda x} = (e^{\lambda x / 4})^4 \geq x^4.
\]
Hence, if $p_x$ is the pdf of $W$, then
\begin{align*}
\Ex{W^4} 
 & = \int_{x=0}^{\infty} x^4 \cdot p_x dx = \int_{x=0}^{\kappa^*} x^4 \cdot p_x dx + \int_{x=\kappa^*}^{\infty} x^4 \cdot p_x dx \\
 & \leq \int_{x=0}^{\kappa^*} \kappa^4 \cdot p_x dx + \int_{x=\kappa^*}^{\infty} e^{\lambda x} \cdot p_x dx \\
 & \leq \kappa^4 \cdot \int_{x=0}^{\infty} p_x dx + \int_{x=0}^{\infty} e^{\lambda x} \cdot p_x dx  = \kappa^4 + \Ex{e^{\lambda W}}. \qedhere
\end{align*}
\end{proof}

Next we state an inequality for a sequence of random variables, related through a recurrence inequality.

\begin{lem} \label{lem:geometric_arithmetic}
Consider a sequence of random variables $(Z_i)_{i \in \mathbb{N}}$ such that there are $0 < a < 1$ and $b > 0$ such that every $i \geq 1$,
\[
\Ex{Z_i \mid Z_{i-1}} \leq Z_{i-1} \cdot a + b.
\]
Then for every $i \geq 1$, 
\[
\Ex{Z_i \mid Z_0}
\leq Z_0 \cdot a^i + \frac{b}{1 - a}.
\]
\end{lem}
\begin{proof}
We will prove by induction that for every $i \in \mathbb{N}$, 
\[
\Ex{Z_i \mid Z_0} \leq Z_0 \cdot a^i + b \cdot \sum_{j = 0}^{i-1} a^j.
\]
For $i = 0$, $\Ex{Z_0 \mid Z_0} \leq Z_0$. Assuming the induction hypothesis holds for some $i \geq 0$, then since $a > 0$,
\begin{align*}
\Ex{Z_{i+1} \mid Z_0} & = \Ex{\Ex{Z_{i+1} \mid Z_i}\mid Z_0} \leq \Ex{Z_{i}\mid Z_0} \cdot a + b \\
 & \leq \Big(Z_0 \cdot a^i + b \cdot \sum_{j = 0}^{i-1} a^j \Big) \cdot a + b \\
 & = Z_0 \cdot a^{i+1} +b \cdot \sum_{j = 0}^i a^j.
\end{align*}
The claims follows using that for $a \in (0,1)$, $\sum_{j=0}^{\infty} a^j = \frac{1}{1-a}$.
\end{proof}

\subsection{Auxiliary Non-Probabilistic Claims}

For the next lemma, we define for two $n$-dimensional vectors $x,y$, $\langle x,y \rangle := \sum_{i=1}^n x_i \cdot y_i$.

\begin{lem}\label{lem:quasilem2}Let $(p_k)_{k=1}^n , (q_k)_{k=1}^n $ be two probability vectors and $(c_k)_{k=1}^n$ be non-negative and non-increasing. Then if $p$ majorizes $q$, i.e., for all $1 \leq k \leq n$, $\sum_{i=1}^k p_i \geq \sum_{i=1}^k q_i$ holds, then \begin{equation*} \label{eq:toprove}
\langle p,c \rangle \geq \langle q,c \rangle.
\end{equation*}
\end{lem}
\begin{proof}
We will consider a sequence of moves between $p$ and $q$, which gradually moves probability mass from lower to higher coordinates. Specifically, we define the following sequence:
\begin{align*}
r^1 &= (p_1, p_2, p_3, p_4, \ldots, p_n ) = p \\
r^2 &= (q_1, p_2 + (p_1 - q_1), p_3, p_4, \ldots,p_n ) \\
r^3 &= (q_1, q_2, p_3 + (p_1 + p_2 - q_1 - q_2), p_4, \ldots,p_n ) \\
r^4 &= (q_1, q_2, q_3, p_4 + (p_1 + p_2 + p_3 - q_1 - q_2 - q_3), \ldots,p_n ) \\
& \vdots \\
r^n &= \bigg(q_1, q_2, q_3, \ldots, q_{n-1}, p_n + \sum_{i = 1}^{n-1} (p_i - q_i) \bigg) = q,
\end{align*}
where in the last equation we used $p_n + \sum_{i = 1}^{n-1} (p_i - q_i) = p_n - p_n + q_n = q_n$.

For any $1 \leq k < n$, since $r^{k}$ and $r^{k+1}$ differ only in the $k$-th and $(k+1)$-st coordinate, and $\sum_{i=1}^{k} (p_i-q_i) \geq 0$, we conclude
it follows that
\begin{align*}
    \langle r^{k}, c \rangle - \langle r^{k+1},  c \rangle
    &\geq r_{k}^{k} c_{k} + r_{k+1}^{k} c_{k+1} - r_{k}^{k+1} c_{k} + r_{k+1}^{k+1} c_{k+1}
    \\
    &= c_k \cdot \bigg( \Big(p_k + \sum_{i = 1}^{k-1}(p_i - q_i) \Big) - q_k\bigg) 
    + c_{k+1} \cdot \bigg(p_{k+1} - \Big(p_{k+1} + \sum_{i = 1}^{k}(p_i - q_i) \Big) \bigg) \\
    & = (c_k - c_{k+1}) \cdot \sum_{i = 1}^{k}(p_i - q_i)\\
    &\geq 0.
\end{align*}
Hence $\langle p, c \rangle = \langle r^1, c \rangle \geq \langle r^2, c \rangle \geq \cdots \geq \langle r^n,c \rangle = \langle q,c \rangle$.

\end{proof}

\begin{lem} \label{lem:decreasing_fn}
The function $f(z) = z \cdot e^{k/z}$ for $k > 0$, is decreasing for $z \in (0, k]$.
\end{lem}
\begin{proof}
By differentiating,
\[
f'(z) = e^{k/z} - z \cdot e^{k/z} \cdot \frac{k}{z^2} = e^{k/z} \cdot \Big( 1 - \frac{k}{z}\Big).
\]
For $z \in (0, k]$, $f'(z) \leq 0$, so $f$ is decreasing.
\end{proof}

\end{document}